\PassOptionsToPackage{dvipsnames}{xcolor}
\documentclass[acmsmall]{acmart}

\usepackage{booktabs}   %
\usepackage{subcaption} %
\usepackage[dvipsnames]{xcolor}
\usepackage{bcprules}
\usepackage{collcell}
\usepackage{listings}
\usepackage[scaled]{beramono}
\usepackage{lipsum}
\usepackage{mathpartir}

\usepackage{mathtools}
\usepackage{mdframed}
\usepackage{amsmath}
\usepackage{stmaryrd}
\usepackage{varioref}
\usepackage{cleveref}
\usepackage{scalerel}
\usepackage{xspace}
\usepackage{calc}
\usepackage{array,varwidth}
\usepackage{makecell}
\usepackage{diagbox}
\usepackage{float}
\usepackage{xfrac}
\usepackage{enumitem}
\usepackage[skip=2.5pt]{caption}
\usepackage{slashbox}
\usepackage{xifthen}
\usepackage{tikz}
\usepackage{tikz-cd}
\usetikzlibrary{arrows}
\usetikzlibrary{fit}
\usepackage{tikz-layers}
\usepackage{cancel}
\usepackage{bm}
\usepackage[export]{adjustbox}
\usepackage{rulelinks}

\usetikzlibrary{arrows, arrows.meta, shapes, positioning}
\tikzset{nodearrow/.style={black, ->, >=myarrow},
myarrow/.tip={Latex[width=1mm, length=1mm]},
tracked/.style={draw=black, fill=blue!10},
locs/.style={draw, circle, tracked, opacity=1},
biglocs/.style={locs, minimum size = 12pt},
shared/.style={fill=yellow!20},
circular locs/.style={locs, fill=purple!10},
shared locs/.style={locs, shared},
lambda/.style={draw, cloud, text centered, cloud puffs=15, aspect=2.5},
untracked lambda/.style={lambda, fill=gray!10, dash pattern=on 5pt off 2pt},
tracked lambda/.style={lambda, tracked},
shared lambda/.style={lambda, shared},
undirected/.style={black, -, >=myarrow},
loop/.style={black, ->, >=myarrow, looseness=8},
arena/.style={draw, rectangle, thick, fill=orange!10, opacity=0.5, inner sep=0.05cm},
bigarena/.style={arena, inner sep=0.2cm},
}

\usetikzlibrary{shadows}
\usetikzlibrary{shadows.blur}
\usetikzlibrary{decorations.pathreplacing} %
\tikzcdset{arrow style=tikz,
           diagrams={>=stealth}
         }    %
\tikzcdset{crossing over clearance=.75ex}
\tikzset{hardd/.style={teal,thick}} %
\tikzset{softd/.style={teal, dashed,thick}} %
\tikzset{datad/.style={black,densely dotted}} %
\tikzset{alld/.style={black,thick}} %
\tikzset{expnode/.style={
    asymmetrical rectangle,rounded corners,draw,fill=white,inner sep=2.5pt,
    font=\footnotesize, %
  }} %
\tikzset{blknode/.style={
    asymmetrical rectangle,sharp corners,draw,violet,fill=violet!10,
    font=\footnotesize %
  }} %
\tikzset{dummynode/.style={fill=none,draw=none,asymmetrical rectangle} %
} %
\tikzset{nmute/.style={opacity=.4,blur shadow={shadow opacity=.4}}} %
\tikzset{expn/.style={
    circle,fill=white,draw,minimum size=5pt,inner sep=0pt,outer sep=0pt,
    font=\footnotesize, %
    label={[font=\footnotesize]left:#1}
  }} %
\tikzset{expn/.default={}}
\tikzset{mexpn/.style={opacity=.4,
    circle,fill=white,draw,minimum size=5pt,inner sep=0pt,outer sep=0pt,
    font=\footnotesize, %
    label={[font=\footnotesize,opacity=.4]left:#1}
  }} %
\tikzset{mexpn/.default={}}
\tikzset{rmexpn/.style={opacity=.4,
    circle,fill=white,draw,minimum size=5pt,inner sep=0pt,outer sep=0pt,
    font=\footnotesize, %
    label={[font=\footnotesize,opacity=.4]right:#1}
  }} %
\tikzset{rmexpn/.default={}}
\tikzset{mexpndummy/.style={opacity=.4,
    circle,fill=white,minimum size=5pt,inner sep=0pt,outer sep=0pt,
    label={[font=\footnotesize,opacity=.4]left:#1}
  }} %
\tikzset{mexpndummy/.default={}}

\usepackage{array,varwidth}
\usepackage{mathpartir}
\usepackage{makecell}
\usepackage{enumitem}
\usepackage[skip=2.5pt]{caption}
\usepackage{xifthen}
\usepackage{tikz}
\usetikzlibrary{arrows}

\makeatletter
\def\arcr{\@arraycr}
\makeatother

\definecolor{dark-cyan}{HTML}{135579}
\definecolor{magenta}{HTML}{a8264f}

\lstdefinelanguage{PolyRT}%
{morekeywords={abstract,%
  case,catch,char,class,function,const,as,%
  def,else,extends,final,finally,for,%
  if,import,implicit,%
  match,module,%
  new,null,%
  object,override,%
  package,private,protected,public,%
  for,public,return,super,%
  this,throw,trait,try,type,%
  val,var,%
  while,%
  yield,%
  let,end,%
	in,fun,alloc,%
  at,withRef,scoped,%
  any, typeof
  },%
  mathescape=true,%
  sensitive,%
  keywordstyle={\color{dark-cyan}\bf\ttfamily},%
  commentstyle=\color{dark-cyan},%
  escapebegin=\color{dark-cyan},%
  morecomment=[l]//,%
  morecomment=[s]{/*}{*/},%
  morecomment=[s][\color{dark-cyan}]{@}{\ },%
  morestring=[b]",%
  morestring=[b]',%
  showstringspaces=false%
}[keywords,comments,strings]%

\makeatletter
\def\ifenv#1{
   \def\@tempa{#1}%
   \ifx\@tempa\@currenvir
      \expandafter\@firstoftwo
    \else
      \expandafter\@secondoftwo
   \fi
}
\makeatother

\newcommand{\Type}[1]{\ensuremath{\ifenv{lstlisting}{\texttt{#1}}{\mathsf{#1}}}}

\newcommand{\ty}[2][]{\ensuremath{\ifthenelse{\isempty{#1}}{#2}{#2^{\,#1}}}}

\newcommand{\TVar}{\Type{Var}}

\newcommand{\TTop}{\Type{Top}}

\newcommand{\twithr}[3][]{\ifthenelse{\isempty{#1}}{\mathsf{ref}~#2~\mathsf{as}~x~\mathsf{in}~#3}{\mathsf{ref}~#2~\mathsf{as}~#1~\mathsf{in}~#3}}

\newcommand{\transform}[1]{\llbracket #1 \rrbracket}
\newcommand{\sourcecode}[1]{{\color{dark-cyan}#1}}

\newcommand{\transformsource}[1]{\transform{$ \sourcecode{\lstinline|#1|} $ }}
\newcommand{\continuation}{\mathscr{k}}
\newcommand{\continuationof}[2][]{\ifthenelse{\isempty{#1}}{\color{black}{\continuation[}#2\color{black}{]}}{\color{black}{\continuation_{#1}[}#2\color{black}{]}}}

\newcommand{\ts}[1][]{\ensuremath{\ifthenelse{\isempty{#1}}{\,\vdash\,}{\,\vdash^{\,#1}\,}}}

\newcommand{\cx}[2][]{\ensuremath{\ifthenelse{\isempty{#1}}{#2}{#2^{\,#1}}}}

\providecommand{\G}{G}
\renewcommand{\G}[1][]{\cx[#1]{\Gamma}}

\newcommand{\lnorm}[1]{\lvert #1 \rvert}
\newcommand{\qtrans}[2][]{\ensuremath{\ifthenelse{\isempty{#1}}{#2\mathord{*}}{{#2}^{#1}}}}

\newcommand{\BOX}[1]{\fbox{$\strut #1$}}

\newcounter{typerule}
\crefformat{typerule}{#2\textit{#1}#3}
\Crefformat{typerule}{#2\textit{#1}#3}

\newcommand{\CX}[3][black]{\ensuremath{{\color{#1}#2\ifthenelse{\isempty{#3}}{}{\hole{{\color{black}#3}}}}}}

\newcommand{\hole}[1]{\ensuremath{[\,#1\,]}}

\newcommand{\ie}{{\em i.e.}\xspace}

\newcommand{\Fsub}{\ensuremath{\mathsf{F}_{<:}}\xspace}

\newcommand{\eg}{{\em e.g.}\xspace}
\newcommand{\judgement}[2]{{\textsf{\textbf{#1}}} \hfill #2}

\newcommand{\DOM}{\fun{dom}}

\newcommand{\fun}[1]{\operatorname{#1}}

\usepackage{etoolbox}

\makeatletter

\newcommand{\flv@names}{}     %
\newcommand{\flv@master}{}    %
\newcommand{\flv@batch}{}     %

\newcommand{\flv@glet}[2]{%
  \expandafter\let\expandafter\flv@tmp\csname#2\endcsname
  \expandafter\global\expandafter\let\csname#1\endcsname\flv@tmp}
\newcommand{\flv@gundef}[1]{%
  \expandafter\global\expandafter\let\csname#1\endcsname\@undefined}

\edef\flv@bs{\expandafter\@gobble\string\\}

\newcommand{\flv@rec}[1]{\flv@note{#1}\flv@NC#1}
\newcommand{\flv@recstar}[1]{\flv@note{#1}\flv@NC*#1}

\newcommand{\flv@note}[1]{%
  \edef\flv@n{\expandafter\@gobble\string#1}%
  \xappto\flv@batch{\noexpand\flv@harvest{\flv@n}}%
  \xifinlist{\flv@n}{\flv@names}{}{%
    \listxadd\flv@names{\flv@n}%
    \csxappto{flv@master}{\noexpand\flv@clear{\flv@n}}%
  }%
}

\newcommand{\flv@harvest}[1]{%
  \flv@glet{\flv@cur @#1}{#1}%
  \flv@gundef{#1}%
  \ifcsundef{\flv@bs#1}{}{%
    \flv@glet{\flv@cur @\flv@bs#1}{\flv@bs#1}%
    \flv@gundef{\flv@bs#1}%
  }%
  \csxappto{flv@list@\flv@cur}{\noexpand\flv@set{#1}}%
}

\newcommand{\FlavorInput}[2]{%
  \def\flv@cur{#1}%
  \ifcsundef{flv@list@#1}{\csgdef{flv@list@#1}{}}{}%
  \global\let\flv@batch\@empty
  \let\flv@NC\newcommand
  \def\newcommand{\@ifstar{\flv@recstar}{\flv@rec}}%
  \input{#2}%
  \let\newcommand\flv@NC
  \flv@batch
}

\newcommand{\FlavorLoad}[2]{%
  \def\flv@load##1{\FlavorInput{#1}{##1}}%
  \forcsvlist\flv@load{#2}%
}

\newcommand{\flv@set}[1]{%
  \csletcs{#1}{\flv@cur @#1}%
  \ifcsundef{\flv@cur @\flv@bs#1}{}{\csletcs{\flv@bs#1}{\flv@cur @\flv@bs#1}}%
}
\newcommand{\flv@clear}[1]{%
  \csdef{#1}{\PackageError{flavor}{%
    \@backslashchar#1\space is not defined in flavor '\flv@cur'}{%
    Load a file defining it into that flavor, put it in macros.tex, or
    switch flavors around this text.}}%
}

\newcommand{\UseFlavor}[1]{%
  \ifcsundef{flv@list@#1}%
    {\PackageError{flavor}{Unknown flavor '#1'}{%
      Flavors are built with \string\FlavorLoad.}}%
    {\def\flv@cur{#1}\flv@master\csuse{flv@list@#1}}%
}

\makeatother
 \FlavorLoad{dr}{macros/domrange_macros,macros/hkdomrange_macros}
\UseFlavor{dr}

\setcopyright{cc}
\setcctype{by}
\acmDOI{10.1145/3871703}
\acmYear{2027}
\acmJournal{PACMPL}
\acmVolume{11}
\acmNumber{POPL}
\acmArticle{1}
\acmMonth{1}
\acmSubmissionID{popl27main-p3-p}
\received{2026-07-09}
\received[accepted]{2026-10-05}

\ccsdesc[500]{Software and its engineering~Functional languages}
\ccsdesc[500]{Software and its engineering~Polymorphism}
\ccsdesc[500]{Software and its engineering~Semantics}

\keywords{polymorphism, type destructors, type systems, subtyping, logical relations}

\makeatletter%
\begin{document}

\title{Let Functions Speak: Lightweight Parametric Polymorphism via Domain and Range Types}

\author{Siyuan He}
\correspondingauthor
\orcid{0009-0002-7130-5592}
\affiliation{%
  \institution{Purdue University}
  \city{West Lafayette}
  \country{USA}
}
\email{he662@purdue.edu}

\author{Songlin Jia}
\orcid{0009-0008-2526-0438}
\affiliation{%
  \institution{Purdue University}
  \city{West Lafayette}
  \country{USA}
}
\email{jia137@purdue.edu}

\author{Tiark Rompf}
\orcid{0000-0002-2068-3238}
\affiliation{%
  \institution{Purdue University}
  \city{West Lafayette}
  \country{USA}
}
\email{tiark@purdue.edu}

\begin{abstract}

Subtyping provides polymorphism at a low cost:
a higher-order function can take a parameter at the concise type $\TFun{\TInt}{\TTop}$, accepting every function on integers whatever it returns.
The parametric alternative $\forall B.\,\TFun{(\TFun{\TInt}{B})}{...}$ spends a quantifier and a type parameter to express the same, and
this cognitive overhead grows as types get larger. %
Polymorphism via subtyping, however, is not always enough, as $\TFun{\TInt}{\TTop}$ forgets the precise result type,
so a program that needs it must abandon the concise signature and return to full parameterization.
What is missing is a lightweight form of polymorphism that is as concise and readable as subtyping yet recovers the precision of parametricity.
Hofmann and Pierce's type destructors provide such polymorphism for covariant type formers but do not address contravariant positions, such as function arguments.
Projection types in practice, such as TypeScript's \texttt{Parameters<T>} and \texttt{ReturnType<T>}, rely on the unsound \texttt{any} type for abstract function types.

To close this gap, we present \domranlang, an extension of System \Fsub with first-class \emph{domain} and \emph{range} projection types,
\TDom{T} and \TRange{T},
and an application rule that types $f(x)$ from $f : F$ and $x : \TDom{F}$ for an arbitrary type $F$.
The requirement to expose $F$ as an arrow type is delayed from the definition site to the call site and discharged by the argument itself,
as a value of type $\TDom{F}$ exists only when $F$ is a subtype of an arrow type.
The type variable $F$ then needs no informative bound to represent a function type, only $<: \TTop$,
yet a function $f : F$ can still be applied without losing precision. 
We call this \emph{boundless quantification}.

We prove semantic type soundness and weak normalization by logical relations,
where path selection delays a projection's interpretation until the underlying arrow type is resolved.
The approach generalizes beyond arrow types, for example to product projections, which cooperate with domain and range in one system.
All of our development has been mechanized in Rocq.
 \end{abstract}

\maketitle

\lstMakeShortInline[keywordstyle=,%
                    flexiblecolumns=false,%
                    language=PolyRT,
                    basewidth={0.56em, 0.52em},%
                    mathescape=true,%
                    basicstyle=\footnotesize\ttfamily]@

\section{Introduction}  \label{domrange:sec:introduction}

Parametric polymorphism is pervasively seen in statically typed programming
languages, allowing function definitions to be reused for all types.
Yet, it presents challenges in type inference, and thus requires
user annotations when algorithms fall short.
In particular, for languages with higher-ranked polymorphism \`a la
System~$\mathsf{F}$, full type reconstruction is
undecidable~\cite{DBLP:journals/apal/Wells99}, and positions of type parameters
cannot be inferred in general.
Therefore, programmers must declare every type parameter explicitly.
In TypeScript, for example, a decorator that wraps a
function and logs the call must name that function's input and output types up
front, carrying the cost of declaring @A@ and @B@ in the surface syntax:
\begin{lstlisting}
  function logDecorator<A, B>(fn: (x: A) => B): (x: A) => B {  // must declare <A, B>!
    return (x: A): B => {
      console.log("calling", fn.name, x);
      return fn(x);
    };
  }
\end{lstlisting}

The cost of such explicit naming compounds quickly in \emph{higher-order} code, where functions are passed to and returned by other functions, as in wrappers, decorators, and callbacks.
A parameter to be applied in the body needs an arrow
type\footnote{Throughout, an \emph{arrow type} is of the form
$\TFun{A}{B}$, which has an \emph{arrow shape}; a
type variable has an \emph{arrow bound} when its declared bound is an
arrow type. We reserve \emph{function types} for those given to function
values, which may be type variables.} $\TFun{A}{B}$ that names its input
and output types abstractly as type parameters, exactly as @logDecorator@
names @A@ and @B@ above.
Every polymorphic function threaded through therefore demands its own pair of type
parameters, and pairs pile up.
Subscribing to an event source with a data handler and an error handler already takes four:
\begin{lstlisting}
  function subscribe<A, B, C, D>(next: (x: A) => B, error: (e: C) => D) { ... }
\end{lstlisting}
Each further handler adds another pair, and every intermediate layer that
forwards the handler repeats the pair.
Parametric function types are thus never named as a whole, but only survive as separate input and output type parameters.

Bounded quantification \`a la System \Fsub~%
\cite{DBLP:journals/iandc/CardelliMMS94} provides a nicer alternative, allowing one 
to type a function by one type variable and to expose its input and output 
types only when used.
This intuition is behind bounded generics in practical languages
including Scala and TypeScript.

\paragraph{Example: bounded quantification in TypeScript}
Type parameters can be declared with an upper bound;
\eg, @<T extends {x: number}>@ lets @T@ range over all record types
extending @{x: number}@, and @<T extends (..._: never) => unknown>@
over function types smaller than $\TFun{\TBot}{\TTop}$ (\ie, all of them).
TypeScript further offers utility types @Parameters<_>@ and @ReturnType<_>@ to
project the input and output types out of a function type.
So, one would like to keep a single @T@ and use @Parameters<T>@ and @ReturnType<T>@ in place of the pair @A@ and @B@:
\begin{lstlisting}
  function logDecorator<T extends (..._: never) => unknown>(fn: T) {
    return (...arg: Parameters<T>): ReturnType<T> => fn(...arg);
    // ${\color{red}\text{Type error}}$: Parameters<T> not assignable to never
    // ${\color{red}\text{Type error}}$: unknown not assignable to ReturnType<T>
  }
\end{lstlisting}
However, the body does not type-check. The bound @(..._: never) => unknown@
ensures that @Parameters<T>@ and @ReturnType<T>@ are well-formed, 
but it is too coarse to justify the call to @fn@, 
which only accepts @never@ and returns @unknown@.
Therefore, actual TypeScript code falls back to the unsafe bound with @any@:
\begin{lstlisting}
  function logDecorator<T extends (..._: any) => any>(fn: T) { ... }   // accepted
\end{lstlisting}
It type-checks only because the \emph{unsafe top} type @any@ bypasses every static check on @fn@ and its result, trading soundness for the clean signature.
\Cref{domrange:sec:overview-practice} further examines the practice.
The tension, however, is not specific to TypeScript, but comes from bounded quantification itself.

\paragraph{Early exposure costs precision.}
To see why, we reduce @logDecorator@ to its core: an $\eta$-expansion of a
function $f$ typed with a single type variable $F$.  To accept every function
for $f$, a sound bound of $F$ should be the most permissive
$\TFun{\TBot}{\TTop}$, the same as that in the rejected @logDecorator@ example. This
yields $\textit{eta} := \Lambda F <: (\TFun{\TBot}{\TTop}).\ \lambda f : F.\ \lambda x : \TBot.\  f(x)$,
whose body applies the parameter~$f$.
To type the application $f(x)$, $f$'s type must be exposed as an arrow type, 
and $F$'s bound at the definition site is the only known arrow type. 
Then the body types $f$ against $\TFun{\TBot}{\TTop}$:
\[
  \inferrule*[Right=\textsc{T-App}]
    { \inferrule*[Right=\textsc{T-Sub}]{ f : F \\ F <: \TFun{\TBot}{\TTop} }{ f : \TFun{\TBot}{\TTop} }
      \\ \quad x : \TBot }
    { f(x) : \TTop }
\]
However, that exposure comes too \emph{early}.
The upcast in the derivation fixes $f$'s type at the declared bound before any
instantiation arrives, whereas instantiations can vary across uses (\eg,
$\TFun{\TInt}{\TBool}$ or $\TFun{\TBool}{\TInt}$), and the sound bound
over-approximates them all.
The per-instantiation precision is lost in that upcast, and the wrapper's input and output types collapse to $\TBot$ and $\TTop$:
\[
  \textit{eta} \;:\; \TAll[F]{(\TFun{\TBot}{\TTop})}{\TFun{F}{\TFun{\TBot}{\TTop}}}
\]

This lost precision is what one would try to keep with @Parameters<T>@ and @ReturnType<T>@, but every way to recover it has a drawback.
Bounding by the unsafe @any@, as @logDecorator@ does, keeps the precision but leaves the uses of @fn@'s result unchecked.
Naming the input and output again as explicit parameters $A$ and $B$ is sound and precise, \ie, $\textit{eta} \;:\; \TAll[A,B,F]{(\TFun{A}{B})}{\TFun{F}{\TFun{A}{B}}}$,
but returns to the (even more) heavyweight decomposition that one originally meant to avoid. 
To keep a function under a single type variable in \Fsub, precision and
soundness seem mutually exclusive.

\paragraph{Type destructors.}
This tension is known, and there has already been a
solution for covariant type formers, whose guiding idea we share.
\citet{DBLP:journals/iandc/HofmannP02} introduced \emph{type destructors}, splitting the two roles of a bound: the \emph{shape} it certifies and the \emph{components} it names.
A product bound, for example $X <: \TPair{\TTop}{\TTop}$, is only a certificate that $X$ is some product,
while destructors $X.1$ and $X.2$ recover its components per instantiation, internalizing $X \equiv_\eta X.1 \times X.2$.
Their \textit{mix} example, rebuilding a pair from the components of two others, then types precisely at $\TAll[X]{(\TPair{\TTop}{\TTop})}{\TFun{X}{\TFun{X}{X}}}$, never upcasting $X$ to its bound.
However, they deliberately left out arrow types, as the contravariant domain ``raises difficult metatheoretic problems'' and they ``could not think of any useful applications''.
We regard higher-order combinators like @logDecorator@ as one such application
in the arrow-type case that has been left open.

\paragraph{Our work.}
We address the arrow-type case in the same spirit but with a different mechanism.
The contravariant domain permits dropping the shape certificate altogether,
which no covariant type former permits;
a function type variable can thus be quantified at the trivial bound $\TTop$ with no \emph{arrow obligation}.\footnote{We call \emph{arrow obligation} the demand that a type be derivably an arrow: that it be an arrow type itself or, as for a variable with an arrow bound, a subtype of a known arrow type.}
We present \domranlang, an extension of System \Fsub, with two ingredients:
\begin{itemize}[leftmargin=*, labelindent=\parindent]
  \item \textbf{Projection.} The projection types, $\TDom{F}$ and $\TRange{F}$, name the input and output of a function type $F$ as \emph{first-class} types.
  Unlike type destructors \cite{DBLP:journals/iandc/HofmannP02}, which are well-kinded only under the target shape (see \Cref{domrange:sec:casestudy-destructors}), 
  projection types are well-formed over every type,
  but are only eliminated at arrows by \mbox{$\TDom{\TFun{A}{B}} \equiv A$} and \mbox{$\TRange{\TFun{A}{B}} \equiv B$}.
  Their subtyping follows the variance of arrows, so the bound $F <: \TFun{A}{B}$ yields the upper bound $\TRange{F} <: B$ and the lower bound $A <: \TDom{F}$.
  Such congruence permits expressiveness beyond \Fsub, which has
  only upper bounds (see \Cref{domrange:sec:boundless-lowerbound}).
  \item \textbf{Delay.} Our proposed application rule \rulename{domrange:tappdomrange} types a call $f(x)$ with $f : F$ through the projections alone, without ever exposing $F$ as an arrow type.
  The arrow obligation on $F$ is instead \emph{delayed} to the call site, witnessed by the actual argument (\Cref{domrange:sec:boundless-whysound}).
  Only the contravariant domain allows this delay (\Cref{domrange:sec:boundless-covariant}).
  The definition site thus never owes the obligation, whereas a shape kind or an arrow bound would expose it too early.
\end{itemize}

{\footnotesize \setlength{\labelminsep}{1ex}
\setlength{\afterruleskip}{-1ex} \typicallabel{this}
\noindent
\begin{tabular}{@{}p{.43\linewidth}@{\hspace{.04\linewidth}}p{.53\linewidth}@{}}
\infrule[standard]{
  f: A \to B \andalso
  x: A
}{
  f(x) : B
}
&
\infrule[proposed]{
  f: F \andalso
  x: \TDom{F}
}{
  f(x): \TRange{F}
}
\\[3em]
\end{tabular}
}

The \textsc{(proposed)} rule on the right is our \rulename{domrange:tappdomrange} with typing contexts omitted, formally defined in \Cref{domrange:sec:formal}.
It types the argument and the result by the projections of $F$ and demands nothing of $F$.
The call can thus be typed soundly and precisely for any $F$, 
without either upcasting to an arrow or reconstructing through
\mbox{$\eta$-equivalence}
(\Cref{domrange:sec:overview-practice,domrange:sec:casestudy-destructors}).
When $f$'s type is an arrow type, projection elimination recovers the \textsc{(standard)} rule, so every System \Fsub program stays typeable and the extension is non-invasive.

\paragraph{Boundless quantification.}
Together, the two ingredients take over both roles of a bound: certifying the shape by delay, and naming the components by projections.
A function type variable can thus be quantified at the trivial bound $\TTop$, yet the function it types can still be applied and typed precisely in the body. We call this \textit{boundless quantification}.
The $\eta$-expansion now has the signature:
\[
  \textit{eta} \;:=\; \Lambda F <: \TTop.\, \lambda f : F.\, \lambda x : \TDom{F}.\, f(x)
  \;:\; \TAll[F]{\TTop}{\TFun{F}{\TFun{\TDom{F}}{\TRange{F}}}}
\]
While \textit{eta} can be expressed in \Fsub with separate $A$ and $B$, 
it becomes more ergonomic here, as it only needs a single type variable.
The pattern generalizes beyond function wrappers: \Cref{domrange:sec:overview-map} types a polymorphic @map@ over Church-encoded lists with the same single type parameter.

We formalize \domranlang in Rocq and prove semantic type soundness and weak normalization by logical relations.
The proof adopts a \emph{path-selection} technique that defers a projection's interpretation until the underlying arrow type is resolved, as the semantic reflection of the syntactic delay.

\paragraph{Scope.}
\domranlang provides a lightweight mechanism for declaring and applying higher-order polymorphic functions in languages with subtyping, complementing parametric polymorphism rather than replacing it.
Such boundless quantification is specific to contravariant type formers, while covariant type formers still require a bound certifying their shape.
The technique itself extends to product projections (\Cref{domrange:sec:casestudy}) and to higher-kinded types (\Cref{domrange:sec:discussion-typeof}).
Our contribution focuses on the declarative type system and its metatheory.
Algorithmic type checking and inference are beyond the scope of this work, as is integration into practical languages such as TypeScript, which faces further challenges, notably overloaded functions typed by intersections (\Cref{domrange:sec:limitations-intersection}).

\paragraph{Contributions.}
The rest of this paper is organized around our main contributions:
\begin{itemize}[leftmargin=*, labelindent=\parindent]
  \item We analyze practical uses of projection types on concrete and abstract function types, then introduce our two ingredients, projection and delay (\Cref{domrange:sec:overview}).
  \item We explain why delaying the arrow obligation to the call site is sound and unique to contravariance, and show what projection types add in expressiveness over \Fsub (\Cref{domrange:sec:boundless}).
  \item We present the formal definition of \domranlang and its relation to \Fsub
    (\Cref{domrange:sec:formal}).
  \item We develop the logical relation of \domranlang
    with path selection, which delays and resolves type projections,
    and prove semantic type soundness and weak normalization (\Cref{domrange:sec:LR}).
  \item We extend \domranlang with primitive pairs and product projection types, showing that the approach generalizes beyond arrows and enables a comparison with type destructors (\Cref{domrange:sec:casestudy}).
\end{itemize}
\Cref{domrange:sec:limitations} states the limitations, 
\Cref{domrange:sec:discussion} discusses extensions and broader connections,
\Cref{domrange:sec:related-work} surveys related work, and
\Cref{domrange:sec:conclusion} concludes.
\section{Informal Overview} \label{domrange:sec:overview}

The introduction traced the tension to bounded quantification, where applying a function-typed parameter requires its type to be exposed as an arrow type at the definition site, and any concrete bound there costs precision.
\Cref{domrange:sec:overview-practice} examines the practical usage of projection types and identifies their failure.
We then overview \domranlang with two ingredients, projection types (\Cref{domrange:sec:overview-projection}) and an application rule that delays the arrow obligation to the call site (\Cref{domrange:sec:overview-domrange-delay}).

\subsection{Projection Types in Practice} \label{domrange:sec:overview-practice}

Before introducing our calculus, we examine how projection types are used in practice.
We divide these uses by whether the projected function type is \emph{concrete} (a known arrow type) or \emph{abstract} (such as a type parameter).
The division is also reflected in existing languages: projection on concrete types is well supported, while on abstract types it is either unsupported or supported unsafely.

\paragraph{Projection on concrete types.}
On a known arrow type, a projection is type-level computation.
Scala~3 expresses it with match types, and TypeScript with the widely used utility types @Parameters<T>@ and @ReturnType<T>@.
Programmers use such projections to let types follow the program: rather than restate the input or output type of an existing function, they read it off the function's type.
A typical instance is the timer type:
\begin{lstlisting}
  type Timer = ReturnType<typeof setTimeout>;    // the platform's timer type
\end{lstlisting}
Here @typeof setTimeout@ yields the declared type of @setTimeout@, and @ReturnType<_>@ reduces by matching it against an arrow type.
@Timer@ is thus the platform's timer type, which differs between Node.js and browsers, so programmers need not write it themselves.
Formally, this is the elimination form $\TRange{\TFun{A}{B}} \equiv B$, which is sound and needs no special treatment.
In the 100 most-starred TypeScript repositories, such uses are the majority: 76\% of the @ReturnType@ and 54\% of the @Parameters@ occurrences project a concrete @typeof@ (details in the extended version~\cite{extended_version}).

\paragraph{Projection on abstract types.}
When the function type is abstract, \eg, a type variable, the projection may not be reducible at the definition site.
For instance, Scala's match types fail on abstract types. %
TypeScript is one of the few languages that support projections on type variables, but only under an arrow bound that is often with @any@.
With such a bound, @logDecorator@ in \Cref{domrange:sec:introduction} returns at type @(...args: Parameters<T>) => ReturnType<T>@, our $\TFun{\TDom{T}}{\TRange{T}}$.
We elaborate that example into a riskier variant that recurs in Rocket.Chat \cite{RocketChatRocketChat2026}:
\begin{lstlisting}
  const wrap = <F extends (...args: any[]) => any>(fn: F) =>
      ((...args: Parameters<F>) => fn(...args)) as F;    // wraps the call, casts to F
  const taggedInc = Object.assign((x: number) => x + 1,  // ((_: number) => number) &
                                  { tag: 'inc' });       //     { tag: string; }
  wrap(taggedInc).tag.toUpperCase();                     // type-checks, ${\color{red}\text{throws at runtime}}$
\end{lstlisting}
Here, @wrap@ returns a fresh function that forwards every call to @fn@, and casts it back to @fn@'s own type via @as F@, so that @wrap(fn)@ can replace @fn@ anywhere.
The cast thus reconstructs $F$ from how @fn@ is called, claiming $\TFun{\TDom{F}}{\TRange{F}} <: F$,
which can be derived from the $\eta$-equivalence that type destructors provide for products (\Cref{domrange:sec:casestudy-destructors}).
However, it is unsound once a function type describes more than its call behavior:
@taggedInc@ also carries a field @tag@ missing in the wrapper, so @wrap(taggedInc).tag@ is @undefined@, and the call of @toUpperCase@ throws.
In Rocket.Chat, 4 declarations cast a wrapper back to its function type in this way. 
Additionally, 13 sound declarations with utility types are accepted by TypeScript only through @any@, and our application rule types them soundly (details in the extended version~\cite{extended_version}).

\paragraph{Calls without exposure.}
Our calculus types the call in @wrap@ without the @any@ bound.
With \textit{projection} (\Cref{domrange:sec:overview-projection}) and \textit{delay} (\Cref{domrange:sec:overview-domrange-delay}),
the call @fn(...args)@ is typed at $\TRange{F}$ from @args@ of type $\TDom{F}$, with $F$ bounded only by $\TTop$.
The wrapper is then typed at $\TFun{\TDom{F}}{\TRange{F}}$, which describes exactly the application, while @fn@ keeps its type $F$.
Our calculus deliberately has no $\eta$-equivalence rule relating $\TFun{\TDom{F}}{\TRange{F}}$ to an abstract $F$ (\Cref{domrange:sec:casestudy-destructors}).
This loses nothing when $F$ is an arrow type, since instantiating $F := \TFun{A}{B}$ turns the wrapper's type into $\TFun{A}{B}$;
and when $F$ describes more than calls, the wrapper's type omits the rest, so @wrap(taggedInc).tag@ would be rejected statically.

\subsection{Domain and Range Types: First-Class Projections} \label{domrange:sec:overview-projection}

Our first ingredient is a pair of first-class projection types, \emph{domain} and \emph{range}.
For any type $F$, whether or not it describes a function, $\TDom{F}$ types the arguments that a value of type $F$ may accept and $\TRange{F}$ the results it may return.
When $F$ describes non-arrow values, $\TDom{F}$ is well-formed but uninhabited, \eg, $\TDom{\TInt}$.
Forming them needs no proof that $F$ has arrow shape, unlike the type destructors \cite{DBLP:journals/iandc/HofmannP02} (\Cref{domrange:sec:casestudy-destructors}), since such a proof would be the arrow bound we set out to remove.

\paragraph{Elimination on a known arrow type.}
When $F$ is known to be an arrow type, the projections eliminate to its corresponding components, as an equivalence encoded by bidirectional subtyping:
\[
  \TDom{\TFun{A}{B}} \;\equiv\; A
  \qquad\qquad
  \TRange{\TFun{A}{B}} \;\equiv\; B.
\]
When $F$ is not, or cannot upcast to an arrow type, $\TDom{F}$ and $\TRange{F}$ cannot eliminate, and $\TDom{F}$ has no inhabitants.

\paragraph{Congruence and variance.}
Like type operators, the projections have congruence in subtyping.
Their variance mirrors that of arrow types, contravariant in the domain and covariant in the range. From $A <: B$ we have:
\[
  \TDom{B} \;<:\; \TDom{A}
  \qquad\qquad
  \TRange{A} \;<:\; \TRange{B}
\]
For instance, $\TFun{\Type{Bool}}{\TInt} <: \TTop$ gives $\TDom{\TTop} <: \TDom{\TFun{\Type{Bool}}{\TInt}} \equiv \TBool$ by congruence and elimination.
This contravariant direction is the hard case flagged by \citet{DBLP:journals/iandc/HofmannP02}, especially when contravariant and covariant projections are combined.
We handle it structurally here, and in the soundness proof by tracking the polarity of accumulated projections (\Cref{domrange:sec:LR}).

\paragraph{First-class use.}
Being first-class, $\TDom{F}$ and $\TRange{F}$ appear wherever a type can, so a program can name a function's input or output type on its own while keeping the function type abstract.
For example, $\TFun{\TDom{F}}{\TBool}$ types a validator for the arguments of a function of type $F$, without exposing $F$ as an arrow type.

\subsection{Domain and Range Types: Application by Delay} \label{domrange:sec:overview-domrange-delay}

The second ingredient is the application rule built on these projections.
The standard application rule must first expose the callee's type as an arrow type to read off its input and output,
and that exposure is what costs precision as discussed before.
Instead, we type an application directly over the projections, letting the function accept an arbitrary type $F$:
\[
  \setlength{\afterruleskip}{0pt}%
  \begin{tabular}{b{.00\linewidth}@{}b{.5\linewidth}}
  &
  \infrule[]{
    f : F \qquad x : \TDom{F}
  }{
    f(x) : \TRange{F}
  }
  \\ [0.5em]
  \end{tabular}
\]
The rule is based on a simple principle: 
\emph{applying a function to a value of its domain is always valid and produces a value of its range}.
It uses the single type $F$ throughout without ever destructing it,
so at the definition site a function type variable needs no meaningful bound or potential reconstruction.
We elaborate on this \emph{boundless quantification} in
\Cref{domrange:sec:boundless}.

\paragraph{The polymorphic $\eta$-expansion.}
As the minimal example, the polymorphic $\eta$-expansion is now typed with a
single type variable $F$ bounded only by $\TTop$:
\[
  \textit{eta} \;:=\; \Lambda F <: \TTop.\, \lambda f : F.\, \lambda x : \TDom{F}.\, f(x)
  \;:\; \TAll[F]{\TTop}{\TFun{F}{\TFun{\TDom{F}}{\TRange{F}}}}
\]
Despite the trivial bound, the wrapper is fully precise at every instantiation. For example, instantiating $F$ to $\TFun{\TInt}{\TBool}$ resolves $\TDom{F}$ to $\TInt$ and $\TRange{F}$ to $\TBool$.

\paragraph{Example: a polymorphic map.} \label{domrange:sec:overview-map}
The rule is not limited to simple wrappers like \textit{eta}.
Church encodings represent data purely by abstractions and applications. 
Such data are functions that fold themselves, so applying a parameter as a function is a common pattern, and a more general application rule carries over to them.
As a more involved example, we encode @map@ over Church lists, where a polymorphic list is its own fold, $\Type{List}\,A := \TAll[R]{\TTop}{\TFun{(\TFun{A}{\TFun{R}{R}})}{\TFun{R}{R}}}$.
One type parameter then suffices for @map@ here, whereas the standard signature needs two:
{\small
\[
\begin{array}{@{}r@{\;\;}l@{}}
  \textit{map} & : \;\; \TAll[F]{\TTop}{\TFun{F}{\TFun{\Type{List}\,\TDom{F}}{\Type{List}\,\TRange{F}}}} \\[0.4ex]
               & := \; \Lambda F <: \TTop.\, \lambda f : F.\, \lambda \ell : \Type{List}\,\TDom{F}.\;
                       \Lambda R <: \TTop.\, \lambda c : \TFun{\TRange{F}}{\TFun{R}{R}}.\, \lambda n : R. \\[0.2ex]
               & \phantom{:=\;} \qquad \ell\,[R]\big(\lambda x : \TDom{F}.\, \lambda a : R.\; c(f(x))(a)\big)(n)
\end{array}
\]
}%
The body is the textbook fold, in which our rule types the application $f(x)$ from $x : \TDom{F}$ to $\TRange{F}$ under the fold's own quantifier $R$, 
while every other application stays standard since $\ell\,[R]$ and $c$ have specific arrow types.
The projections appear nested inside $\Type{List}$ as the element types and need no special treatment, since they are first-class types.
Nothing changes for callers, as instantiating $F := \TFun{A}{B}$ allows the projections to eliminate by $\Type{List}\,\TDom{\TFun{A}{B}} \equiv \Type{List}\,A$, 
so the call accepts an ordinary $\Type{List}\,A$ and returns a $\Type{List}\,B$, just as the classic two-parameter @map@ does.

\paragraph{Delaying the arrow obligation.}
The arrow obligation on $F$ is not dropped, only \emph{delayed} from the definition site to the call site, where a future argument discharges it.
A value of type $\TDom{F}$ can exist only when $F$ is a subtype of an arrow type, so the argument also serves as a \emph{witness} of the arrow shape of the function.
\Cref{domrange:sec:boundless-whysound} further explains how degenerate cases (where no real call can ever happen) work.

\paragraph{Recovering standard application.}
When $F$ is already a known arrow type $\TFun{A}{B}$, the projections can be eliminated and the standard application rule is recovered (\Cref{domrange:sec:boundless-whysound}):
\[
  \setlength{\afterruleskip}{0pt}%
  \begin{tabular}{b{.00\linewidth}@{}b{.5\linewidth}}
  &
  \infrule[]{
    f : \TFun{A}{B} \qquad x : \TDom{\TFun{A}{B}} \equiv A
  }{
    f(x) : \TRange{\TFun{A}{B}} \equiv B
  }
  \\ [0.5em]
  \end{tabular}
\]
Our proposed rule therefore only adds the ability to apply functions of abstract types.
Every \Fsub-typeable program stays typeable, and we informally conjecture that the extension is also conservative, deriving no new subtyping among projection-free types.

\paragraph{Composing with other features.}
Domain and range types are non-invasive and compose with the rest of \Fsub, and they can coexist with other projections.
We showcase this by adding projections for product types in \Cref{domrange:sec:casestudy}, 
which cooperate well with domain and range, so a program can choose which structure to expose.
The system also scales well to higher-kinded type operators and higher-order subtyping following \citet{DBLP:journals/tcs/PierceS97}, developed in the extended version \cite{extended_version} due to page limits.
\section{Boundless Quantification}  \label{domrange:sec:boundless}

Previously we introduced the ingredients for boundless quantification, typing an application $f(x)$ over an arbitrary type $F$.
Dropping the bound on $F$ raises three natural questions, which this section examines:
(1) why the delay is sound (\Cref{domrange:sec:boundless-whysound}),
(2) whether boundless quantification applies to all projections (\Cref{domrange:sec:boundless-covariant}),
and (3) whether \Fsub can already express the pattern (\Cref{domrange:sec:boundless-lowerbound}).
The discussion here is semi-formal, while \Cref{domrange:sec:LR} makes it formal.

\subsection{Soundness: Witness or Vacuous} \label{domrange:sec:boundless-whysound}

Boundless quantification may appear unsound at first.
Since the bound is only $\TTop$, nothing prevents a non-arrow instantiation, and we can form a partial application that looks alarming:
\[
  \textit{etaErr} \;:=\; \textit{eta}\,[\TInt]\,(42) \;:\; \TFun{\TDom{\TInt}}{\TRange{\TInt}}
\]
Here $\textit{etaErr}$ instantiates $\textit{eta}$ at the non-arrow type $\TInt$ and supplies $42$ as its function argument, and if it could ever be called, it would apply a number as if it were a function.
The soundness question is thus whether a well-typed call $f(x)$ can ever apply a non-function at runtime.
We answer it by asking when the call can actually happen, that is, when a function value for $f$ and an argument value for $x$ can both be supplied.
The whole reasoning is based on one fact about the function's argument, the \emph{witness property}, which we read both statically and operationally:
\begin{center}
  \emph{a value of type $\TDom{F}$ can exist only when $F$ is a subtype of an arrow type.}
\end{center}
In other words, 
the mere existence of an argument certifies that $F$ has arrow shape, so the argument witnesses exactly what the definition site delays checking.
Every instance of the rule then falls into one of two cases, and both are sound:
(1) if the call can eventually happen, it performs the standard application; and
(2) if it can never happen, it is dead code and vacuously safe.
We assume call-by-value reduction, which the delay relies on to discharge the arrow obligation, since only an argument that is already a value can witness that $\TDom{F}$ is inhabited (\Cref{domrange:sec:discussion-cbv}).

\paragraph{Reasoning from static typing.}
Our proposed rule asks the argument to be typed as $\TDom{F}$.
Since $\TDom{F}$ represents the domain of $F$, this requirement says exactly that the argument lies in the domain of the function, analogous to $x : A$ against a known arrow $\TFun{A}{B}$.
Suppose the parameter $x$ is supplied with a closed value $v$ of type $T_v$, such that $T_v <: \TDom{F}$.
Setting aside less interesting cases including reflexivity, transitivity, and the projection elimination rules, which say nothing about $F$ itself, 
three subtyping rules can end a derivation at $\TDom{F}$:
\begin{enumerate}
  \item \emph{introduction}: $A <: \TDom{\TFun{A}{B}}$, which requires the projected type to be literally an arrow;
  \item \emph{congruence}: $\TDom{S} <: \TDom{F}$ from $F <: S$, retelling the question to a supertype $S$; and
  \item \emph{bottom}: $\TBot <: \TDom{F}$, which cannot have the value witness, since no value has type $\TBot$.
\end{enumerate}
A derivation of $T_v <: \TDom{F}$ must therefore reach an introduction of some arrow, and the congruence steps along the way require $F <: \TFun{A}{B}$.
Hence if $\TDom{F}$ has any inhabitant, $F$ is a subtype of an arrow, which reflects the witness property.
Notably, the informal intuition is for illustrative purposes only, as it inverts subtyping derivations, for which neither transitivity elimination nor an inversion lemma is proved.
The witness property itself, however, does not depend on this inversion, and we derive it from a normal form of types \cite{extended_version}.

The witness property makes $F$ a subtype of an arrow, not yet an arrow itself, and the remaining gap is closed by the function value at the call site.
The degenerate $F \equiv \TBot$ is a subtype of every arrow, and its domain projection is even inhabited by every value, as $\TDom{\TFun{A}{B}} <: \TDom{\TBot}$ for every arrow, so $\TDom{\TBot} \equiv \TTop$.
However, a call still never happens at $F \equiv \TBot$, because the call also needs a function value for $f$, and no value can have type $\TBot$.
Whenever a call does happen, the argument value then makes $F$ a subtype of some arrow $\TFun{A}{B}$, and the function value $f$, having type $F <: \TFun{A}{B}$, must itself be a function.

\paragraph{Reasoning from reduction.}
The witness property also has a dynamic reading.
The wrapper $\textit{eta}$ consumes its arguments in order. 
The original term ($\textit{eta}$) is already a value as a type abstraction; once $F$ is instantiated and $f$ is supplied, 
the closure $\lambda x : \TDom{F}.\, f(x)$ is again a value; 
and the body $f(x)$ is only reduced when an argument \emph{value} of type $\TDom{F}$ arrives.
At that point $F$ is a subtype of an arrow type $\TFun{A}{B}$, so the supplied $f$ is itself a function (finally reduced to a $\lambda$), and $f(x)$ reduces as a standard application.
The application is thus safe without exposing the arrow shape at the definition site, since reduction never reaches a stuck application.

\paragraph{Combining the two readings.}
Both type and reduction reasoning confirm the two cases:
\begin{enumerate}
  \item \emph{The call can happen}: a function value and an argument value are both supplied. The argument makes $F$ a subtype of an arrow, and since a function value also inhabits $F$, $F$ is equivalent to an arrow, say $\TFun{A}{B}$. The call is then $f : \TFun{A}{B}$ applied to $x : \TDom{F} \equiv A$ with result $\TRange{F} \equiv B$, recovering standard application exactly.
  \item \emph{The call can never happen}: one of the two values can never be supplied. When $F$ has no arrow supertype (\eg, $\TInt$ or $\TTop$), the domain $\TDom{F}$ is uninhabited and no argument ever arrives, so a partial application such as $\textit{etaErr}$ is a well-typed value whose body is unreachable, namely a function with an empty domain. When $F \equiv \TBot$, no satisfying function value exists. Either way the application is dead code, and no reduction ever applies a non-function.
\end{enumerate}
The degenerate $\TBot$ case is also supported by the semantics.
In the logical relation in \Cref{domrange:sec:LR}, a projection reached along a \emph{negative} path over $\TBot$ is interpreted as the set of all values, and safety follows from the emptiness of $\TBot$ itself.

In short, a function abstraction describes a future application, not an immediate one.
So the witness that $F$ has arrow shape is never required at the definition site, 
and can arrive with the argument at the call site, 
exactly when reduction is about to occur.
TypeScript follows a similar call-site strategy, inspecting the function once you have it and then reading off its projections, but needs the unsafe @any@ to do so.
Our rule captures the strategy soundly, with no arrow bound and no auxiliary mechanism to carry structural information into the body.

\subsection{Delay Is Unique to Contravariance} \label{domrange:sec:boundless-covariant}

The reasoning above relied on the witness arriving from \emph{outside} of the abstraction, which is special to functions.
It is natural to ask whether the same delay applies to covariant projections, such as the components $\TFst{P}$ and $\TSnd{P}$ of a product, which would extend the boundless style beyond functions.
The answer is no, and the reason shows what is special about arrows.

The difference comes from the variance of the projections.
By contravariance, $F <: \TFun{A}{B}$ gives $A \equiv \TDom{\TFun{A}{B}} <: \TDom{F}$, so an argument of type $\TDom{F}$ shows that $F$ is a subtype of an arrow type, and a value of type $F$ must then be a function.
By covariance, $\TPair{A}{B} <: P$ gives $A \equiv \TFst{\TPair{A}{B}} <: \TFst{P}$, so a value of type $\TFst{P}$ only shows that $P$ is a supertype of some product type, and $\TFst{\TTop}$ is inhabited by every value.
No value can thus witness that $P$ is a product type, so for a projection $p.1$ with $p : P$, the obligation that $P$ is a product must travel into the body, carried by a bound $P <: \TPair{\TTop}{\TTop}$, which is the shape kinding of \citet{DBLP:journals/iandc/HofmannP02}.

The two mechanisms are thus complementary, each matching the variance of the type it projects.
A covariant projection carries its shape obligation \emph{inward} through a bound, as kinding does. 
A contravariant projection can instead defer it \emph{outward} to the call site, because the value that triggers reduction is also the one that witnesses the shape.
Exploiting this asymmetry is what lets \domranlang drop the arrow bound entirely for function types.
\Cref{domrange:sec:casestudy} adds product projections alongside domain and range, showing the two kinds of projections cooperating in one system.

\subsection{Expressiveness: Lower Bounds from Contravariance} \label{domrange:sec:boundless-lowerbound}

The delay mechanism discharges the shape of $F$, while with a concrete arrow bound, the projections must remain valid under subtyping.
It is then natural to ask whether \Fsub can already express projection types or at least boundless quantification.
The answer is no for projection types in general, but yes for boundless quantification.

Projection types must agree with the variance of the function type.
Namely, an arrow bound $F <: \TFun{A}{B}$ gives opposite restrictions, $\TRange{F} <: B$ by covariance and $A <: \TDom{F}$ by contravariance.
The covariant bound is harmless, but the contravariant one is a lower bound, while \Fsub provides only upper bounds.
This lets us construct an example that \Fsub cannot type:
\[
  \begin{array}{r@{\;}l}
    g & := \Lambda F <: \TFun{\TBool}{\TBool}.\, \lambda f : F.\, \lambda h : \TFun{\TRange{F}}{\TDom{F}}.\, f(h(f(\mathsf{true}))) \\
      & : \TAll[F]{\TFun{\TBool}{\TBool}}{\TFun{F}{\TFun{(\TFun{\TRange{F}}{\TDom{F}})}{\TRange{F}}}}
  \end{array}
\]
Here $g$ applies $f$ twice around an adapter $h$ of type $\TFun{\TRange{F}}{\TDom{F}}$, which flips the input and output of $F$.
The inner call $f(\mathsf{true})$ needs $\TBool <: \TDom{F}$ to validate the application.
The type of $h$ follows each instantiation of $F$. For instance at $F := \TFun{\TTop}{\TBool}$, $h$ has the flipped type $\TFun{\TBool}{\TTop}$ and admits functions that return more than booleans.
Without domain and range types to express the flip, \Fsub can only attempt $g$ by expanding $F$ into two type variables:
\[
  \begin{array}{r@{\;}l}
    g' & := \Lambda A <: \TTop.\, \Lambda B <: \TBool.\, \lambda f : \TFun{A}{B}.\, \lambda h : \TFun{B}{A}.\, f(h(f({\color{red}\mathsf{true}}))) \\
       & : \TAll[A]{\TTop}{\TAll[B]{\TBool}{\TFun{(\TFun{A}{B})}{\TFun{(\TFun{B}{A})}{B}}}}
  \end{array}
\]
The inner call $f(\mathsf{true})$ now requires $\mathsf{true} : A$, that is, $\TBool <: A$.
But \Fsub gives type variables only upper bounds, so $\TBool <: A$ is underivable.
In $g$ with projection types, the same requirement $\TBool <: \TDom{F}$ is met because the bound $F <: \TFun{\TBool}{\TBool}$ yields $\TBool \equiv \TDom{\TFun{\TBool}{\TBool}} <: \TDom{F}$.
The only way out in \Fsub is to fix $A := \TBool$, losing the polymorphism on the domain, since $h$ must then return booleans even when $f$ accepts more, as at $F := \TFun{\TTop}{\TBool}$.

\paragraph{\Fsub is not enough to express projection types.}
The example shows that \Fsub alone is not expressive enough for projection types, as destructing each function type into two type parameters does not suffice.
In general, the domain projection turns an upper bound on a function type into a lower bound on its domain, which \Fsub cannot state.
We expect that a system with lower bounds on type variables, such as Scala's interval bounds, would close the gap by declaring $A >: \TBool$ and $B <: \TBool$ to match the variance.
\domranlang needs no such declaration, since the projections inherit the variance of the function type.
Nested projections work as well, as $F <: \TFun{(\TFun{S_1}{S_2})}{T}$ gives $\TDom{\TDom{F}} <: S_1$, the same polarity that selection paths (\Cref{domrange:sec:LR}) semantically track.

\paragraph{Boundless quantification needs no lower bounds.}
Although \Fsub cannot encode projection types in general, boundless quantification via projection types can still be implemented in \Fsub without lower bounds.
Boundless quantification, $F <: \TTop$, is meaningless and yields only the trivial bounds $\TBot <: \TDom{F}$ and $\TRange{F} <: \TTop$, so no lower-bound reasoning is needed.
The other requirement, that $F$ be an arrow $\TFun{A}{B}$, needs no bound either, as the witness of \Cref{domrange:sec:boundless-whysound} discharges it at the call site.
Moreover, the body of a boundless quantification never has to show a type below $\TDom{F}$, 
since projections always stand for the most precise type.
Intuitively, this is why boundless programs such as \textit{eta} survive the expansion of $F$ into two unbounded variables $A$ and $B$.
The gain of boundless quantification over \Fsub is thus ergonomic, replacing $A$ and $B$ with a single variable $F$,
apart from dead code such as \textit{etaErr} that only \domranlang types.
Bounded programs such as $g$ gain more, since they rely on lower-bound reasoning.

\paragraph{Projection and delay together.}
To conclude, boundless quantification relies on two ingredients, 
which together take over the two roles of the bound at the definition site.
The delayed application rule discharges the arrow obligation on $F$ at the call site with the argument as a witness (\Cref{domrange:sec:boundless-whysound,domrange:sec:boundless-covariant}).
The projection types name the components, whose bounds become trivial under $F <: \TTop$, so no lower-bound reasoning arises.
Together they let an abstract type variable $F$ bounded only by $\TTop$ stand for a function at the definition site.

\section{Formal Presentation} \label{domrange:sec:formal}

We present the \domranlang-calculus, an extension of the \Fsub-calculus with domain and range types.
The \Fsub fragment is standard, so we keep it brief and focus on projections and the new application rule.

\subsection{Syntax} \label{domrange:sec:formal-syntax}

\begin{figure}[t]\small
\begin{mdframed}
\begin{minipage}[t]{1.0\textwidth}
\judgement{Syntax}{\BOX{\domranlang}}\small %
  \[\begin{array}{l@{\qquad}l@{\qquad}l@{\qquad}l}
    x,y,z   & \in & \TVar        & \text{Variables}               \\
    f,g,h   & \in & \TVar       & \text{Function Variables}      \\
    X       & \in & \TVar       & \text{Type Variables} \\
    t       & ::= & c \mid x \mid \lambda x{:}T.\,t \mid t~t \mid \TLam{X <: T}{t} \mid \TApp{t}{T}    &  \text{Terms}           \\ [2ex]
    S,T,U & ::= & B \mid \TFun{T}{T} \mid \TAll[X]{T}{T} \mid \TBot  & \\
            &     & \mid \TTop \mid  \TDom{T} \mid \TRange{T}    & \text{Types}                   \\[2ex]

    \Gamma  & ::= & \varnothing\mid \ctxext{\Gamma}{x : T} \mid \ctxext{\Gamma}{X <: T}             & \text{Typing Environment}     \\[2ex]
    \end{array}\]

\end{minipage} %

\caption{The syntax of \domranlang calculus.} \label{domrange:fig:formal-syntax}
\end{mdframed}
\vspace{-2ex}
\end{figure} 
\Cref{domrange:fig:formal-syntax} defines the syntax of \domranlang.
Terms are standard \Fsub terms, including value constants $c$, term variables $x$, $\lambda$-abstractions, applications, type $\Lambda$-abstractions, and type instantiations.
Types include standard base types $B$, arrow types, bounded quantification, the bottom type $\TBot$, and the top type $\TTop$, 
and are extended with the domain and range projections $\TDom{T}$ and $\TRange{T}$ formed over arbitrary types.
The $\Gamma$ context mixes term and type variable bindings, following \citet{tapl_book}; our mechanization instead keeps type variables and their bounds in a separate context, which is an equivalent presentation and better fits the higher-kinded system (\Cref{domrange:sec:discussion-typeof}).

\subsection{Static Subtyping} \label{domrange:sec:formal-subtyping}

\begin{figure}[t]
\setlength{\afterruleskip}{\bigskipamount}
\small
\begin{mdframed}
\judgement{Subtyping}{\BOX{\strut \stp{\G}{T}{T} }}\\
\typicallabel{strans}
\begin{tabular}{b{.48\linewidth}@{}b{.5\linewidth}}
  \infrule[\ruledef{domrange:stop}{s-top}]{
  }{
    \stp{\G}{T}{\TTop}
  }
  &
  \infrule[\ruledef{domrange:sbot}{s-bot}]{
  }{
    \stp{\G}{\TBot}{T}
  }
  \\ [0.5em]
\end{tabular}
\begin{tabular}{b{.29\linewidth}@{}b{.29\linewidth}@{}b{.4\linewidth}}
  \infrule[\ruledef{domrange:srefl}{s-refl}]{
  }{
    \stp{\G}{T}{T}
  }
  &
  \infrule[\ruledef{domrange:svar}{s-var}]{
    X <: T \in \G
  }{
    \stp{\G}{X}{T}
  }  
  &
  \infrule[\ruledef{domrange:strans}{s-trans}]{
    \stp{\G}{T_1}{T_2} \quad 
    \stp{\G}{T_2}{T_3}
  }{
    \stp{\G}{T_1}{T_3}
  }
\\ [0.5em]
\end{tabular}
\begin{tabular}{b{.48\linewidth}@{}b{.5\linewidth}}
  \infrule[\ruledef{domrange:sfun}{s-fun}]{
    \stp{\G}{T_1}{S_1} \quad
    \stp{\G}{S_2}{T_2}
  }{
    \stp{\G}{\TFun{S_1}{S_2}}{\TFun{T_1}{T_2}}
  }
  &
  \infrule[\ruledef{domrange:sall}{s-all}]{
    \stp{\ctxext{\G}{X <: U_2}}{S}{T} \quad
    \stp{\G}{U_2}{U_1}
  }{
    \stp{\G}{\TAll{U_1}{S}}{\TAll{U_2}{T}}
  } 
  \\ [0.5em]
\end{tabular}
\begin{tabular}{b{.48\linewidth}@{}b{.5\linewidth}}
  \infrule[\ruledef{domrange:sdomintro}{s-dom-intro}]{
  }{
    \stp{\G}{T_1}{\TDom{\TFun{T_1}{T_2}}}
  }
  &
  \infrule[\ruledef{domrange:srangeelim}{s-range-elim}]{
  }{
    \stp{\G}{\TRange{\TFun{T_1}{T_2}}}{T_2}
  }
  \\ [0.5em]
  \infrule[\ruledef{domrange:sdomelim}{s-dom-elim}]{
  }{
    \stp{\G}{\TDom{\TFun{T_1}{T_2}}}{T_1}
  }
  &
  \infrule[\ruledef{domrange:srangeintro}{s-range-intro}]{
  }{
    \stp{\G}{T_2}{\TRange{\TFun{T_1}{T_2}}}
  }
  \\ [0.5em]
  \infrule[\ruledef{domrange:sdom}{s-dom-congr}]{
    \stp{\G}{T}{S}
  }{
    \stp{\G}{\TDom{S}}{\TDom{T}}
  }
  &
  \infrule[\ruledef{domrange:srange}{s-range-congr}]{
    \stp{\G}{S}{T}
  }{
    \stp{\G}{\TRange{S}}{\TRange{T}}
  }
\end{tabular} \\[1.5em]
\caption{Subtyping rules of \domranlang calculus.} \label{domrange:fig:formal-subtyping}
\end{mdframed}

\vspace{-2ex}
\end{figure} 
\Cref{domrange:fig:formal-subtyping} presents the declarative subtyping rules, with explicit reflexivity and transitivity rules.
The \Fsub fragment is the full variant\footnote{Kernel \Fsub also suffices once \rulename{domrange:ttapp} regains its standard premise $S <: U$. We present the full variant for expressiveness as a declarative system. Boundless quantification does not rely on bounds for type variables as previously discussed.}, whose rule \rulename{domrange:sall} compares bounds contravariantly.
The remaining rules describe the two projections.

For domain types, \rulename{domrange:sdomintro} is the introduction rule: any type $T_1$ upcasts to $\TDom{\TFun{T_1}{T_2}}$, since a function of type $\TFun{T_1}{T_2}$ accepts every argument of type $T_1$ and its domain is therefore at least $T_1$.
Paired with the elimination rule \rulename{domrange:sdomelim}, the projection is equivalent to $T_1$ on a known arrow.
The congruence rule \rulename{domrange:sdom} is contravariant, mirroring the variance of arrow types in \rulename{domrange:sfun}, so upcasting a function type shrinks its domain.

Range types are dual.
The elimination rule \rulename{domrange:srangeelim} upcasts $\TRange{\TFun{T_1}{T_2}}$ to the result type $T_2$, since such a function returns results of type at most $T_2$.
Paired with \rulename{domrange:srangeintro}, the projection is equivalent to $T_2$ on a known arrow, and the congruence rule \rulename{domrange:srange} is covariant, again mirroring \rulename{domrange:sfun}.

Notably, the arrows appearing in \rulename{domrange:sdomintro} and \rulename{domrange:srangeelim} may be uninhabited.
For example, $\TTop <: \TDom{\TFun{\TTop}{\TBot}}$ is permitted, where no function has the arrow type $\TFun{\TTop}{\TBot}$, which serves only as a type-level witness.
The rules hold for every arrow, inhabited or not, and their soundness follows from \Cref{domrange:sec:LR}.
This primitive treatment, with no well-formedness obligation on the projected type, is a deliberate design choice; 
we contrast it with the type destructors of \citet{DBLP:journals/iandc/HofmannP02} in \Cref{domrange:sec:casestudy-destructors}.

\subsection{Term Typing}

\Cref{domrange:fig:formal-typing} presents the typing rules, standard for \Fsub except for application.

\begin{figure}[t]
\setlength{\afterruleskip}{\bigskipamount}
\small
\begin{mdframed}
\judgement{Typing}{\BOX{\strut \hastype{\G}{t}{T} }}\\
\typicallabel{tsub}
\begin{tabular}{b{.4\linewidth}@{}b{.58\linewidth}}
  \infrule[\ruledef{domrange:tcst}{t-cst}]{
  }{
    \hastype{\G}{c}{B}
  }
  &
  \infrule[\ruledef{domrange:tsub}{t-sub}]{
    \stp{\G}{T_1}{T_2} \qquad
    \hastype{\G}{t}{T_1}
  }{
    \hastype{\G}{t}{T_2}
  }
  \\[0.5em] %
  \infrule[\ruledef{domrange:tvar}{t-var}]{
    x : T \in \G   
  }{
    \hastype{\G}{x}{T}
  }
  &
  \infrule[\ruledef{domrange:tappdomrange}{t-app-dr}]{
    \hastype{\G}{t_1}{F} \qquad
    \hastype{\G}{t_2}{\TDom{F}}
  }{
    \hastype{\G}{t_1 ~ t_2}{\TRange{F}}
  } \\ [0.5em]
\end{tabular}
\begin{tabular}{b{.45\linewidth}@{}b{.53\linewidth}}
  \infrule[\ruledef{domrange:tabs}{t-abs}]{
    \hastype{\ctxext{\G}{x : T_1}}{t}{T_2}
  }{
    \hastype{\G}{\lambda x{:}T_1.\, t}{\TFun{T_1}{T_2}}
  }
  &
  \color{gray}{\infrule[\ruledef{domrange:tapp}{t-app}]{
    \hastype{\G}{t_1}{\TFun{T_1}{T_2}} \qquad
    \hastype{\G}{t_2}{T_1}
  }{
    \hastype{\G}{t_1 ~ t_2}{T_2}
  }}
  \\ [0.5em]
  \infrule[\ruledef{domrange:ttabs}{t-tabs}]{
    \hastype{\ctxext{\G}{X <: U}}{t}{T}
  }{
    \hastype{\G}{ \TLam{X <: U}{t}}{\TAll{U}{T}}
  }
  &
  \infrule[\ruledef{domrange:ttapp}{t-tapp}]{
    \hastype{\G}{t}{\TAll{U}{T}}
  }{
    \hastype{\G}{\TApp{t}{U}}{\substty[X]{U}{T}}
  }
\end{tabular} \\[1.5em]
\caption{Typing rules of \domranlang calculus. The standard application rule \rulefmt{t-app} (in {\color{gray} gray}) is implied by the domain and range application rule \rulefmt{t-app-dr} via subtyping on both domain and range types.
} \label{domrange:fig:formal-typing}
\end{mdframed}

\vspace{-2ex}
\end{figure} 
The new rule \rulename{domrange:tappdomrange} types an application without exposing the function type as an arrow.
To apply $f : F$, the argument must have the domain type $\TDom{F}$, and the result has type $\TRange{F}$, so the application is typed while $F$ stays fully abstract.
Formally, soundness follows from the semantic well-formed functionality (\Cref{domrange:sec:LR-invariants}), which is not exposed in the declarative typing.

The standard rule \rulename{domrange:tapp} is derivable and shown in gray.
Given $f : \TFun{A}{B}$ and an argument of type $A$, rule \rulename{domrange:sdomintro} upcasts the argument to $\TDom{\TFun{A}{B}}$, the call types at $\TRange{\TFun{A}{B}}$ by \rulename{domrange:tappdomrange}, and \rulename{domrange:srangeelim} upcasts the result to $B$.
Type application \rulename{domrange:ttapp} instantiates a quantifier at its declared bound.
Instantiation at any $S <: U$ is recovered by lowering the bound through the subtyping rule \rulename{domrange:sall}, from $\TAll{U}{T}$ to $\TAll{S}{T}$, and then instantiating at $S$.

Hence, every \Fsub-typeable program stays typeable in \domranlang.
We conjecture that the extension is also conservative, meaning that it derives no new subtyping between projection-free types.
Proving it would require eliminating transitivity in the extended subtyping relation, which is left open.
\section{Logical Relations} \label{domrange:sec:LR}

In this section, we prove weak normalization and semantic type soundness for \domranlang by logical relations \cite{LogicalTypeSoundness}.
Each type $T$ receives a \emph{value interpretation}, the set of well-behaved values inhabiting it, 
and the fundamental theorem (\Cref{domrange:thm:fundamental}) states that every well-typed term semantically inhabits the interpretation of its type.
Notably, our static type system is declarative.

Our development extends the big-step definitional interpreters for \Fsub of \citet{DBLP:conf/popl/AminR17}.
We introduce a \emph{path selection} technique for domain and range types, whose interpretations cannot be fixed until the underlying type resolves to an arrow.
The technique is inspired by \citet{DBLP:conf/ecoop/WangR17}, 
who use bound selectors to reason about the bounds of abstract type members.

\subsection{Syntax, Reduction, and Path Selection} \label{domrange:sec:LR-syntax}

\Cref{domrange:fig:LR-syntax} presents the syntax and notation of our logical relation.

\begin{figure}[t]\small
\begin{mdframed}
\begin{minipage}[t]{1.0\textwidth}
\judgement{Syntax of Logical Relations}{\BOX{\domranlang}}\small %
  \[\begin{array}{l@{\qquad}l@{\qquad}l@{\qquad}l}
    v       & ::= & c \mid \vabs{H}{t} \mid \vtabs{H}{t}   & \text{Values}      \\ 
    \mathbb{V} & ::= & \{ \valsl{v}{s}  \}   & \text{Value Type} \\
    V       & \in & \mathbb{V}  & \text{Value Type Variables}    \\ [2ex] 

    H       & ::= & \varnothing \mid H,\, (x,\,v)   & \text{Value Environment} \\
    \rho & ::= & \varnothing \mid \rho,\, (X,\, V)   & \text{Value Type Environment} \\[2ex]
  
  \end{array}\]    
    
\judgement{\pathselection{Selection Path}}{}\small %
  \[\begin{array}{l@{\qquad}l@{\qquad}l@{\qquad}l@{\qquad}l}
    \pathselection{sl}      & ::= &  \pathselection{\sldom} \mid \pathselection{\slrange}  &  &  \text{Selector} \\
    \pathselection{s}     & ::= & \slroot \mid \pathselection{\slappend{sl}{s}} &  & \text{Path} \\[2ex]
    \withpath{V}{s}        & ::= & \valtypeset[v]{ \valsl{v}{s} \in V } &  & \text{Value Type with Path}  \\[2ex]
  \end{array}\]

\end{minipage} 

\caption{The syntax and path selection of the logical relations for \domranlang calculus.} \label{domrange:fig:LR-syntax}
\end{mdframed}
\vspace{-2ex}
\end{figure} 
Values are constants, $\lambda$-abstractions, and type $\Lambda$-abstractions.
The two abstractions are closures, pairing the binder with the value environment $H$ captured when they were reduced, where $H$ is a partial map from variables to values.
Reduction is big-step, written $H,\, t \Downarrow v$ and defined in \Cref{domrange:fig:LR-bigstep}.

\begin{figure}[t]
\setlength{\afterruleskip}{\bigskipamount}
\small
\begin{mdframed}
\judgement{Big-Step Reduction}{\BOX{\strut H,\,t \Downarrow v }}\\
\typicallabel{ecst}
\begin{tabular}{b{.25\linewidth}@{}b{.73\linewidth}}
  \infrule[\ruledef{domrange:ecst}{e-cst}]{
  }{
    H,\,c \Downarrow c
  }
  &
  \infrule[\ruledef{domrange:eapp}{e-app}]{
    H,\, t_1 \Downarrow \vabs[x]{H'}{t'} \quad 
    H,\, t_2 \Downarrow v_x \quad
    (H',\,(x,v_x)),\, t' \Downarrow v
  }{
    H,\,t_1 ~ t_2 \Downarrow v
  }
\\ [0.5em]
\end{tabular}
\begin{tabular}{b{.25\linewidth}@{}b{.73\linewidth}}
  \infrule[\ruledef{domrange:evar}{e-var}]{
    H(x) = v
  }{
    H,\,x \Downarrow v
  }  
  &
  \infrule[\ruledef{domrange:etapp}{e-tapp}]{
    H,\,t \Downarrow \vtabs[X]{H'}{t'} \qquad
    H',\, t' \Downarrow v 
  }{
    H,\,\TApp{t}{T} \Downarrow v 
  } 
\\ [0.5em]
\end{tabular}
\begin{tabular}{b{.44\linewidth}@{}b{.50\linewidth}}
  \infrule[\ruledef{domrange:eabs}{e-abs}]{
  }{
    H,\, \lambda x{:}T.\,t \Downarrow \vabs[x]{H}{t}
  }
  &
  \infrule[\ruledef{domrange:etabs}{e-tabs}]{
  }{
    H,\, \TLam{X <: T}{t} \Downarrow \vtabs[X]{H}{t}
  }
\end{tabular}
\\[1.5em]
\caption{Big-step Reduction of \domranlang-calculus.} \label{domrange:fig:LR-bigstep}
\end{mdframed}

\vspace{-2ex}
\end{figure} 
\paragraph{Path Selection.}
Domain and range types cannot be interpreted immediately, 
since their meaning depends on whether the underlying type eventually resolves to an arrow.
We therefore delay their interpretation.
A \emph{path selector}, \pathselection{\sldom} or \pathselection{\slrange}, records one pending projection, and a \emph{selection path} \pathselection{s} recording all pending projections is a sequence of selectors appended to the root \slroot.
Since a type has a different semantic meaning under different pending projections, value types $\mathbb{V}$ are values \emph{paired} with selection paths, $\valsl{v}{s}$, rather than values alone.
The notation $\withpath{V}{s}$ fixes a value type $V$ with path $\pathselection{s}  $, yielding a plain set of values as in a traditional logical relation.
We write $V$ for value types and $\rho$ for value type environments, partial maps from type variables to value types.
At the root path having no pending projection, the relation degenerates to the standard interpretation of \Fsub.

\subsection{Semantic Interpretations} \label{domrange:sec:LR-valtype}

\Cref{domrange:fig:LR-valtype} defines the value interpretations and the relation between them under selection paths.

\subsubsection{Variance of Selection Paths.}
Selection paths carry variance.
For a subtyping $\TFun{T_1}{T_2} <: \TFun{S_1}{S_2}$, the interpretation of the arrow itself shrinks, while the interpretation of its domain grows due to the contravariance of the domain ($S_1 <: T_1$).
A \pathselection{\sldom} selector therefore reverses the variance of interpretations, while \pathselection{\slrange} preserves it.
We record this as the \emph{path positivity} $\possl{s}$, and generalize the subset relation between value interpretations to the positivity-aware relation $\slleq$, under which interpretations vary covariantly on a positive path ($\slplus{s}$) and contravariantly on a negative one ($\slminus{s}$).
For $F <: \TFun{A}{B}$, the interpretation of $\TDom{F}$ is thus larger than that of $A$, which is the semantic reflection of the lower bound $A <: \TDom{F}$ in \Cref{domrange:sec:boundless-lowerbound}.
This relation is how subtyping is interpreted semantically in the presence of pending projections.

\subsubsection{Value Interpretations.}
$\valtype[]{T}{s}$ reads as the value type of the syntactic type $T$ at selection path $\pathselection{s}$ under environment $\rho$.
With the path fixed, it is a plain set of values.

\begin{figure}[t]\small
\begin{mdframed}
\begin{minipage}[t]{1.0\textwidth}
\judgement{Path Positivity and Sub Value Interpretation}{\BOX{\vstp{V}{V}}}\small

  \[\begin{array}{l@{\qquad}l@{\qquad}l@{\qquad}l@{\qquad}r}
    \possl{\slroot} = \text{True}    & \possl{(\slappend{\sldom}{s})} = \neg\, \possl{s} & \possl{(\slappend{\slrange}{s})} = \possl{s}  & \text{Path Positivity} \\[1ex]
  \end{array}\]    

  \[\begin{array}{l@{\quad}l@{\quad}l@{\qquad \qquad}r}
    \slplus{s} & \defeq & \{ \pathselection{s} \mid \possl{s} \} & \text{Positive Path} \\[1ex]
    \slminus{s} & \defeq & \{ \pathselection{s} \mid \neg\, \possl{s} \} & \text{Negative Paths} \\[1ex]
    \vstp{V_1}{V_2} &\defeq& \forall \pathselection{s},\, \withpath{V_1}{\slplus{s}} \subseteq \withpath{V_2}{\slplus{s}} \, \land\, \withpath{V_2}{\slminus{s}} \subseteq \withpath{V_1}{\slminus{s}} & \text{Sub-Value Interpretation}  \\[2ex] 
  \end{array}\]

\judgement{Value Interpretations}{\BOX{\valtype{T}{s}}}\small %
  \[\begin{array}{l@{\qquad}l@{\qquad}l}
    \valtype{B}{\slroot}       & = & \{ c \}    \\[1ex]
    \valtype{\TTop}{\slplus{s}} \,,\, \valtype{\TBot}{\slminus{s}}   & = & \{ v \} \\[0.3ex]
    \valtype{\TTop}{\slminus{s}} \,,\, \valtype[]{\TBot}{\slplus{s}}  & = & \varnothing \\ [1ex]
    \valtype{X}{s}       & = & \valtypeset{ v \in \withpath{\rho(X)}{s} } \\[1ex]

    \valtype{\TFun{T_1}{T_2}}{\slroot} & = & \valtypeset[]{
      \exists H,\, t,\, v = \vabs{H}{t} \land \forall v_1,\, v_1 \in \valtype{T_1}{\slroot} \to \\ 
      & & \qquad \quad \exists v_2,\,\tevaln{(H,\,(x,v_1))}{t}{v_2} \land v_2 \in \valtype{T_2}{\slroot} }\\[0.3ex]
    \valtype{\TFun{T_1}{T_2}}{\slappend{\sldom}{s}} & = &  \valtype{T_1}{s} \\[0.3ex]
    \valtype{\TFun{T_1}{T_2}}{\slappend{\slrange}{s}} & = &  \valtype{T_2}{s} \\[1ex]

    \valtype{\TDom{T}}{s} & = & \valtype{T}{\slappend{\sldom}{s}} \\[0.3ex]
    \valtype{\TRange{T}}{s} & = & \valtype{T}{\slappend{\slrange}{s}} \\ [1ex]    

    \valtype{\TAll{U}{T}}{\slroot} & = & \valtypeset[\vtabs{H}{t}]{
      \forall V,\, \likeFunctionType{V} \to \vstp{V}{\valtype{U}{}} \to \\
      & & \qquad \qquad \qquad  \exists v,\,\tevaln{H}{t}{v} \land v \in \valtype[(\rho,\,(X, V))]{T}{\slroot}
    } \\[1ex]
    \valtype{\TAll{U}{T}}{\slplus{(\slappend{sl}{s})}}, \, \valtype{B}{\slplus{(\slappend{sl}{s})}} & = & \{v\} \\[0.3ex]
    \valtype{\TAll{U}{T}}{\slminus{(\slappend{sl}{s})}}, \, \valtype{B}{\slminus{(\slappend{sl}{s})}} & = & \varnothing \\[1ex]

    \end{array}\]      

\end{minipage} 

\caption{The definitions of value interpretations for \domranlang calculus.} \label{domrange:fig:LR-valtype}
\end{mdframed}
\vspace{-2ex}
\end{figure} 
$\TDom{T}$ and $\TRange{T}$ delay their interpretations, pushing the corresponding selector to the current path and interpreting $T$ at the new path.
For example, $v \in \valtype[\varnothing]{\TDom{T}}{\slroot}$ holds exactly when $v \in \valtype[\varnothing]{T}{\slappend{\sldom}{\slroot}}$.
Dually, the interpretation of an arrow type under a non-empty path does not describe function values,
but instead resolves the pending projection against the arrow's structure.
If the first selector is $\pathselection{\sldom}$ for $\valtype[\varnothing]{\TFun{T_1}{T_2}}{\slappend{\sldom}{\slroot}}$, the interpretation becomes that of $T_1$, and $T_2$ is irrelevant.
Interpreting $\valtype[\varnothing]{\TDom{\TFun{T_1}{T_2}}}{s}$ thus first delays, reaching $\valtype[\varnothing]{\TFun{T_1}{T_2}}{\slappend{\sldom}{s}}$, and then resolves to $\valtype[\varnothing]{T_1}{s}$, extracting the domain interpretation.

At the root path, an arrow type is interpreted as usual, as the set of closures \vabs[x]{H}{t} that map arguments to results semantically: whenever a value $v_1$ inhabits the interpretation of $T_1$, evaluating the body $t$ under $H$ extended with $v_1$ yields a value in the interpretation of $T_2$.
The interpretation of bounded quantification follows likewise.
The body must inhabit its interpretation under any value interpretation $V$ for the abstracted variable that is semantically bounded, via $\slleq$, by the interpretation of the syntactic bound $U$.
The additional constraint $\likeFunctionType{V}$ requires the \emph{well-formed functionality} of $V$ (\Cref{domrange:sec:LR-invariants}), justified at type instantiation.

Base types are interpreted as sets of value constants by definition.
$\TTop$ is interpreted as the set of all values under \emph{positive} selection paths, and $\TBot$ as the set of all values under \emph{negative} paths.
Under negative paths, subtyping swaps the two: $\TDom{\TTop} <: \TDom{T}$ holds for every $T$ by contravariance, forcing $\TTop$ to be empty there, and $\TDom{\TBot} \equiv \TTop$ (\Cref{domrange:sec:boundless-whysound}) forces $\TBot$ to contain all values.
A base type or a quantifier under a non-root path, as in $\TRange{\TInt}$, has no components to project, so we interpret it as $\TTop$ under the same path, which makes the domain of a non-function empty, such as $\TDom{\TInt}$.

\subsection{Environment Interpretations and Invariants} \label{domrange:sec:LR-invariants}

\begin{figure}\small
\begin{mdframed}
\judgement{Well-Formed Functionality}{\BOX{\likeFunctionType{V}}}%
  
\[\begin{array}{l@{\quad}l@{\quad}l@{}l}
  \likeFunctionType{V} &\defeq& \forall v,\,\pathselection{s},\,
      v \in \withpath{V}{s} \to \forall v_1,\, v_1 \in \withpath{V}{\slprepend{s}{\sldom}} \to \exists H,\,t,\, v = \vabs[]{H}{t} \,\land \\
    & & \exists v_2,\, \tevaln{(H,\,(x,v_1))}{t}{v_2} \land v_2 \in \withpath{V}{\slprepend{s}{\slrange}} \\[2ex] 
\end{array}\]

\vspace{2ex}
\judgement{Environment and Expression Interpretation}{\BOX{\envtype{\Gamma}} \BOX{\exptype{T}}}
\[\begin{array}{l@{\quad}l@{\quad}l@{}l}
  \envtype{\Gamma} &\defeq& \envtypeset[]{\DOM \Gamma = \DOM H\, \cup\, \DOM \rho \, \land \\ 
    & & \qquad\qquad (\forall x,\, T,\, x : T \in \Gamma \to \exists v,\, H(x) = v \land v \in \valtype[]{T}{\slroot})\, \land \\
    & & \qquad\qquad (\forall X,\, U,\, X <: U \in \Gamma \to \exists V,\, \rho(X) = V \land \vstp{\valtype[]{X}{}}{\valtype[]{U}{}} \land \likeFunctionType{V} )
  } \\[1ex] 
  \exptype{T} &\defeq& \{ \langle H,\rho, t \rangle \mid \exists v,\, \tevaln{H}{t}{v} \land v \in \valtype[]{T}{\slroot}   \} \\[2ex] 
\end{array}\]

\vspace{1pt}

\caption{The environment and expression interpretation of the \domranlang-calculus. }\label{domrange:fig:LR-env-interpretation}
\end{mdframed}
\vspace{-2ex}
\end{figure} 
\Cref{domrange:fig:LR-env-interpretation} presents the environment interpretation and the well-formedness invariant for value interpretations.

\subsubsection{Well-Formed Functionality.}
The property $\likeFunctionType{V}$ states that a value interpretation behaves like that of a function type once its domain is inhabited. 
If some value $v_x$ inhabits the domain projection of $V$,
then every value $v$ of $V$ is a closure (the semantic witness property of \Cref{domrange:sec:boundless-whysound}),
and applying $v$ to $v_x$ produces a result inhabiting the range projection.
For non-function types the property holds trivially, since their domain projections are uninhabited.
Notably, the selectors are \emph{prepended} to the selection path (written \pathselection{\slprepend{s}{\sldom}}) rather than appended.
Prepending makes the new projection the last to be resolved, so the property constrains the interpretation as a whole and holds independently of the pending projections already in \pathselection{s}.

Well-formed functionality is the most important invariant in our system, as it supports the projection-based application \rulename{domrange:tappdomrange}.
Intuitively, it states that every syntactic type can appear at the function position of an application, and describes how its values behave like functions when applied.
The shape of \likeFunctionType{V} also reflects the call-by-value requirement, as its premise is a value inhabiting the domain, so the argument must be evaluated before the call happens.
Our semantics therefore does not cover call-by-name, as discussed in \Cref{domrange:sec:discussion-cbv}.

\subsubsection{Environment and Expression Interpretations.}
A syntactic environment $\Gamma$ is interpreted as a pair of a value environment $H$ and a value type environment $\rho$.
A term binding $x : T$ requires the value recorded in $H$ to inhabit the interpretation of $T$.
A type binding $X <: U$ requires $X$ to represent an interpretation $V$ that is bounded by $\valtype[]{U}{\slroot}$ and satisfies well-formed functionality.

\begin{figure}\small
\begin{mdframed}
\judgement{Semantic Interpretation of Subtypes and Types}{\BOX{\semstp{\Gamma}{T}{T}} \BOX{\semtype{\Gamma}{t}{T}}}%
  
\[\begin{array}{l@{\quad}l@{\quad}l@{}l}
  \semstp{\Gamma}{T_1}{T_2} &\defeq& \forall \envs{H}{\rho} \in \envtype{\Gamma},\, \vstp{\valtype[]{T_1}{}}{\valtype[]{T_2}{}} \\[1ex]
  \semtype{\Gamma}{t}{T} &\defeq& \forall \envs{H}{\rho} \in \envtype{\Gamma},\, \exps{H}{\rho}{t} \in \exptype{T}
\end{array}\]

\vspace{2ex}
\judgement{Semantic Interpretation of Typing (Selected)}{\BOX{\semtype{\Gamma}{t}{T}}} \\
\typicallabel{stsub}
\begin{tabular}{b{.48\linewidth}@{}b{.5\linewidth}}
  \infrule[\ruledef{domrange:stsub}{st-sub}]{
    \semtype{\Gamma}{t}{T_1} \qquad
    \semstp{\Gamma}{T_1}{T_2}
  }{
    \semtype{\Gamma}{t}{T_2}
  }
  &
  \infrule[\ruledef{domrange:stappdr}{st-appdr}]{
    \semtype{\Gamma}{t_1}{F} \qquad
    \semtype{\Gamma}{t_2}{\TDom{F}}
  }{
    \semtype{\Gamma}{t_1~t_2}{\TRange{F}}
  }
\end{tabular} \\[1.5em]

\vspace{1pt}

\caption{The semantic interpretation of subtyping and typing (selected) of the \domranlang-calculus. }\label{domrange:fig:LR-semtype}
\end{mdframed}
\vspace{-2ex}
\end{figure} 
\subsubsection{Semantic Interpretation of Typing and Subtyping.}
\Cref{domrange:fig:LR-semtype} defines the interpretation of typing and subtyping.
The interpretation of subtyping lifts $\slleq$ to open types under an environment; value bindings are irrelevant there, so only $\rho$ is consulted.
The interpretation of typing is defined in the standard way over semantic environments, so we omit its full presentation.

\subsection{Metatheory}  \label{domrange:sec:LR-metatheory}
We state the main results and refer to our artifact for the full development; the higher-kinded $\hkdomranlang$, which strictly extends $\domranlang$, is detailed in the extended version~\cite{extended_version}.

The key lemma is that every value interpretation drawn from a well-interpreted environment satisfies well-formed functionality, which justifies the invariant assumed at type instantiation.
\begin{lemma}[semantic functionality] \label{domrange:thm:likeFunctionType}
  If $\,\envs{H}{\rho} \in \envtype{\Gamma}$, then for any type $T$, $\likeFunctionType{\valtype{T}{}}$.
\end{lemma}
\begin{proof}
  By induction on $T$. Type variables rely on the environment interpretation, while the other non-arrow base cases hold vacuously as the domain selector flips positivity. The remaining cases follow from the induction hypothesis, which covers all paths.
\end{proof}

\begin{theorem}[fundamental theorem of subtyping] \label{domrange:thm:stp-fundamental}
  If $\,\stp{\Gamma}{T_1}{T_2}$, then $\semstp{\Gamma}{T_1}{T_2}$.
\end{theorem}
\begin{theorem}[fundamental theorem] \label{domrange:thm:fundamental}
  If $\,\hastype{\Gamma}{t}{T}$, then $\semtype{\Gamma}{t}{T}$.
\end{theorem}
\begin{proof}
  Each by induction on the derivation, analyzing the selection path and its positivity at projections, with an inner induction on the path for subtyping.
\end{proof}

\begin{corollary}[safety]
  If $\,\hastype{\varnothing}{t}{T}$, then $\tevaln{\varnothing}{t}{v}$ in finitely many steps, and $v \in \valtype[\varnothing]{T}{\slroot}$.
\end{corollary}
\begin{proof}
  By specializing \Cref{domrange:thm:fundamental} to the empty environment.
\end{proof} %
\section{Product Projections}  \label{domrange:sec:casestudy}

Domain and range types project arrow types, while this section turns to product types, which are covariant in both components.
We enrich \domranlang with primitive pairs and, following the same approach, product projection types $\TFst{T}$ and $\TSnd{T}$.
The extension serves three purposes:
(1) it shows that the path selection technique of \Cref{domrange:sec:LR} is not specific to arrow types;
(2) it shows that two projection families cooperate in one system, so a program can choose which structure to expose; and
(3) it enables a direct comparison with the type destructors \cite{DBLP:journals/iandc/HofmannP02} on the common ground of the two designs (\Cref{domrange:sec:casestudy-destructors}).

\subsection{Syntactic Extensions} \label{domrange:sec:casestudy-syntactic}

\begin{figure}[t]\small
\begin{mdframed}
\begin{minipage}[t]{1.0\textwidth}
\judgement{Syntax}{\BOX{\domranlang}}\small %
  \[\begin{array}{l@{\qquad}l@{\qquad}l@{\qquad}l}
    t       & ::= & \dots \mid \tpair{t}{t} \mid \tfst{t} \mid \tsnd{t}    &  \text{Terms}           \\ [1ex]
    S,T,U & ::= &  \dots \mid \TPair{T}{T} \mid \TFst{T} \mid \TSnd{T}   & \text{Types}                   \\[2ex]
    \end{array}\]      

\judgement{Extended Subtyping}{\BOX{\stp{\Gamma}{T}{T}}} \\
\typicallabel{strans}
  \infrule[\ruledef{domrange:spair}{s-pair}]{
    \stp{\Gamma}{S_1}{T_1} \qquad \stp{\Gamma}{S_2}{T_2}
  }{
    \stp{\Gamma}{\TPair{S_1}{S_2}}{\TPair{T_1}{T_2}}
  } 
\vspace{-0.5em}
\begin{tabular}{b{.48\linewidth}@{}b{.5\linewidth}}
  \infrule[\ruledef{domrange:sfstintro}{s-fst-intro}]{
  }{
    \stp{\Gamma}{T_1}{\TFst{\TPair{T_1}{T_2}}}
  }  
  &
  \infrule[\ruledef{domrange:sfstelim}{s-fst-elim}]{
  }{
    \stp{\Gamma}{\TFst{\TPair{T_1}{T_2}}}{T_1}
  }
  \\[0.5em]
  \infrule[\ruledef{domrange:ssndintro}{s-snd-intro}]{
  }{
    \stp{\Gamma}{T_2}{\TSnd{\TPair{T_1}{T_2}}}
  }  
  &
  \infrule[\ruledef{domrange:ssndelim}{s-snd-elim}]{
  }{
    \stp{\Gamma}{\TSnd{\TPair{T_1}{T_2}}}{T_2}
  }
  \\[0.5em]  
  \infrule[\ruledef{domrange:sfstcongr}{s-fst-congr}]{
    \stp{\Gamma}{T_1}{T_2}
  }{
    \stp{\Gamma}{\TFst{T_1}}{\TFst{T_2}}
  }  
  &
  \infrule[\ruledef{domrange:ssndcongr}{s-snd-congr}]{
    \stp{\Gamma}{T_1}{T_2}
  }{
    \stp{\Gamma}{\TSnd{T_1}}{\TSnd{T_2}}
  }
\end{tabular} \\[1.5em]

\judgement{Extended Typing}{\BOX{\hastype{\Gamma}{t}{T}}} \\
\typicallabel{tsub}
  \infrule[\ruledef{domrange:tpair}{t-pair}]{
    \hastype{\Gamma}{t_1}{T_1} \qquad 
    \hastype{\Gamma}{t_2}{T_2}
  }{
    \hastype{\Gamma}{\tpair{t_1}{t_2}}{\TPair{T_1}{T_2}}
  } 
\vspace{-1em}
\begin{tabular}{b{.48\linewidth}@{}b{.5\linewidth}}
  \color{gray}{\infrule[\ruledef{domrange:tfststd}{t-fst-std}]{
    \hastype{\Gamma}{t}{\TPair{T_1}{T_2}}
  }{
    \hastype{\Gamma}{\tfst{t}}{T_1}
  }}
  &
  \color{gray}{\infrule[\ruledef{domrange:tsndstd}{t-snd-std}]{
    \hastype{\Gamma}{t}{\TPair{T_1}{T_2}}
  }{
    \hastype{\Gamma}{\tsnd{t}}{T_2}
  }}
  \\[0.5em]
  \infrule[\ruledef{domrange:tfst}{t-fst}]{
    \hastype{\Gamma}{t}{T} \qquad 
    \stp{\Gamma}{T}{\TPair{T_1}{T_2}}
  }{
    \hastype{\Gamma}{\tfst{t}}{\TFst{T}}
  }  
  &
  \infrule[\ruledef{domrange:tsnd}{t-snd}]{
    \hastype{\Gamma}{t}{T} \qquad 
    \stp{\Gamma}{T}{\TPair{T_1}{T_2}}
  }{
    \hastype{\Gamma}{\tsnd{t}}{\TSnd{T}}
  }  
\end{tabular} \\[1.5em]

\end{minipage} %

\caption{The syntactic extensions for pairs of the \domranlang calculus.
The standard projection rules \rulefmt{t-fst-std} and \rulefmt{t-snd-std} (in {\color{gray}gray}) are derivable from \rulefmt{t-fst} and \rulefmt{t-snd} via subtyping on first and second types.
} \label{domrange:fig:casestudy-syntactic}
\end{mdframed}
\vspace{-2ex}
\end{figure}
 
\Cref{domrange:fig:casestudy-syntactic} presents the syntactic extensions.
\tpair{t_1}{t_2} constructs a pair, projections \tfst{t} and \tsnd{t} eliminate it, and pairs receive product types \TPair{T_1}{T_2}.
Analogous to domain and range types, the projection types \TFst{T} and \TSnd{T} name the components of $T$ as first-class types over arbitrary $T$.

Subtyping is extended with congruence for product types, congruence for the projections, and introduction and elimination rules.
Since both components of a product are covariant, the projection congruences \rulename{domrange:sfstcongr} and \rulename{domrange:ssndcongr} are covariant.
The introduction and elimination rules (\eg, \rulename{domrange:sfstintro} and \rulename{domrange:sfstelim}) make the projection equivalent to the component on a precise product, $\TFst{\TPair{T_1}{T_2}} \equiv T_1$, like \rulename{domrange:sdomintro} and \rulename{domrange:sdomelim} on a precise arrow.

Typing adds the standard pair construction \rulename{domrange:tpair} and the projection rules.
Our rules \rulename{domrange:tfst} and \rulename{domrange:tsnd} type a projection at the abstract component $\TFst{T}$ or $\TSnd{T}$, requiring only that $T$ is a subtype of some product.
The component thus keeps its per-instantiation precision instead of collapsing to the bound's component.
The standard rules \rulename{domrange:tfststd} and \rulename{domrange:tsndstd} are derivable by projection elimination, so the figure shows them in gray, mirroring \rulename{domrange:tapp}.

Unlike boundless quantification via domain and range types, the projection rules do demand structural information here.
From a value of type $T$ alone, the system cannot conclude that $T$ has product structure, so projecting it requires a shape bound such as $T <: \TPair{\TTop}{\TTop}$:
\[
  \Lambda T <: \TPair{\TTop}{\TTop}.\, \lambda p : T.\, \tfst{p}
  \;:\; \TAll[T]{\TPair{\TTop}{\TTop}}{\TFun{T}{\TFst{T}}}
\]
Here the bound certifies that $T$ denotes some product type, letting the body project $p$ while the result keeps the precise projection type $\TFst{T}$.
This is the variance asymmetry of \Cref{domrange:sec:boundless-covariant}.

\subsection{Cooperating Projections} \label{domrange:sec:casestudy-examples}

With the pair extension, the system has two families of projections: domain and range for functions, and first and second for products.
They let a program abstract over type structure while exposing only what the code requires.

To illustrate both families at once, consider mapping $f$ over the first component of a pair.
The baseline names every component:
\[
  \begin{array}{r@{\;}l}
    \textit{mapFst} & := \Lambda A <: \TTop.\, \Lambda A' <: \TTop.\, \Lambda B <: \TTop.\, \lambda f : \TFun{A}{A'}.\, \lambda p : \TPair{A}{B}.\, \tpair{(f\,(\tfst{p}))}{(\tsnd{p})} \\
                    & : \TAll[A]{\TTop}{\TAll[A']{\TTop}{\TAll[B]{\TTop}{\TFun{(\TFun{A}{A'})}{\TFun{\TPair{A}{B}}{\TPair{A'}{B}}}}}}
  \end{array}
\]
It needs three type parameters, and its body is typed by the standard rules alone.
The variants below type the same function with different signatures, so we only list their types.

With projection types, we can flexibly abstract one structure, either the function or the pair, giving two signatures with two type parameters each:
\[
  \begin{array}{l@{\;\;:\;\;}l}
    \textit{mapFst}_1 & \TAll[F]{\TTop}{\TAll[B]{\TTop}{\TFun{F}{\TFun{\TPair{\TDom{F}}{B}}{\TPair{\TRange{F}}{B}}}}} \\
    \textit{mapFst}_2 & \TAll[T]{\TPair{\TTop}{\TTop}}{\TAll[A']{\TTop}{\TFun{(\TFun{\TFst{T}}{A'})}{\TFun{T}{\TPair{A'}{\TSnd{T}}}}}}
  \end{array}
\]
$\textit{mapFst}_1$ makes the function abstract, typing $f(\tfst{p})$ by \rulename{domrange:tappdomrange} from $f : F$ and $\tfst{p} : \TDom{F}$ without exposing $F$ as an arrow.
$\textit{mapFst}_2$ makes the pair abstract instead, where $\tfst{p} : \TFst{T}$ needs the bound $T <: \TPair{\TTop}{\TTop}$ as the premise of \rulename{domrange:tfst}.
The two signatures show the asymmetry of \Cref{domrange:sec:boundless-covariant}: the function needs no meaningful bound because its argument witnesses the arrow structure externally, while the pair must declare its structure up front internally.

Abstracting both structures at once is typeable, but gains nothing:
\[
  \textit{mapFst}_3 \;:\; \TAll[F]{\TTop}{\TAll[T]{\TPair{\TDom{F}}{\TTop}}{\TFun{F}{\TFun{T}{\TPair{\TRange{F}}{\TSnd{T}}}}}}
\]
The bound of $T$ must tie its first component to the domain of $F$, since otherwise $f(\tfst{p})$ would not type, and the signature still needs two type parameters.

\subsection{Comparison with Type Destructors} \label{domrange:sec:casestudy-destructors}

With product types, we can compare our approach against the type destructors \cite{DBLP:journals/iandc/HofmannP02} on the ground shared by the two systems.
The two designs reach the same per-instantiation precision on product types, but differ in what they ask of the type system and in whether a type must be cast.
Destructors cast a type via the $\eta$-equivalence into the form expected by the standard typing rules, which needs a bound to certify the shape, and the converse cast reconstructs the original type from its components.
Our typing rules instead keep the type as it is and project the expected components, such as $\TFst{T}$ in \rulename{domrange:tfst} or $\TDom{F}$ and $\TRange{F}$ in \rulename{domrange:tappdomrange}, so no type is cast or reconstructed.

\paragraph{Type destructors.}
The destructors $X.1$ and $X.2$ extract the components of a product type $X$, and the reasoning proceeds through the $\eta$-equivalence $X \equiv_\eta X.1 \times X.2$.
The well-formedness of destructors is governed by a \emph{kinding} judgment $\Gamma \vdash T \ll K$, which assigns every well-formed type $T$ a kind $K$ describing its shape, such as $K_1 \times K_2$ for a product type and $\TTop$ for a base or arrow type.
A destructor $X.1$ is well-kinded only when $X$ has a product kind, 
and a type variable reads its kind off its declared bound,
so $X.1$ is well-formed when the bound of $X$ certifies a product shape.
The $\eta$-equivalence relies on the kinding, so it covers every variable bounded by a product.
Their \textit{mix} example, which takes the first component of its first argument and the second component of its second argument, uses it in both directions:
\[
  \textit{mix} \;:=\; \Lambda X <: \TPair{\TTop}{\TTop}.\, \lambda x : X.\, \lambda y : X.\, \tpair{(\tfst{x})}{(\tsnd{y})}
\]
With destructors, \textit{mix} types at $\TAll[X]{(\TPair{\TTop}{\TTop})}{\TFun{X}{\TFun{X}{X}}}$.
One direction casts $x$ and $y$ to $X.1 \times X.2$ for the projections, and the other reconstructs the return value of type $X.1 \times X.2$ as $X$.

The destructors for arrow types, which they named but omitted, would follow the same pattern with the standard application rule.
They state the $\eta$-equivalence $F \equiv_\eta \TFun{F.\mathsf{dom}}{F.\mathsf{cod}}$ for these destructors whenever $F$ is a subtype of an arrow type.\footnote{\citet{DBLP:journals/iandc/HofmannP02} name these destructors $\mathsf{Dom}$ and $\mathsf{Cod}$ in an informal remark and give them no formal treatment. By analogy with their product kind $K_1 \times K_2$, we write the arrow kind as $K_1 \to K_2$. We write $F.\mathsf{dom}$ and $F.\mathsf{cod}$, in the style of their $X.1$, to distinguish them from our $\TDom{F}$ and $\TRange{F}$.}
Applying $f : F$ then first casts $f$ to $\TFun{F.\mathsf{dom}}{F.\mathsf{cod}}$ through the $\eta$-equivalence to expose its arrow shape, which needs $F$ to have an arrow kind from an arrow bound (left).
A wrapper of $f$ has type $\TFun{F.\mathsf{dom}}{F.\mathsf{cod}}$, and using it at $F$ needs the other direction $\TFun{F.\mathsf{dom}}{F.\mathsf{cod}} <: F$, \ie, the reconstruction of $F$ from its components.
{\mprset{sep=0.8em}
\[
  \inferrule*[right=\textsc{t-app}]{
    \inferrule*[right=$\eta$]{\hastype{\G}{f}{F} \\ \G \vdash F \ll K_1 \to K_2}{\hastype{\G}{f}{\TFun{F.\mathsf{dom}}{F.\mathsf{cod}}}}
    \\ \hastype{\G}{t}{F.\mathsf{dom}}
  }{
    \hastype{\G}{f~t}{F.\mathsf{cod}}
  }
  \quad
  \inferrule*[right=\textsc{t-app-dr}]{
    \hastype{\G}{f}{F} \\ \hastype{\G}{t}{\TDom{F}}
  }{
    \hastype{\G}{f~t}{\TRange{F}}
  }
\]}

\paragraph{Projection types.}
Our projections are instead well-formed over every type, and \domranlang has no shape kinding judgment.
Without shape kinding, $\TFst{\TBool}$ is still a well-formed type, as $\TDom{\TInt}$ is.
For functions, our rule \rulename{domrange:tappdomrange} checks an application directly from $f : F$ (right), so $f$ is never cast to an arrow type, and $F$ needs no exposure of arrow shape.
In exchange, we have no $\eta$-equivalence for abstract types.
The equivalence is recovered only when a type is a known product or arrow type, by the introduction and elimination rules, never from it being a subtype of one.
For \textit{mix}, no rule can derive $\TPair{\TFst{X}}{\TSnd{X}} <: X$, so we type the same term with the result expanded manually:
\[
  \textit{mix} \;:\; \TAll[X]{(\TPair{\TTop}{\TTop})}{\TFun{X}{\TFun{X}{\TPair{\TFst{X}}{\TSnd{X}}}}}
\]
Our signature thus cannot write the result type as $X$, which is an ergonomic loss, and for an abstract $X$, the result cannot be used where an $X$ is expected.
The precision is still preserved, as any valid instantiation $X := \TPair{A}{B}$ makes $\TPair{\TFst{X}}{\TSnd{X}} \equiv \TPair{A}{B} \equiv X$.

In return, our design needs neither shape kinding nor reconstruction.
Without the kinding, a component of type $\TFst{T}$ can be received, stored, or consumed by an arrow such as $\TFun{\TFst{T}}{A}$ while $T$ itself stays bounded by $\TTop$, and the bound $T <: \TPair{\TTop}{\TTop}$ is needed only where a value of type $T$ is actually projected.
Our typing rules project from a type only the components they check, such as $\TDom{F}$ and $\TRange{F}$ in \rulename{domrange:tappdomrange}, 
so $f$ keeps the type $F$ and there is no potentially lossy reconstruction (\Cref{domrange:sec:overview-practice}).

\subsection{Semantic Extensions}

\begin{figure}[t]\small
\begin{mdframed}
\begin{minipage}[t]{1.0\textwidth}
\judgement{Values and Path Selectors}{\BOX{\domranlang}}\small %
  \[\begin{array}{r@{\qquad}l@{\qquad}l@{\qquad}l}
    v       & ::= &  \dots \mid \vpair{v}{v}   &  \text{Values}           \\ [1ex]
    \pathselection{sl} & ::= &  \dots \mid \pathselection{\slfst} \mid \pathselection{\slsnd}   & \text{Path Selectors}           \\[1ex]
    \possl{(\slappend{\slfst}{s})} = \possl{s} & & \hspace{-2ex} \possl{(\slappend{\slsnd}{s})} = \possl{s}  & \text{Path Positivity} \\[2ex]
  \end{array}\]    

\judgement{Big-Step Reduction}{\BOX{\strut H,\,t \Downarrow v }}\\
\typicallabel{strans}
\begin{tabular}{b{.38\linewidth}@{}b{.30\linewidth}@{}b{.30\linewidth}}
  \infrule[\ruledef{domrange:epair}{e-pair}]{
    H,\,t_1 \Downarrow v_1 \qquad
    H,\,t_2 \Downarrow v_2 
  }{
    H,\,\tpair{t_1}{t_2} \Downarrow \vpair{v_1}{v_2}
  }
  &
  \infrule[\ruledef{domrange:efst}{e-fst}]{
    H,\,t \Downarrow \vpair{v_1}{v_2}
  }{
    H,\,\tfst{t} \Downarrow v_1
  }
  &
  \infrule[\ruledef{domrange:esnd}{e-snd}]{
    H,\,t \Downarrow \vpair{v_1}{v_2}
  }{
    H,\,\tsnd{t} \Downarrow v_2
  }
\\[1.5em]
\end{tabular}

\judgement{Well-Formed Pair}{\BOX{\likePairType{V}}}%
  
\[\begin{array}{l@{\quad}l@{\quad}l@{}l}
  \likePairType{V} &\defeq& \forall v_1,\,v_2,\, \pathselection{s},\, 
      \vpair{v_1}{v_2} \in \withpath{V}{s} \to v_1 \in \withpath{V}{\slprepend{s}{\slfst}} \land v_2 \in \withpath{V}{\slprepend{s}{\slsnd}} \\[2ex] 
\end{array}\]

\end{minipage} %

\caption{The big-step reductions and well-formedness for pairs from \domranlang calculus.} \label{domrange:fig:casestudy-reduction}
\end{mdframed}
\vspace{-2ex}
\end{figure}

\begin{figure}[t]\small
\begin{mdframed}
\begin{minipage}[t]{1.0\textwidth}

\judgement{Value Interpretations (Updated)}{\BOX{\valtype{T}{s}}}\small %
  \[\begin{array}{l@{\qquad}l@{\qquad}l}
    \valtype{\TAll{U}{T}}{\slroot} & = & \valtypeset[\vtabs{H}{t}]{
      \forall V,\, \likeFunctionType{V} \to \likePairType{V} \to \vstp{V}{\valtype{U}{}} \to \\
      & & \qquad \qquad \qquad  \exists v,\,\tevaln{H}{t}{v} \land v \in \valtype[(\rho,\,(X, V))]{T}{\slroot}
    } \\[1ex]

    \valtype{\TPair{T_1}{T_2}}{\slroot} & = & \valtypeset[]{
      \exists v_1,\, v_2,\, v = \vpair{v_1}{v_2}\, \land v_1 \in \valtype[]{T_1}{\slroot} \land v_2 \in \valtype[]{T_2}{\slroot}
    } \\[0.3ex]
    \valtype{\TPair{T_1}{T_2}}{\slappend{\slfst}{s}} & = & \valtype{T_1}{s} \\[0.3ex]
    \valtype{\TPair{T_1}{T_2}}{\slappend{\slsnd}{s}} & = & \valtype{T_2}{s} \\[1ex]

    \valtype{\TFst{T}}{s}  & = & \valtype{T}{\slappend{\slfst}{s}} \\[0.3ex]
    \valtype{\TSnd{T}}{s}  & = & \valtype{T}{\slappend{\slsnd}{s}} \\[0.3ex]
    \end{array}\]    

\judgement{Environment Interpretation (Updated)}{\BOX{\envtype{\Gamma}}}
\[\begin{array}{l@{\quad}l@{\quad}l@{}l}
  \envtype{\Gamma} &\defeq& \envtypeset[]{\DOM \Gamma = \DOM H\, \cup\, \DOM \rho \, \land \\ 
    & & \qquad\qquad (\forall x,\, T,\, x : T \in \Gamma \to \exists v,\, H(x) = v \land v \in \valtype[]{T}{\slroot})\, \land \\
    & & \qquad\qquad (\forall X,\, U,\, X <: U \in \Gamma \to \exists V,\, \rho(X) = V \land \vstp{\valtype[]{X}{}}{\valtype[]{U}{}} \land \likeFunctionType{V} \land \likePairType{V} )
  } \\[2ex] 
\end{array}\]

\end{minipage} %

\caption{The well-formedness, value interpretations, and updated environment types for pairs.} \label{domrange:fig:casestudy-semantic}
\end{mdframed}
\vspace{-2ex}
\end{figure} 
We extend the semantic interpretations and reduction rules to establish soundness and weak normalization for the extended system.
\Cref{domrange:fig:casestudy-reduction} presents the new path selectors, the big-step reduction for pairs, and the well-formedness property for product types.
The big-step rules are standard:
two components evaluate independently to values, forming the pair value $\vpair{v_1}{v_2}$, and the projections \tfst{t} and \tsnd{t} then evaluate to the corresponding components.

\paragraph{Path Selectors.}
Analogous to the domain and range selectors, the first selector (\pathselection{\slfst}) and the second selector (\pathselection{\slsnd}) let the logical relation delay the interpretation of product projections until product structure becomes available.
Since both components of a product are covariant, the new selectors preserve path positivity.

\paragraph{Well-Formedness.}
The well-formed pair property \likePairType{V} states how a value interpretation behaves like that of a product type. 
If a pair value $\vpair{v_1}{v_2}$ belongs to $V$, then $V$ admits the corresponding component projections.
It plays the role that well-formed functionality \likeFunctionType{V} plays for arrow types, allowing the interpretation to stay delayed until a pair witness is observed, 
at which point the underlying structure becomes available for projection.

\paragraph{Value Interpretations.}
The value interpretation of product types in \Cref{domrange:fig:casestudy-semantic} follows the design for arrow types.
At the root path, a product type is interpreted as the pair values whose components inhabit the interpretations of $T_1$ and $T_2$.
When the head of the selection path is \pathselection{\slfst} or \pathselection{\slsnd}, the interpretation extracts the corresponding component interpretation.
With the new invariant \likePairType{V}, the interpretation of bounded quantification and the semantic environment additionally require the abstract value interpretations to satisfy it, exactly as with well-formed functionality.

\subsection{Metatheory}

The proofs of the fundamental theorems extend those of the base calculus, since $\TFst{T}$ and $\TSnd{T}$ behave like the covariant range type.
The key is the well-formed property like Lemma~\ref{domrange:thm:likeFunctionType} for pairs.
\begin{lemma}[semantic well-formed pairs] \label{domrange:thm:likePairType}
  If $\,\envs{H}{\rho} \in \envtype{\Gamma}$, then for any type $T$, $\likePairType{\valtype{T}{}}$.
\end{lemma}
\begin{proof}
  By induction on $T$ as for Lemma~\ref{domrange:thm:likeFunctionType}, since the new selectors preserve positivity.
\end{proof}
The cases \rulename{domrange:tfst} and \rulename{domrange:tsnd} hold because their premise $\stp{\Gamma}{T}{\TPair{T_1}{T_2}}$ and \Cref{domrange:thm:stp-fundamental} make the value of $t$ a pair, 
so Lemma~\ref{domrange:thm:likePairType} puts its components in the interpretations of $\TFst{T}$ and $\TSnd{T}$.

\subsection{Summary}

The extension required only new selectors, a well-formedness invariant, and value interpretations for product types, with the rest of the development untouched.
Path selection thus generalizes beyond arrow types, and the same pattern should extend further to other composite type formers.
Programs may also freely mix multiple projection families, exposing only the structure each part of the code requires.
\section{Limitations} \label{domrange:sec:limitations}

Here, we discuss the known limitations of \domranlang on intersection types and call-by-value reduction.

\subsection{Intersection Types} \label{domrange:sec:limitations-intersection}
Our application rule relies on $F$ being an arrow whenever a call happens (\Cref{domrange:sec:boundless-whysound}), so that $\TDom{F}$ and $\TRange{F}$ give the domain and range of that one arrow.
With intersection types, $F$ may instead be an intersection of several arrows, whose two projections need not come from the same arrow.
For example, instantiate $F$ with $(\TFun{\TInt}{\TInt}) \wedge (\TFun{\TBool}{\TBool})$, which an identity function inhabits in $F_\wedge$ \cite{DBLP:journals/mscs/Pierce97}.
By congruence, the second arrow gives $\TBool <: \TDom{F}$ and the first gives $\TRange{F} <: \TInt$, so \rulename{domrange:tappdomrange} types $f(\mathsf{true})$ at $\TInt$, although the identity returns a boolean.
In general, the domain of an intersection of arrows grows to the union of their domains, while its range shrinks to the intersection of their ranges, so the correlation between each argument and its result is lost.
Union types, in contrast, are less precise and thus harmless, as the domain then shrinks to the intersection of the domains and the range grows to the union of the ranges.

Intersection types can naturally characterize function overloading by
typing the overloaded function at the intersection of all signatures.
Consequently, our system cannot support overloading this way.
Scala has intersection types, but not in the general form of $F_\wedge$. A
Scala function has a single arrow type and inhabits an intersection of arrows 
only if the original arrow is a subtype of every component; thus, Scala 
functions cannot be polymorphic through an intersection.
Scala further allows method overloading, but requires resolving each call to a
single signature, without yielding an intersection.
In contrast, TypeScript has general intersections. It types
overloaded functions directly and lets the identity inhabit the intersection
above. However, it faces the same problem as our system: no single pair of
parameter and return types describes such an intersection.
The utility types therefore pick the last overload, which is acceptable for
TypeScript, as it does not aim for soundness, but not for our system:
congruence relates $\TDom{F}$ to the domain of every arrow above $F$, and fixing
it to any one of them would collapse subtyping.
The root cause is the witness, as the only witness in \domranlang is the
argument, which certifies the domain but says nothing about which arrow is
producing the result.
Supporting intersections thus needs more than this witness, and \Cref{domrange:sec:discussion-settheoretic} discusses a potential solution.

\subsection{Call-by-Value} \label{domrange:sec:discussion-cbv}
The delayed application relies on its argument serving as the witness that $F$ is an arrow (\Cref{domrange:sec:boundless-whysound}), which works only once the argument is a value.
Under call-by-value, the argument is evaluated before the call, whereas under call-by-name or lazy evaluation it is passed unevaluated and proves nothing yet.
Our calculus has no general recursion, but adding it would expose the difference, since a diverging term $\Omega$ can inhabit every type, including the otherwise empty $\TDom{\TInt}$.
For example, for the term ($\textit{etaErr}\;\Omega$) with $\textit{etaErr}$ in \Cref{domrange:sec:boundless-whysound},
call-by-name substitutes $\Omega$ for $x$ even though the value $42$ supplied for $f$ is not a function, resulting in the stuck application ($42\;\Omega$).
The same principle underlies path-dependent types in DOT, where a type-level construct is safe because a value witnesses it \cite{DBLP:conf/oopsla/RompfA16}.
Scala~3 accordingly treats a lazy value as a path only when its type is realizable, roughly when its type members have consistent bounds and so could be witnessed, and never treats a by-name parameter as one.
A call-by-name variant would thus need another witness, such as an arrow bound on $F$, which brings back the early exposure that boundless quantification avoids.

\section{Discussion and Future Work}  \label{domrange:sec:discussion}

We discuss possible extensions, destructors we do not recover, and the relation to type members and set-theoretic types.

\subsection{Effect Polymorphism and Relative Effects} \label{domrange:sec:discussion-effects}
The delayed arrow obligation may also apply to more than types.
Effect systems face an analogous definition-site problem, as the effect of a higher-order function includes the latent effect of its formal parameter, and conventional effect polymorphism names that latent effect as an explicit effect variable in every signature.
\citet{DBLP:conf/ecoop/RytzOH12} sketch a lightweight alternative, an annotation @pure(f)@ declaring a function pure except for the latent effect of its parameter @f@, so the effect stays unnamed at the definition site.
Such relative effects are natural in our setting.
Projecting the latent effect out of $f$'s function type plays the role of @pure(f)@, as the type is never exposed as an effect-annotated arrow type at the definition site, and the effect resolves at the call site once a concrete arrow type arrives.
The delay for effects is in fact easier here than for types, since effects occur covariantly.

Concretely, with typing judgments $\Gamma \vdash t : T \mid e$ carrying an effect $e$, arrows $T_1 \xrightarrow{e} T_2$ carrying a latent effect, and effects composed sequentially by $\rhd$ as in an effect quantale \cite{DBLP:journals/toplas/Gordon21}, a third projection $\mathsf{Eff}(F)$ reads the latent effect off a function type:
\[
  \mathsf{Eff}(T_1 \xrightarrow{e} T_2) \equiv e
  \qquad\quad
  \inferrule{F <: G}{\mathsf{Eff}(F) <: \mathsf{Eff}(G)}
  \qquad\quad
  \inferrule{\Gamma \vdash f : F \mid e_1 \\ \Gamma \vdash t : \TDom{F} \mid e_2}{\Gamma \vdash f\,t : \TRange{F} \mid e_1 \rhd e_2 \rhd \mathsf{Eff}(F)}
\]
The first rule eliminates the projection on a known arrow, the second gives it the covariant congruence, and the third extends our application rule with the latent effect, sequenced in call-by-value order.
The relative effect @pure(f)@ then becomes $\mathsf{Eff}(\Type{typeof}(f))$.

We leave the formal development to future work, as making effects meaningful requires an actual effect language, which is orthogonal to these rules.
Effects can be represented as labels or qualifiers on types as in \citet{DBLP:journals/pacmpl/LeeZLYSB24}, and reachability types \cite{DBLP:journals/pacmpl/BaoWBJHR21,DBLP:journals/pacmpl/WeiBJBR24} would further allow effects that depend on the parameters in a lightweight way.

\subsection{Higher-Kinded Types and \texttt{typeof}} \label{domrange:sec:discussion-typeof}

In the extended version~\cite{extended_version}, we extend \domranlang with type operators to \hkdomranlang, which adds domain and range projections to the higher-order subtyping system \FOmegaSub \cite{DBLP:journals/tcs/PierceS97} following the presentation of \citet{tapl_book}.
The projections as well as selection paths carry over with little change alongside a higher-kinded logical relation.
The @typeof@ operator in \Cref{domrange:sec:overview-practice} admits a straightforward formal account, either by a type $\Type{typeof}(x) <: \Gamma(x)$ with the variable typing $x : \Type{typeof}(x)$, or by binding a companion type variable alongside each term parameter, which is a parallel development and omitted here.

\subsection{The Remaining Type Destructors} \label{domrange:sec:discussion-destructors}

\Cref{domrange:sec:casestudy} recovers the product fragment of the type destructors of \citet{DBLP:journals/iandc/HofmannP02}, whose system also has destructors for existential and recursive types.
We expect projections on existential types to integrate as smoothly as those on pairs, since an existential type is likewise eliminated by a term-level $\Type{unpack}$ and its projection would take the shape witness from a bound.
One application they sketch for the existential type destructor is encoding type operators, which we instead address directly by developing projection types in the higher-kinded setting of \hkdomranlang.
We skip projections on recursive types deliberately, as recursive types themselves would bring general recursion and force a step-indexed logical relation, losing weak normalization.
Since our contribution is mainly the contravariant arrow case, recovering all the destructors systematically would add little to it.

\subsection{Type Parameters and Type Members} \label{domrange:sec:discussion-members}
Readers familiar with DOT may notice similarities, from the selection paths of our logical relation (\Cref{domrange:sec:LR}) to the value that certifies our application rule (\Cref{domrange:sec:discussion-cbv}).
They reflect a common question, how type parameters relate to type members, and unifying the two would be a natural follow-up to this work.
DOT has no type parameters of its own and encodes them as members selected from terms \cite{DBLP:conf/birthday/AminGORS16}.
In early experiments with Dotty, the Scala~3 compiler, \citet{DBLP:conf/scala/OderskyMP16} adopted this view to explain one concept instead of two, encoding type arguments as refinements whose bounds follow variance and wildcards as members left abstract.

\paragraph{Java wildcards.}
Wildcards also defer a function type's components to its use sites, and in the view above, @Function<?, ?>@ leaves both members abstract.
However, wildcards are existential, and Java captures every use as a fresh type, so two uses are not comparable, as in @list.set(0, list.get(1))@ on a @List<?>@.
A @Function<?, ?>@ likewise accepts only @null@, as no other value has its fresh domain type, and returns only @Object@, behaving as $\TFun{\TBot}{\TTop}$ unless its components are type parameters.

\paragraph{Scala type members.}
Scala instead selects members from a stable term path.
Encoding a function as an object with abstract type members,
\begin{lstlisting}
  trait Fn { type Arg; type Res; def apply(x: Arg): Res }
  def eta(f: Fn)(x: f.Arg): f.Res = f.apply(x)
\end{lstlisting}
names the domain and range of @f@ as @f.Arg@ and @f.Res@.
The bound @Fn@ costs no precision, unlike $\TFun{\TBot}{\TTop}$, and interval bounds on its members allow typing $g$ in \Cref{domrange:sec:boundless-lowerbound}.
This seems close to \domranlang, but uses a different design, selecting members from a term rather than projecting an abstract type.
First, @f.Arg@ is a path-dependent type, which depends on the term @f@, while the type-level alternative @F#Arg@ for an abstract @F@ is rejected by Scala~3 as unsound because of bad bounds \cite{dotty1050}.
Second, @Fn@ is not Scala's function type @Function1@, whose type parameters are kept private, so a function exposes its components only when wrapped as an @Fn@ object.

\paragraph{Domain and range types.}
Our projections start from the other end.
Rather than encoding type parameters as members, they project the abstract parameter $F$ itself, with variance matching Dotty's encoding.
They thus depend on a type rather than a term, so a projection is per type rather than per function: $\TDom{F}$ is common to every function of type $F$, whereas Scala's @f.Arg@ and @g.Arg@ stay distinct even for @f@ and @g@ of the same abstract type.
A term is still needed, but only at the application.
Additionally, our results show where a witness is needed: a contravariant component can wait for the argument of the call, whereas a covariant one still needs a bound (\Cref{domrange:sec:boundless-covariant}).
Toward unifying type parameters and type members, the next steps are to extend such projections beyond arrows and to combine them with type members, including their intersections and bounds.

\subsection{A Set-Theoretic View of Projections} \label{domrange:sec:discussion-settheoretic}
Set-theoretic types (see \Cref{domrange:sec:related-work}) suggest a potential solution to the limitation of intersection types (\Cref{domrange:sec:limitations-intersection}).
There, a function type $F$ comes with a domain $\mathrm{dom}(F)$ and an application $F \circ S$, the type of the results of a function with type $F$ on arguments with type $S$, read as the image of $S$ under $F$ \cite{DBLP:journals/jacm/FrischCB08,DBLP:conf/popl/Castagna0XILP14}.
Instantiating $F$ with $(\TFun{\TInt}{\TInt}) \wedge (\TFun{\TBool}{\TBool})$ gives $\mathrm{dom}(F) = \TInt \vee \TBool$ but $F \circ \TBool = \TBool$.
In this view, our $\TDom{F}$ corresponds to their $\mathrm{dom}(F)$, while our $\TRange{F}$ matches their $F \circ \mathrm{dom}(F)$ only when $F$ is below a single arrow.
This suggests a first-class counterpart of their application operator, projecting the result of $F$ under a known argument type $S$, unlike our $\TRange{F}$, which is independent of the argument type.
Typing a call $f(x)$ with $x : S$ by this projection could keep each argument correlated with its result under intersections.
A semantic model and congruence rules for such a projection are left to future work.
\section{Related Work} \label{domrange:sec:related-work}

System \Fsub is introduced by \citet{DBLP:journals/csur/CardelliW85} and
formalized by
\citet{DBLP:journals/mscs/CurienG92,DBLP:journals/iandc/CardelliMMS94}.
It extends the polymorphic $\lambda$-calculus System~$\mathsf{F}$
\cite{girard1972interpretation, DBLP:conf/icalp/Reynolds74} with bounded
quantification $\TAll[]{U}{T}$, restricting $X$'s instantiations to subtypes of
$U$, thereby providing structural information about $X$ by exposing it to $U$.
Our \domranlang calculus builds on the full variant of \Fsub,
but our domain and range type projections need
no informative bound for arrow types; $<: \TTop$ suffices.
This enables $\eta$-expansion wrappers to be expressed without the precision
loss seen in \Cref{domrange:sec:introduction}.  
System \FOmegaSub extends \Fsub with higher-kinded type operators \cite{DBLP:journals/tcs/PierceS97}, and we adapt our work to it (\Cref{domrange:sec:discussion-typeof}).
Below, we examine related
work on reasoning about the structures of types in dynamic, gradual, and static
settings, respectively.

\subsection{Dynamically/Gradually Typed Systems}

\paragraph{TypeScript.}
TypeScript extends JavaScript with optional static
typing~\cite{goldberg2022learning}, supporting incremental annotation of
dynamically typed programs. Our work is motivated by its utility types
\cite{typescriptUtilityDocs}, in particular @ReturnType@ and @Parameters@, which
project the result and argument types of a function type.
We find these utilities widely used in open-source repositories (see the extended version \cite{extended_version}), while their definitions rely on the unsafe escape hatch @any@ for abstract function types.
Instead of formalizing or evolving all of TypeScript,
our goal here is to isolate this recurring idiom and to give it a sound foundation, 
which would also apply to other languages with similar type systems.

\paragraph{Python.}
Along the lines of how TypeScript programmers use @Parameters@,
Python allows forwarding function argument types via @ParamSpec@ \cite{pep612}:
given @P@ as an instance of @ParamSpec@, programmers can access the parameter
list as @P.args@ and @P.kwargs@;
with @P@ and @R@ bound as type variables, the type of a decorator can be
expressed as @Callable[P, R] -> Callable[P, R]@. Note that two variables (@P,R@)
and the arrow shape (@Callable@) are committed to at the definition site,
which is precisely the surface syntax burden we try to relieve via
@Dom@/@Range@ type projections.
Variadic generics \cite{pep646} are similarly designed to forward a variable number of type parameters, 
while they do not deal with keyword arguments (@.kwargs@ in @ParamSpec@), and are not specific to function parameters either.

\paragraph{Gradual Typing.}
It is worth first clarifying that neither TypeScript nor Python is a
\emph{sound} gradual type system: both of them have escape hatches to opt out
of type checking entirely, @any@ in TypeScript and @Any@ in Python.
Orthogonal to our work on statically safe function projections, 
several works study the formal foundations and extend the safety guarantees of 
TypeScript in gradual settings
\cite{DBLP:conf/ecoop/BiermanAT14,DBLP:conf/popl/RastogiSFBV15,DBLP:conf/pldi/VekrisCJ16,richards_et_al:LIPIcs.ECOOP.2015.76}.
\citet{DBLP:journals/pacmpl/GierczakMDA24} observe that \emph{transient}
gradually typed languages should be \emph{vigilant} for a \emph{tag} type
system, where tags represent only top-level shapes of types.
This aligns with our observation that seeing the top-level arrow
(in our work through @Dom@) suffices for safe function applications
at the shallowest level.
However, with polymorphism omitted and deeper type information erased,
their results likely do not carry over directly to a static, polymorphic system like \domranlang.

\subsection{Static Type-level Computation}

\paragraph{Type Destructors.}
With a similar motivation to improve the precision of \Fsub regarding
early exposure,
\citet{DBLP:journals/iandc/HofmannP02}
present $\mathsf{F}^{\textit{TD}}_{<:}$ with \emph{type destructors},
which we have compared with our projections in \Cref{domrange:sec:casestudy-destructors}.
They introduce destructors for \emph{covariant} positions only (the components
of products, the body of existentials, the unfolding of recursive types) and
explicitly omit destructors for arrows where \emph{contravariance} is involved,
even though they have named the omitted ones @Dom@ and @Cod@.
They cite ``difficult metatheoretic problems'' including algorithmic concerns
for subtyping, which we also leave to future work, and the lack of compelling applications at the time.
Our work is thus complementary:
we identify type forwarding patterns in TypeScript as an interesting use case,
and we deliver the omitted @Dom@ and @Range@ projections,
a declarative subtyping theory for them,
an application rule that delays resolving function types,
and a mechanized semantic type soundness proof
using polarity-based path selectors.
We show that our mechanism generalizes by revisiting their covariant product
destructors in \Cref{domrange:sec:casestudy} and discuss the rest in \Cref{domrange:sec:discussion-destructors}.

\paragraph{Type Constructors/Operators.}
Our projection types $\TDom{}$/$\TRange{}$ look similar to type constructors
or type operators \`a la System \FOmega and \FOmegaSub
\cite{girard1972interpretation,DBLP:journals/tcs/PierceS97}, but they are
largely different.
Crucially, type constructors/operators make new types, whereas type
destructors/projections resolve to existing components of given types.
Our \domranlang can be orthogonally combined with higher-kinded subtyping (\Cref{domrange:sec:discussion-typeof}).

Extended forms of type operators add expressiveness, yet we find no evidence
that they allow encoding destructors or projections.
F-bounded polymorphism \cite{DBLP:conf/fpca/CanningCHOM89} lets a type
variable's bound mention the variable itself (e.g.,
$T <: \Type{Comparable}[T]$), which is the right tool for \emph{describing} what
shape a variable must have.
\citet{DBLP:journals/pacmpl/StuckiG21} replace bounded quantification in
\FOmegaSub with interval kinds, allowing
a unified treatment of important type- and kind-level abstraction mechanisms.
Neither of them directly yields projection forms like $\TDom{F}$.

Some languages can encode projection types, but only by using mechanisms far
more powerful than higher-kinded types. C++ is the
extreme case: through recursive template instantiation and specialization, its
template system is believed to be
Turing-complete~\cite{veldhuizenTemplatesAreTuring2003},
supporting arbitrary (and undecidable) type-level computation,
as demonstrated by Unruh's 1994 program, which
prints prime numbers in error messages~\cite{unruh1994prime}.
Scala~3 and its match types stand at a more interesting middle point, which we
analyze below.

\paragraph{Match Types.}
Scala~3 match types \cite{DBLP:journals/pacmpl/BlanvillainBKO22} can \emph{encode} a domain projection
operationally: @type Dom[F] = F match { case (a => b) => a }@.
However, such a match reduces only when its type is concrete enough; for an
abstract @F@ it stalls.
Using @Dom[F]@ then forces a bound that exposes @F@'s arrow shape, bringing
back the same early-exposure problem seen in
\Cref{domrange:sec:introduction}.
As a concrete example:
\begin{lstlisting}
  type Dom[F] = F match { case (a => b) => a }
  def apply1[F](f: F, x: Dom[F]) = f(x)                   // ERROR: irreducible Dom[F]
  def apply2[F <: Int=>String](f: F, x: Dom[F]) = f(x)    // ERROR: expect Int, got Dom[F]
  def apply3(f: Int=>String, x: Dom[Int=>String]) = f(x)  // ok but monomorphic
\end{lstlisting}
Only the fully monomorphic @apply3@ type-checks, which is the least interesting case.
Our $\TDom{}$ and $\TRange{}$ are instead primitive projections that remain meaningful
on polymorphic types and do not require arrows to be exposed.

\paragraph{Wildcards and Type Members.}
Java's wildcards and Scala's type members both avoid naming the components of a function type as type parameters; we discuss both in \Cref{domrange:sec:discussion-members}.

\paragraph{Type Soundness and Set-Theoretic Types.}
We establish \emph{semantic type soundness} for \domranlang via logical
relations \cite{LogicalTypeSoundness}, rather than syntactically deriving
\emph{progress and preservation} theorems.
The operational semantics of \domranlang is defined in the big-step style, as a
definitional interpreter
\cite{DBLP:journals/lisp/Reynolds98a,DBLP:conf/popl/AminR17}.
Our path-selection technique is inspired by the strong normalization proof of
Dependent Object Types (DOT) \cite{DBLP:conf/ecoop/WangR17}, which uses
selectors for path-sensitive reasoning about type members. We instead apply
selectors to function types to reason about projections in a delayed manner.

While our semantic type soundness proof interprets types as
\emph{sets of possible values/terms},
our model differs from \emph{set-theoretic types}
\cite{DBLP:journals/jacm/FrischCB08,DBLP:conf/popl/Castagna0XILP14,DBLP:journals/programming/CastagnaDV23}
that include set connectives $\vee$, $\wedge$, $\neg$, $\setminus$ as
type formers.
For set-theoretic types,
their interpretations build on a domain with finite
relations to avoid cardinality problems.
To support expressive commutative and distributive laws,
\citet{DBLP:journals/jacm/FrischCB08} choose a \emph{semantic}
definition of subtyping, \ie, the set inclusion of interpretations,
as opposed to our definition using \emph{syntactic} rules.

For the analytic (algorithmic) versions of their systems,
\citet{DBLP:journals/jacm/FrischCB08,Cas24} define $\mathrm{dom}(\cdot)$ to
compute the domain types of functions and $(\cdot \circ \cdot)$ the types of 
applications to handle the complexity where the type of
a function can be a union of intersections of arrow types and their negations.
In \Cref{domrange:sec:discussion-settheoretic}, we discuss how these operators
could relate to our domain and range projection types and shed light on our
limitation on intersection types (see
\Cref{domrange:sec:limitations-intersection}).
Nevertheless, such operators in set-theoretic types are defined at
the meta level, not as first-class type formers, and provide no meaningful
reading of type variables.
Integrating them, especially the operator for application types, with our
delayed projection types remains open.

\section{Conclusion} \label{domrange:sec:conclusion}

We presented \domranlang, an extension of System \Fsub with first-class domain and range projection types and a redesigned application rule, together allowing boundless quantification.
Application delays the arrow obligation to the call site, 
so a function type variable needs no meaningful bound,
yet its applications still type soundly and precisely.
Every System \Fsub program remains typeable, and with an arrow bound, projection types express programs that \Fsub cannot.
We mechanized \domranlang in Rocq and proved semantic type soundness and weak normalization by logical relations, 
where path selection delays a projection's interpretation until the underlying arrow type is resolved.
Product projections extend the approach beyond arrows, 
cooperating with domain and range in one signature and reaching the same per-instantiation precision as the type destructors.
We expect the discipline to carry over to practical languages with subtyping, once the discussed limitations are addressed, 
realizing their projection idioms soundly and letting functions speak for themselves.

\section*{Disclosure of Use of AI in Research}

The Rocq mechanization in this work was developed by the authors, and generative AI tools were only utilized to assist in writing comments and README documentation in the artifact. 
In the extended version of this work~\cite{extended_version}, AI was employed to design and execute the experiment analyzing utility type usage.
The data cited in \Cref{domrange:sec:overview-practice} about utility types in 100 most-starred TypeScript repositories is from the experiment by AI, as elaborated in the extended version.

\section*{Data Availability Statement}
The Rocq mechanizations of \domranlang, its pair extension, and the higher-kinded \hkdomranlang (details in the extended version \cite{extended_version}), together with the examples, are available on Zenodo~\cite{artifact}.
The artifact is also available at \url{https://github.com/sweetsinpackets/Type-Destructors}.

\bibliography{references}

\newpage
\appendix

\section{Domain and Range Types with Higher-Order Subtyping} \label{domrange:appendix:fomega}

This appendix substantiates the claim of \Cref{domrange:sec:overview-domrange-delay} that domain and range projections are non-invasive and scale beyond the base calculus.
We extend \domranlang to a \emph{higher-kinded} setting, obtaining \hkdomranlang, the calculus \FOmegaSub with domain and range projection types.
The base calculus is recovered exactly by restricting \hkdomranlang to the base kind \kinding{\KTpe},
so the main development is left undisturbed while type operators, type-level abstraction, and type-level application are added on top.

We present the formal system in \Cref{image:sec:hkdr-formal} and establish its soundness by logical relations in \Cref{image:sec:hkdr-LR}.
Constructs shared with the base calculus, in particular the big-step operational semantics are unchanged, and we point back to the main paper rather than restate them.

\subsection{Formal Presentation} \label{image:sec:hkdr-formal}

In this section, we present the formal theory of \hkdomranlang, which extends the \FOmegaSub calculus with domain and range types.
This system strictly extends \domranlang (\Cref{domrange:sec:formal}) to a higher-kinded setting, \ie, \domranlang can be viewed as \hkdomranlang restricted to the base kind \kinding{\KTpe}.

The syntactic system mainly follows the presentation in \citet{tapl_book}.
To facilitate the logical relations (\Cref{image:sec:hkdr-LR}), we split the context $\Gamma$ in \citet{tapl_book} into three separate contexts: $\Gamma$ (term variables), $\Dt$ (type variable bounds), and $\Dk$ (type variable kinds).

\subsubsection{Syntax} \label{image:sec:hkdr-formal-syntax}

\begin{figure}[!ht]\small
\begin{mdframed}
\begin{minipage}[t]{1.0\textwidth}
\judgement{Syntax}{\BOX{\domranlang}}\small %
  \[\begin{array}{l@{\qquad}l@{\qquad}l@{\qquad}l}
    x,y,z   & \in & \TVar        & \text{Variables}               \\
    f,g,h   & \in & \TVar       & \text{Function Variables}      \\
    X       & \in & \TVar       & \text{Type Variables} \\
    t       & ::= & c \mid x \mid \lambda x{:}T.\,t \mid t~t \mid \TLam{X <: T}{t} \mid \TApp{t}{T}    &  \text{Terms}           \\ [2ex]

    \kind{K} & ::= & \kind{\KTpe} \mid  \kind{\KArr{K}{K}}  &  \text{Kinds}                   \\[2ex]

    S,T,U & ::= & B \mid \TFun{T}{T} \mid \TAll[X]{T}{T}   & \\
            &     & \mid \TTop \mid \TBot \mid \TAbs[X]{K}{T} \mid T ~ T &   \\
            &     & \mid \TDom{T} \mid \TRange{T}    & \text{Types}                   \\[2ex]

    \Gamma  & ::= & \varnothing\mid \Gamma,\, x : T                                 & \text{Term Variable Binding}     \\
    \Dt     & ::= & \varnothing\mid \Dt ,\, X <: T    & \text{Type Variable Bounds}  \\
    \Dk     & ::= & \varnothing \mid \Dk,\, X \hksymbol \kind{K}    & \text{Type Variable Kinds}   \\[2ex]
    \end{array}\]

\judgement{Type Macros}
\[ \TopK{\KTpe} := \TTop \qquad \TopK{\KArr{K_1}{K_2}} := \TAbs{K_1}{\TopK{K_2}}  \\[2ex]
\]

\end{minipage} %

\caption{The syntax of \hkdomranlang calculus.} \label{image:fig:hkdr-formal-syntax}
\end{mdframed}
\vspace{-2ex}
\end{figure} 
\Cref{image:fig:hkdr-formal-syntax} presents the syntax definitions and notations of the \hkdomranlang calculus, following the conventions of \domranlang (\Cref{domrange:sec:formal}) for continuity.

Terms include value constants, term variables (in lower case), $\lambda$-abstractions, term applications, type $\Lambda$-abstractions, and type instantiations.
Kinds, highlighted in \kinding{muted teal}, include base kinds \kinding{\KTpe} to classify types and arrow kinds \kinding{\KArr{K_1}{K_2}} to classify type operators.
Types include base types, arrow types\footnote{We use \emph{arrow types} for the shape $\TFun{A}{B}$, and \emph{function types} for arbitrary types that potentially describe a function.}, bounded type quantification for arbitrary kinds, a top type and a bottom type, as well as domain and range projection types.
Type operators of higher kinds are represented as type-level $\lambda$s and have type-level applications.

While \citet{tapl_book} uses $\Gamma$ to store both term variables and type variables, we store only term variables in $\Gamma$ and leave the type bounds of type variables in the $\Dt$ context.
As type variables may represent type operators of higher kinds, we additionally introduce a $\Dk$ context to store their kinds.
Note that the $\Dk$ context is not essential in theory, because the kinds of type variables (and their bound types) can be inferred from the kinding judgement.
However, $\Dk$ decouples the mutual dependency between subtyping \rulename{image:s1var} and kinding, allowing the kinding system to be defined independently of the $\Dt$ context, which simplifies the definition of semantic interpretations.

The type macro $\TopK{K}$ constructs a top type for any kind $\kind{K}$, serving as a valid type upper bound for type operators.

\subsubsection{Kinding}  \label{image:sec:hkdr-kinding}

\begin{figure}[t]
\setlength{\afterruleskip}{\bigskipamount}
\small
\begin{mdframed}
\judgement{Kinding}{\BOX{\strut \haskind{\Dk}{T}{K} }}\\
\typicallabel{k-cst}
\begin{tabular}{b{.3\linewidth}@{}b{.35\linewidth}@{}b{.34\linewidth}}
  \infrule[\ruledef{image:k1cst}{k-cst}]{
  }{
    \haskind{\Dk}{B}{\KTpe}
  }
  &
  \infrule[\ruledef{image:k1top}{k-top}]{
  }{
    \haskind{\Dk}{\TTop}{\KTpe}
  }
  &
  \infrule[\ruledef{image:k1var}{k-var}]{
    X \hksymbol \kind{K} \in \Dk
  }{
    \haskind{\Dk}{X}{K}
  }
\end{tabular} \\ [0.5em]
\begin{tabular}{m{.48\linewidth}@{}m{.5\linewidth}}
  \infrule[\ruledef{image:k1fun}{k-fun}]{
    \haskind{\Dk}{T_1}{\KTpe} \qquad
    \haskind{\Dk}{T_2}{\KTpe}
  }{
    \haskind{\Dk}{\TFun{T_1}{T_2}}{\KTpe}
  }  
  &
  \infrule[\ruledef{image:k1all}{k-all}]{
    \haskind{\Dk}{U}{K} \quad
    \haskind{\Dk\,,X\hksymbol\kind{K}}{T}{\KTpe}
  }{
    \haskind{\Dk}{\TAll{U}{T}}{\KTpe}
  }
\end{tabular} \\ [0.5em]
\begin{tabular}{m{.43\linewidth}@{}m{.55\linewidth}}
  \infrule[\ruledef{image:k1abs}{k-abs}]{
    \haskind{\Dk\,,X\hksymbol\kind{K_1}}{T}{\kind{K_2}}
  }{
    \haskind{\Dk}{\TAbs{K_1}{T}}{K_2}
  }
  & 
  \infrule[\ruledef{image:k1app}{k-app}]{
    \haskind{\Dk}{T_1}{\KArr{K_1}{K_2}} \qquad
    \haskind{\Dk}{T_2}{K_1}
  }{
    \haskind{\Dk}{T_1 ~ T_2}{K_2}
  }
\end{tabular} \\ [0.5em]
\begin{tabular}{m{.48\linewidth}@{}m{.5\linewidth}}
  \infrule[\ruledef{image:k1dom}{k-dom}]{
    \haskind{\Dk}{T}{\KTpe}
  }{
    \haskind{\Dk}{\TDom{T}}{\KTpe}
  }  
  &
  \infrule[\ruledef{image:k1range}{k-range}]{
    \haskind{\Dk}{T}{\KTpe}
  }{
    \haskind{\Dk}{\TRange{T}}{\KTpe}
  }
\end{tabular}
\\ [1.5em]
\caption{Kinding rules of \hkdomranlang calculus.} \label{image:fig:hkdr-formal-haskind}
\end{mdframed}

\vspace{-2ex}
\end{figure} 
\Cref{image:fig:hkdr-formal-haskind} defines the kinding system of the \hkdomranlang calculus.
Notably, the kinding system only requires the $\Dk$ context for kind witnesses of type variables.
This design keeps the kinding judgement minimal in its context requirements.
As a kinding judgement for a type naturally serves as a well-formedness proof, we can automatically exclude types that can never be inhabited or derived in a well-typed program; for example, $\TDom{\TAbs[]{\TTop}{\TTop}}$ is excluded due to kind mismatch.

All types already present in \domranlang have kind $\kind{\KTpe}$, as they are normal types.
The bounded quantification $\TAll[]{U}{T}$ also has kind $\kind{\KTpe}$, while the type upper bound $U$ is extended to higher kinds, permitting the abstraction and instantiation of type operators.
The type-level $\lambda$ constructs type operators, deriving type operators of arrow kinds.
Informally, the type-level $\lambda$-abstractions and applications can be viewed as a simply typed $\lambda$-calculus at the type level.

\subsubsection{Equivalence, Subtyping, and Typing} \label{image:sec:hkdr-static-typing}
\Cref{image:fig:hkdr-formal-equiv} presents the type equivalence rules of the \hkdomranlang calculus, \Cref{image:fig:hkdr-formal-subtyping} defines its subtyping, and \Cref{image:fig:hkdr-formal-typing} describes the static typing.

\begin{figure}[t]
\setlength{\afterruleskip}{\bigskipamount}
\small
\begin{mdframed}
\judgement{Type Equivalence}{\BOX{\strut \eqtype{\Dk}{T}{T} }}\\
\typicallabel{qbeta}
\begin{tabular}{b{.4\linewidth}@{}b{.58\linewidth}}
  \infrule[\ruledef{image:q1refl}{q-refl}]{
  }{
    \eqtype{\Dk}{T}{T}
  }
  &
  \infrule[\ruledef{image:q1beta}{q-$\beta$}\,]{
    \haskind{\Dk}{T_1}{K_1}
  }{
    \eqtype{\Dk}{\TTApp{(\TAbs[X]{K_1}{T_2})}{T_1}}{\substty[X]{T_1}{T_2}}
  } 
  \\[0.5em]
  \infrule[\ruledef{image:q1symm}{q-symm}]{
    \eqtype{\Dk}{T_1}{T_2}
  }{
    \eqtype{\Dk}{T_2}{T_1}
  }
  &
  \infrule[\ruledef{image:q1trans}{q-trans}]{
    \eqtype{\Dk}{T_1}{T_2} \quad
    \eqtype{\Dk}{T_2}{T_3}
  }{
    \eqtype{\Dk}{T_1}{T_3}
  }
\end{tabular} \\ [0.5em]
\begin{tabular}{b{.48\linewidth}@{}b{.5\linewidth}}
  \infrule[\ruledef{image:q1fun}{q-fun}]{
    \eqtype{\Dk}{T_1}{S_1} \quad 
    \eqtype{\Dk}{T_2}{S_2}
  }{
    \eqtype{\Dk}{\TFun{T_1}{T_2}}{\TFun{S_1}{S_2}}
  }
  &
  \infrule[\ruledef{image:q1all}{q-all}]{
    \eqtype{\Dk\,,X\hksymbol\kind{K}}{T_2}{S_2} \\
    \eqtype{\Dk}{T_1}{S_1}  \quad
    \haskind{\Dk}{T_1}{K}
  }{
    \eqtype{\Dk}{\TAll{T_1}{T_2}}{\TAll{S_1}{S_2}}
  }
  \\[0.5em] %
  \infrule[\ruledef{image:q1abs}{q-abs}]{
    \eqtype{\Dk\,,X\hksymbol\kind{K}}{T}{S}
  }{
    \eqtype{\Dk}{\TAbs{K}{T}}{\TAbs{K}{S}}
  }
  &
  \infrule[\ruledef{image:q1app}{q-app}]{
    \eqtype{\Dk}{T_1}{T_2} \quad
    \eqtype{\Dk}{S_1}{S_2}
  }{
    \eqtype{\Dk}{\TTApp{T_1}{T_2}}{\TTApp{S_1}{S_2}}
  } \\ [0.5em]
  \infrule[\ruledef{image:q1domrefl}{q-dom-refl}]{
    \eqtype{\Dk}{T}{S}
  }{
    \eqtype{\Dk}{\TDom{T}}{\TDom{S}}
  }
  &
  \infrule[\ruledef{image:q1rangerefl}{q-range-refl}]{
    \eqtype{\Dk}{T}{S}
  }{
    \eqtype{\Dk}{\TRange{T}}{\TRange{S}}
  } \\ [0.5em]
  \infrule[\ruledef{image:q1dom}{q-dom}]{
    \haskind{\Dk}{T_1}{\KTpe} \qquad 
    \haskind{\Dk}{T_2}{\KTpe}
  }{
    \eqtype{\Dk}{\TDom{\TFun{T_1}{T_2}}}{T_1}
  }
  &
  \infrule[\ruledef{image:q1range}{q-range}]{
    \haskind{\Dk}{T_1}{\KTpe} \qquad
    \haskind{\Dk}{T_2}{\KTpe}
  }{
    \eqtype{\Dk}{\TRange{\TFun{T_1}{T_2}}}{T_2}
  }  
\end{tabular} \\[1.5em]
\caption{Equivalence rules of \hkdomranlang calculus.} \label{image:fig:hkdr-formal-equiv}
\end{mdframed}

\vspace{-2ex}
\end{figure} 
\paragraph{Type Equivalence.}
The type conversion rules (\Cref{image:fig:hkdr-formal-equiv}) specify when two types are always interchangeable via \rulename{image:s1eq} and \rulename{image:t1sub}.
Applications can happen at the type level with type operators, and \rulename{image:q1beta} captures the evaluation of type applications via $\beta$-reduction on types.
Type equivalence describes definitional equivalence based solely on syntactic shapes, which is analogous to the reduction of terms.
Notably, type equivalence depends only on the $\Dk$ context for witnesses that types are well-kinded.

However, the equivalence relation is symmetric \rulename{image:q1symm}, which introduces significant flexibility in the shapes of possible types for a given term.
For example, a type $T$ can always be replaced with a more complicated shape $\TTApp{(\TAbs[]{\KTpe}{T})}{\TTop}$,
and thus there exists no canonical form of types for a value.
This observably increases the difficulty of justifying the soundness properties of the system, as any attempt at inversion on values must decouple the infinite equivalence and subtyping.

The domain and range types are equivalent to their corresponding projected types given a witness of an arrow type,
serving as both the introduction and elimination forms of domain and range types.
This relation is presented through bidirectional subtyping rules in the base \domranlang calculus (\Cref{domrange:sec:formal}),
while we encode it in type equivalence for a more compact presentation.

\begin{figure}[t]
\setlength{\afterruleskip}{\bigskipamount}
\small
\begin{mdframed}
\judgement{Subtyping}{\BOX{\strut \hkstp{\Dt}{\Dk}{T}{T} }}\\
\typicallabel{seq}
\begin{tabular}{b{.44\linewidth}@{}b{.54\linewidth}}
  \infrule[\ruledef{image:s1eq}{s-eq}]{
    \eqtype{\Dk}{T_1}{T_2} \\
    \haskind{\Dk}{T_1}{\kind{K}} \quad \haskind{\Dk}{T_2}{\kind{K}}
  }{
    \hkstp{\Dt}{\Dk}{T_1}{T_2}
  }
  &
  \infrule[\ruledef{image:s1trans}{s-trans}]{
    \hkstp{\Dt}{\Dk}{T_1}{T_2} \\
    \hkstp{\Dt}{\Dk}{T_2}{T_3} \quad     
    \haskind{\Dk}{T_2}{\kind{K}} 
  }{
    \hkstp{\Dt}{\Dk}{T_1}{T_3}
  }
\\ [0.5em]
\end{tabular}
\begin{tabular}{b{.36\linewidth}@{}b{.62\linewidth}}
  \infrule[\ruledef{image:s1top}{s-top}]{
    \haskind{\Dk}{T}{\KTpe}
  }{
    \hkstp{\Dt}{\Dk}{T}{\TTop}
  }
  &
  \infrule[\ruledef{image:s1fun}{s-fun}]{
    \hkstp{\Dt}{\Dk}{T_1}{S_1} \quad
    \hkstp{\Dt}{\Dk}{S_2}{T_2}
  }{
    \hkstp{\Dt}{\Dk}{\TFun{S_1}{S_2}}{\TFun{T_1}{T_2}}
  }
  \\[0.5em] %
\end{tabular}
\begin{tabular}{b{.4\linewidth}@{}b{.58\linewidth}}
  \infrule[\ruledef{image:s1var}{s-var}]{
    X <: T \in \Dt \quad \haskind{\Dk}{T}{K}
  }{
    \hkstp{\Dt}{\Dk}{X}{T}
  }
  &
  \infrule[\ruledef{image:s1all}{s-all}]{
    \hkstp{\Dt}{\Dk}{U_2}{U_1} \quad  \haskind{\Dk}{U_2}{K} \\
    \hkstp{\Dt\,,X <: U_2}{\Dk\,,X\hksymbol\kind{K}}{T_1}{T_2}
  }{
    \hkstp{\Dt}{\Dk}{\TAll{U_1}{T_1}}{\TAll{U_2}{T_2}}
  } \\ [0.5em]
\end{tabular}
\begin{tabular}{b{.58\linewidth}@{}b{.4\linewidth}}
  \infrule[\ruledef{image:s1abs}{s-abs}]{
    \hkstp{\Dt\,,X <: \TopK{K}}{\Dk\,,X\hksymbol\kind{K}}{S}{T}
  }{
    \hkstp{\Dt}{\Dk}{\TAbs{K}{S}}{\TAbs{K}{T}}
  }
  &
  \infrule[\ruledef{image:s1app}{s-app}]{
    \hkstp{\Dt}{\Dk}{S}{T}
  }{
    \hkstp{\Dt}{\Dk}{\TTApp{S}{U}}{\TTApp{T}{U}}
  }
  \\ [0.5em]
\end{tabular}
\begin{tabular}{b{.48\linewidth}@{}b{.5\linewidth}}
  \infrule[\ruledef{image:s1dom}{s-dom}]{
    \hkstp{\Dt}{\Dk}{T}{S} \quad 
    \haskind{\Dk}{T}{\KTpe}
  }{
    \hkstp{\Dt}{\Dk}{\TDom{S}}{\TDom{T}}
  }
  &
  \infrule[\ruledef{image:s1range}{s-range}]{
    \hkstp{\Dt}{\Dk}{S}{T} \quad
    \haskind{\Dk}{S}{\KTpe}
  }{
    \hkstp{\Dt}{\Dk}{\TRange{S}}{\TRange{T}}
  }
\end{tabular} \\[1.5em]
\caption{Subtyping rules of \hkdomranlang calculus.} \label{image:fig:hkdr-formal-subtyping}
\end{mdframed}

\vspace{-2ex}
\end{figure} 
\paragraph{Subtyping.}
The subtyping (\Cref{image:fig:hkdr-formal-subtyping}) extends the standard subtyping rules of the \FOmegaSub calculus presented in \citet{tapl_book}.
Without loss of generality, we impose extra well-kindedness premises on non-structural subtyping rules such as \rulename{image:s1var} and \rulename{image:s1trans},
so that the important property that subtyping preserves well-kindedness is self-contained.

The subtyping rule for bounded quantification is the full, descriptive rule rather than the kernel rule of \citet{tapl_book}, allowing the abstracted type bounds to shrink, consistent with \domranlang.
As type operators are not bounded-quantified, \ie, can be applied to arbitrary types of kind $\kind{K}$,
we use $\TopK{K}$ for their type bounds.
The type application subtyping rule \rulename{image:s1app} requires the same actual type argument.
The congruence rule for domain types is contravariant and the one for range types is covariant, reflecting the variance of arrow types.

\begin{figure}[t]
\setlength{\afterruleskip}{\bigskipamount}
\small
\begin{mdframed}
\judgement{Typing}{\BOX{\strut \hkhastype{\G}{\Dt}{\Dk}{t}{T} }}\\
\typicallabel{tsub}
\begin{tabular}{b{.4\linewidth}@{}b{.58\linewidth}}
  \infrule[\ruledef{image:t1cst}{t-cst}]{
  }{
    \hkhastype{\G}{\Dt}{\Dk}{c}{B}
  }
  &
  \infrule[\ruledef{image:t1sub}{t-sub}]{
    \hkhastype{\G}{\Dt}{\Dk}{t}{T_1} \\
    \hkstp{\Dt}{\Dk}{T_1}{T_2} \qquad
    \haskind{\Dk}{T_2}{\KTpe} 
  }{
    \hkhastype{\G}{\Dt}{\Dk}{t}{T_2}
  }
  \\[0.5em] %
  \infrule[\ruledef{image:t1var}{t-var}]{
    x : T \in \G   
  }{
    \hkhastype{\G}{\Dt}{\Dk}{x}{T}
  }
  &
  \infrule[\ruledef{image:t1appdr}{t-app-dr}]{
    \hkhastype{\G}{\Dt}{\Dk}{t_1}{F} \\
    \hkhastype{\G}{\Dt}{\Dk}{t_2}{\TDom{F}}
  }{
    \hkhastype{\G}{\Dt}{\Dk}{t_1 ~ t_2}{\TRange{F}}
  } \\ [0.5em]
\end{tabular}
\begin{tabular}{b{.52\linewidth}@{}b{.46\linewidth}}
  \infrule[\ruledef{image:t1abs}{t-abs}]{
    \haskind{\Dk}{T_1}{\KTpe} \qquad 
    \haskind{\Dk}{T_2}{\KTpe} \\
    \hkhastype{\G\,,T_1}{\Dt}{\Dk}{t}{T_2}
  }{
    \hkhastype{\G}{\Dt}{\Dk}{\lambda x{:}T_1.\, t}{\TFun{T_1}{T_2}}
  }
  &
  \infrule[\ruledef{image:t1app}{t-app}]{
    \hkhastype{\G}{\Dt}{\Dk}{t_2}{T_1} \\
    \hkhastype{\G}{\Dt}{\Dk}{t_1}{\TFun{T_1}{T_2}}
  }{
    \hkhastype{\G}{\Dt}{\Dk}{t_1 ~ t_2}{T_2}
  }
  \\ [0.5em]
  \infrule[\ruledef{image:t1tabs}{t-tabs}]{
    \haskind{\Dk}{U}{K} \qquad
    \haskind{\Dk\,,X\hksymbol\kind{K}}{T}{\KTpe}  \\
    \hkhastype{\G}{\Dt\,,X<:U}{\Dk\,,X\hksymbol\kind{K}}{t}{T}
  }{
    \hkhastype{\G}{\Dt}{\Dk}{ \TLam{X <: U}{t}}{\TAll{U}{T}}
  }
  &
  \infrule[\ruledef{image:t1tapp}{t-tapp}]{
    \hkhastype{\G}{\Dt}{\Dk}{t}{\TAll{U}{T}}
  }{
    \hkhastype{\G}{\Dt}{\Dk}{\TApp{t}{U}}{\substty[X]{U}{T}}
  }
\end{tabular} \\[1.5em]
\caption{Typing rules of \hkdomranlang calculus.} \label{image:fig:hkdr-formal-typing}
\end{mdframed}

\vspace{-2ex}
\end{figure} 
\paragraph{Typing.}
As terms cannot be typed with type operators, the typing rules only consider types of base kind $\kind{\KTpe}$, which reduces to the typing of the \domranlang calculus.
The bounded quantification rule \rulename{image:t1tabs} allows abstraction of type operators, while its body must derive a type of base kind $\kind{\KTpe}$ that can type a term.
With the descriptive subtyping rule for bounded quantification \rulename{image:s1all}, we can drop the subtyping premise in the type instantiation rule \rulename{image:t1tapp}, as we can shrink the type bound to match the actual type.
The non-standard application rule \rulename{image:t1appdr} types the application with more precise domain and range types,
which subsumes the standard application rule \rulename{image:t1app} when $F$ is known to be an arrow type.

\subsubsection{Metatheory of Kinding} \label{image:sec:hkdr-syntactic-metatheory}

Although type soundness is proven semantically through logical relations in \Cref{image:sec:hkdr-LR},
the metatheory of kinding still needs to be established from the syntactic rules.

\begin{lemma}[Kind Uniqueness]
  If\phantom{n}$\haskind{\Dk}{T}{K}$ and $\haskind{\Dk}{T}{K'}$, then $\kind{K} = \kind{K'}$.
\end{lemma}
\begin{lemma}[Kinding Proof Uniqueness] \label{image:hkdr-thm-hkunique}
  If\phantom{n}$a$ and $b$ prove $\haskind{\Dk}{T}{K}$, then $a = b$.
\end{lemma}
\begin{proof}
  By structural induction on the type $T$.
\end{proof}

The two lemmas ensure that a type can only be assigned a unique kind, and that all proofs of a kinding judgement are equivalently convertible.
This establishes the foundation for our interpretation of kinding judgements, in that kinding judgements from different sources derive the same value type.

\begin{lemma}[Kinding of Type Substitution] \label{image:hkdr-thm-hksubst}
  If\phantom{n}\haskind{\ctxext{\Dk}{X\hksymbol\kind{K'}}}{T}{K}, and \haskind{\Dk}{T'}{K'}, then $\haskind{\Dk}{\substty[X]{T'}{T}}{K}$. Also, if\phantom{n}\haskind{\Dk}{\substty[X]{T'}{T}}{K}, and \haskind{\Dk}{T'}{K'}, then \haskind{\ctxext{\Dk}{X\hksymbol\kind{K'}}}{T}{K}.
\end{lemma}
\begin{proof}
  By structural induction on the type $T$. The heterogeneous equality is established through \Cref{image:hkdr-thm-hkunique}.
\end{proof}

This lemma establishes that type substitution preserves the kinding judgement, and conversely, substitution on types implicitly infers the well-kindedness of the original type.

\begin{lemma}[Equivalence Preserves Kinding] \label{image:hkdr-thm-eqtype-preserve-kind}
  If\phantom{n}\eqtype{\Dk}{T_1}{T_2}, then for every $\kind{K}$, $\haskind{\Dk}{T_1}{K}$ if and only if $\haskind{\Dk}{T_2}{K}$.
\end{lemma}
\begin{proof}
  By structural induction on the equivalence judgement $\eqtype{\Dk}{T_1}{T_2}$, proving both directions together, as \rulename{image:q1symm} swaps them. The $\beta$ case is derived from \Cref{image:hkdr-thm-hksubst}. Other cases follow straightforwardly from the induction hypothesis.
\end{proof}
\begin{lemma}[Subtyping Preserves Kinding] \label{image:hkdr-thm-stp-preserve-kind}
  If\phantom{n}\stp{\ctxdelta{\Dt}{\Dk}}{T_1}{T_2} and $\wfdelta{\Dt}{\Dk}$, then for every $\kind{K}$, $\haskind{\Dk}{T_1}{K}$ if and only if $\haskind{\Dk}{T_2}{K}$.
\end{lemma}
\begin{proof}
  By structural induction on the subtyping judgement $\stp{\ctxdelta{\Dt}{\Dk}}{T_1}{T_2}$. The variable case uses $\wfdelta{\Dt}{\Dk}$, and the equivalence case is by \Cref{image:hkdr-thm-eqtype-preserve-kind}. Other cases follow from the induction hypothesis and \Cref{image:hkdr-thm-hkunique}.
\end{proof}

These syntactic lemmas ensure that the kinding system is consistent and preserved up to equivalence and subtyping casts,
and therefore interpretations based on kinding judgements are safe. %
\subsection{Logical Relation} \label{image:sec:hkdr-LR}

In this section, we prove weak normalization and semantic type soundness of \hkdomranlang using logical relations.
The design of the logical relations and interpretations mainly follows that of the \domranlang calculus (\Cref{domrange:sec:LR}),
with extensions to support higher-kind type operators.
A key difference is that our value type is interpreted on a kinding judgement $\haskind{\Dk}{T}{K}$ as a kind witness,
instead of on types alone as in the interpretation of the \domranlang calculus, which does not scale to type operators.

\subsubsection{Syntax and Definitions} \label{image:sec:hkdr-LR-syntax}

\begin{figure}[t]\small
\begin{mdframed}
\begin{minipage}[t]{1.0\textwidth}
\judgement{Syntax and Definitions of Logical Relations}{\BOX{\hkdomranlang}}\small %
  \[\begin{array}{l@{\qquad}l@{\qquad}l@{\qquad}l}
    v       & ::= & c \mid \vabs{H}{t} \mid \vtabs{H}{t}   & \text{Values}      \\ [2ex]

    \valkind[\KTpe] & ::= & \{ \valsl{v}{s}  \}   & \text{Value Type of \KTpe} \\[0.3ex]
    \valkind[\KArr{K_1}{K_2}]   & ::= & \valkind[K_1] \Rightarrow \valkind[K_2]   & \text{Value Type of Higher Kind} \\ 
    \VVar[]{V}{K}      & \in & \valkind[K]  & \text{Value Type Variables}    \\ [2ex] 

    H       & ::= & \varnothing \mid H,\, (x,\,v)   & \text{Value Environment} \\
    \rho & ::= & \varnothing \mid \rho,\, (X,\, \VVar[]{V}{K})   & \text{Value Type Environment} \\[2ex]

  \end{array}\]    
    
\judgement{\pathselection{Selection Path}}{}\small %
  \[\begin{array}{l@{\qquad}l@{\qquad}l@{\qquad}l@{\qquad}l}
    \pathselection{sl}      & ::= &  \pathselection{\sldom} \mid \pathselection{\slrange}  &  &  \text{Selector} \\
    \pathselection{s}     & ::= & \slroot \mid \pathselection{\slappend{sl}{s}} &  & \text{Path} \\[2ex]
    \VVar[s]{V}{K}    & ::= & \valtypeset[v]{ \valsl{v}{s} \in V } &  & \text{Value Type with Path}  \\[2ex]
  \end{array}\]

\end{minipage} 

\caption{The syntax and definitions of the logical relations for \hkdomranlang calculus.} \label{image:fig:hkdr-LR-syntax}
\end{mdframed}
\vspace{-2ex}
\end{figure} 
\Cref{image:fig:hkdr-LR-syntax} presents the syntax conventions and definitions of the logical relations for \hkdomranlang.
Values are defined as constants, term $\lambda$-abstractions, and type $\Lambda$-abstractions.
The abstractions are paired with a value environment $H$, which stores value bindings for term variables whose bodies have not yet been reduced.
The big-step semantics is identical to that of the base \domranlang calculus (Figure 5); higher kinds are erased at runtime, so we do not repeat it here.

The interpretations of higher kinds are operators from value types to value types, reflecting type operators.
We therefore require value types to take kind $\kind{K}$ as a parameter, denoted $\VVar[]{V}{K}$.

By specifying a selection path, a value type is reduced to a set of values, and we use $\VVar[s]{V}{K}$ to denote a value type $V$ of kind $\kind{K}$ specified for selection path $\pathselection{s}$.
With only base kind $\kind{\KTpe}$ and root path $\slroot$, the interpretation degenerates to that of System \Fsub without projection types and kinds.

We use $H$ to denote the value environment for term reduction, and $\rho$ for the value type environment, where type variables bind their corresponding value types of kind $\kind{K}$ as recorded in the $\Dk$ context.

\subsubsection{Semantic Interpretation of Equivalence and Subtyping} \label{image:sec:hkdr-LR-vstp}

\begin{figure}[t]\small
\begin{mdframed}
\begin{minipage}[t]{1.0\textwidth}
\vspace{0.5ex}
\judgement{Path Positivity}{}\small

  \[\begin{array}{l@{\qquad}l@{\qquad}l@{\qquad}l@{\qquad}r}
    \possl{\slroot} = \text{True}    & \possl{(\slappend{\sldom}{s})} = \neg\, \possl{s} & \possl{(\slappend{\slrange}{s})} = \possl{s}  & \text{Path Positivity} \\[1ex]
  \end{array}\]    

  \[
  \slplus{s} \,\, \defeq \,\, \{ \pathselection{s} \mid \possl{s} \} \quad \text{Positive Paths} \qquad\quad
  \slminus{s} \,\, \defeq \,\, \{ \pathselection{s} \mid \neg\, \possl{s} \} \quad \text{Negative Paths}
  \]\vspace{2ex}

\judgement{Value Type Equivalence and Subtyping}{\BOX{\hkveq{\VVar[]{V}{K}}{\VVar[]{V}{K}}{K}} \BOX{\hkvstp{\VVar[]{V}{K}}{\VVar[]{V}{K}}{K}}}\small

  \[\begin{array}{l@{\quad}l@{\quad}l@{ \qquad}r}
    \hkveq{\VVar[]{V_1}{\KTpe}}{\VVar[]{V_2}{\KTpe}}{\KTpe} &\defeq& \forall \pathselection{s},\, \VVar[s]{V_1}{\KTpe} = \VVar[s]{V_2}{\KTpe} & \text{Eq-Interpretation of $\KTpe$} \\
    \hkveq{\VVar[]{V_1}{\KArr{K_1}{K_2}}}{\VVar[]{V_2}{\KArr{K_1}{K_2}}}{\KArr{K_1}{K_2}} &\defeq& \forall \VVar[]{V_x}{K_1},\, \hkveq{\valtypeapp{\VVar[]{V_1}{\KArr{K_1}{K_2}}}{\VVar[]{V_x}{K_1}}}{\valtypeapp{\VVar[]{V_2}{\KArr{K_1}{K_2}}}{\VVar[]{V_x}{K_1}}}{K_2} & \text{ of Higher Kind} \\[2ex]

    \hkvstp{\VVar[]{V_1}{\KTpe}}{\VVar[]{V_2}{\KTpe}}{\KTpe} &\defeq& \forall \pathselection{s},\, \VVar[\slplus{s}]{V_1}{\KTpe} \subseteq \VVar[\slplus{s}]{V_2}{\KTpe} \, \land\, \VVar[\slminus{s}]{V_2}{\KTpe} \subseteq \VVar[\slminus{s}]{V_1}{\KTpe} & \text{Sub-Interpretation of $\KTpe$}  \\
    \hkvstp{\VVar[]{V_1}{\KArr{K_1}{K_2}}}{\VVar[]{V_2}{\KArr{K_1}{K_2}}}{\KArr{K_1}{K_2}} &\defeq& \forall \VVar[]{V_x}{K_1},\, \hkvstp{\valtypeapp{\VVar[]{V_1}{\KArr{K_1}{K_2}}}{\VVar[]{V_x}{K_1}}}{\valtypeapp{\VVar[]{V_2}{\KArr{K_1}{K_2}}}{\VVar[]{V_x}{K_1}}}{K_2} & \text{ of Higher Kind} \\[2ex]
  \end{array}\]

\end{minipage} 
\caption{The definitions of value interpretations for \hkdomranlang calculus.} \label{image:fig:hkdr-path-veq-vstp}
\end{mdframed}
\vspace{-2ex}
\end{figure} 
\Cref{image:fig:hkdr-path-veq-vstp} defines path positivity and the path-aware semantic interpretations of equivalence and subtyping.
Path positivity measures the polarity of variance with respect to the remaining delayed path selectors.
As the domain type is contravariant, the path selector $\sldom$ accordingly inverts the variance of value types, as reflected by the congruence subtyping rules.
For example, if $T_1 <: T_2$, then $\TDom{T_2} <: \TDom{T_1}$ by congruence,
and semantically $\withpath{\llbracket \TDom{T_2} \rrbracket}{\slroot} \subseteq \withpath{\llbracket \TDom{T_1} \rrbracket}{\slroot}$,
which is delayed by the selector to $\withpath{\llbracket T_2 \rrbracket}{\slappend{\sldom}{\slroot}} \subseteq \withpath{\llbracket T_1 \rrbracket}{\slappend{\sldom}{\slroot}}$.
Notably, the subset direction is reversed due to the domain selector, so the variance is path-sensitive.

The semantic interpretations of equivalence and subtyping capture path positivity, and are defined for all selection paths.
In the $\kind{\KTpe}$ case, the semantic interpretation of equivalence is defined by set equality of two value types specified with the same arbitrary path,
while that of subtyping is defined by subset inclusion of two value types according to the variance of the specified path positivity.
The higher-order $\kind{\KArr{K_1}{K_2}}$ cases are recursively defined on their kinds, encoded as functions mapping from a value type of kind $\kind{K_1}$ to a value type of kind $\kind{K_2}$;
we establish their relations via the equivalence or subtyping after applying them to the same argument value type of kind $\kind{K_1}$.
This definition strategy is standard for higher-order mappings.

\subsubsection{Semantic Interpretations} \label{image:sec:hkdr-LR-interpretation}

\begin{figure}[t]\small
\begin{mdframed}
\begin{minipage}[t]{1.0\textwidth}
\judgement{Value Interpretations}{\BOX{\hkvaltype[s]{\haskind{\Dk}{T}{K}}{\rho} \in \valkind[K] \text{ where }\envkv{\Dk}{\rho}}}\small %

  \[\begin{array}{l@{\qquad}l@{\qquad}l}
    \hkvaltype[\slroot]{\haskind{\Dk}{B}{\KTpe}}{\rho} & = & \{ c \}   \\ [1ex]

    \hkvaltype[\slplus{s}]{\haskind{\Dk}{\TTop}{\KTpe}}{\rho}  & = & \{ v \} \\[0.3em]
    \hkvaltype[\slminus{s}]{\haskind{\Dk}{\TBot}{\KTpe}}{\rho} & = & \{ v \}  \\[1em]

    \hkvaltype[s]{\haskind{\Dk}{X}{K}}{\rho} & = & \hkwithpathkind{\rho(X)}{s}{K}  \\[1ex]

    \hkvaltype[\slroot]{\haskind{\Dk}{\TAll[X]{\VVar[]{U}{K}}{T}}{\KTpe}}{\rho} & = & \valtypeset[\vtabs{H}{t}]{
      \forall \VVar[]{V}{K},\, \likeFunctionType{\VVar[]{V}{K}} \to \\
      & & \qquad \quad \hkvstp{\VVar[]{V}{K}}{\hkvaltype[]{\haskind{\Dk}{\VVar[]{U}{K}}{K}}{\rho}}{K} \to \\
      & & \qquad \quad  \exists v,\,\tevaln{H}{t}{v}\, \land \\
      & & \qquad \quad v \in \hkvaltype[\slroot]{\haskind{\Dkext[X]{\Dk}{K}}{T}{\KTpe}}{\rho,\,(X,\VVar[]{V}{K})}
    } \\[1ex]

    \hkvaltype[\slroot]{\haskind{\Dk}{\TFun{T_1}{T_2}}{\KTpe}}{\rho} & = & \valtypeset[]{
      \exists H,\, t,\, v = \vabs{H}{t} \land \\
      & & \qquad \quad \forall v_1,\, v_1 \in \hkvaltype[\slroot]{\haskind{\Dk}{T_1}{\KTpe}}{\rho} \to \\ 
      & & \qquad \quad \exists v_2,\,\tevaln{(H,\,(x,v_1))}{t}{v_2} \land \\
      & & \qquad \quad v_2 \in \hkvaltype[\slroot]{\haskind{\Dk}{T_2}{\KTpe}}{\rho} }\\[0.3ex]
    \hkvaltype[\slappend{\sldom}{s}]{\haskind{\Dk}{\TFun{T_1}{T_2}}{\KTpe}}{\rho} & = & \hkvaltype[s]{\haskind{\Dk}{T_1}{\KTpe}}{\rho} \\[0.3ex]
    \hkvaltype[\slappend{\slrange}{s}]{\haskind{\Dk}{\TFun{T_1}{T_2}}{\KTpe}}{\rho} & = & \hkvaltype[s]{\haskind{\Dk}{T_2}{\KTpe}}{\rho} \\[1ex]

    \hkvaltype[s]{\haskind{\Dk}{\TDom{T}}{\KTpe}}{\rho} & = & \hkvaltype[\slappend{\sldom}{s}]{\haskind{\Dk}{T}{\KTpe}}{\rho} \\[0.3ex]
    \hkvaltype[s]{\haskind{\Dk}{\TRange{T}}{\KTpe}}{\rho} & = & \hkvaltype[\slappend{\slrange}{s}]{\haskind{\Dk}{T}{\KTpe}}{\rho} \\ [1ex]    

    \hkvaltype[s]{\haskind{\Dk}{\TAbs[X]{K_1}{T}}{\KArr{K_1}{K_2}}}{\rho} & = & \VVar[]{V}{K_1} \Rightarrow \hkvaltype[s]{\haskind{\Dkext[X]{\Dk}{K_1}}{T}{K_2}}{\rho,\,(X,\VVar[]{V}{K_1})}   \\[0.3ex]
    \hkvaltype[s]{\haskind{\Dk}{\TTApp{T_1}{T_2}}{K_2}}{\rho} & = & \withpath{(\valtypeapp{\hkvaltype[]{\haskind{\Dk}{T_1}{\KArr{K_1}{K_2}}}{\rho}}{\hkvaltype[]{\haskind{\Dk}{T_2}{K_1}}{\rho}})}{s} \\[1ex]

    \text{otherwise} & = & \varnothing  \\[1ex]

    \end{array}\]      

\end{minipage} 

\caption{The definitions of value interpretations for \hkdomranlang calculus.} \label{image:fig:hkdr-LR-valtype}
\end{mdframed}
\vspace{-2ex}
\end{figure} 
\Cref{image:fig:hkdr-LR-valtype} defines the semantic interpretation of the \hkdomranlang calculus, which is the key part of the logical relation.
We write $\hkvaltype[s]{\haskind{\Dk}{T}{K}}{\rho}$ for the interpretation of a kinding judgement $\haskind{\Dk}{T}{K}$ under a semantic environment $\rho$,
which computes a value type of kind $\kind{K}$ optionally specified on a selection path $\pathselection{s}$.
Notably, the interpretation takes a kinding judgement instead of just a type, in order to provide a kind witness as well as well-kindedness to exclude ill-formed types.

Interpretations for types of base kind $\kind{\KTpe}$ are actual types for terms,
so we directly define the possible values inhabiting the type.
The interpretation of base types is a set of value constants at root path $\slroot$;
that of top types is all values at positive paths;
conversely, that of bottom types is all values at negative paths;
and the interpretation of type variables simply reads from the $\rho$ context.

The interpretation of a bounded quantification states that for an actual value type instantiation $\VVar[]{V}{K}$ that is bounded by the interpretation of the type bound $U$ up to polarity,
we can interpret the body type with $\VVar[]{V}{K}$ appended to the semantic context.
As the interpretation of the instantiated type $\VVar[]{V}{K}$ should respect well-formedness,
$\likeFunctionType{\VVar[]{V}{K}}$ is a carried property that holds for the interpretation of every well-kinded type.

The interpretation of an arrow type depends on the selection path:
if the path is non-root, we simply interpret the corresponding domain or range according to the top selector.
If the path is root (\slroot), then the arrow type $\TFun{T_1}{T_2}$ actually types $\lambda$-abstractions: for any value $v_1$ as the actual argument inhabiting the interpretation of $T_1$,
the body $t$ reduces to a value $v_2$ that inhabits the interpretation of $T_2$.
The interpretation of domain and range types simply appends selectors to the current selection path,
delaying the interpretation until the type $T$ is known to be an arrow type.

Finally, the interpretation of higher-kind $\kind{\KArr{K_1}{K_2}}$ type operators derives a map from value types of kind $\kind{K_1}$ to a value type of kind $\kind{K_2}$.
The type-level application is then interpreted straightforwardly by applying the map to actual value types.
The selection path is specified after all higher-kind abstractions or applications are interpreted,
as selection paths are for projections on types of kind $\kind{\KTpe}$.

\subsubsection{Environment and Invariants} \label{image:sec:hkdr-LR-env}

\begin{figure}\small
\begin{mdframed}
\judgement{Well-Formed Functionality}{\BOX{\likeFunctionType{\VVar[]{V}{K}}}}%
  
\[\begin{array}{l@{\quad}l@{\quad}l@{}l}
  \likeFunctionType{\VVar[]{V}{\KTpe}} &\defeq& \forall v,\,\pathselection{s},\,
      v \in \VVar[s]{V}{\KTpe} \to \forall v_1,\, v_1 \in \VVar[\slprepend{s}{\sldom}]{V}{\KTpe} \to \exists H,\,t,\, v = \vabs[]{H}{t} \,\land \\
    & & \exists v_2,\, \tevaln{(H,\,(x,v_1))}{t}{v_2} \land v_2 \in \VVar[\slprepend{s}{\slrange}]{V}{\KTpe} \\
  \likeFunctionType{\VVar[]{V}{\KArr{K_1}{K_2}}} &\defeq& \forall \VVar[]{V_x}{K_1},\, \likeFunctionType{\VVar[]{V_x}{K_1}} \rightarrow \likeFunctionType{\valtypeapp{\VVar[]{V}{\KArr{K_1}{K_2}}}{\VVar[]{V_x}{K_1}}} \\[2ex] 
\end{array}\]
\vspace{2ex}

\judgement{Environment Well-Formedness}{\BOX{\wfdelta{\Dt}{\Dk}} \BOX{\envkv{\Dk}{\rho}}}%
\[\begin{array}{l@{\quad}l@{\quad}l@{}l}
  \wfdelta{\Dt}{\Dk} &\defeq& \lnorm{\Dt} = \lnorm{\Dk} \land \forall X,\,T,\, X <: T \in \Dt \to \exists \kind{K},\, X\hksymbol\kind{K}\in \Dk \land \haskind{\Dk}{T}{K}    \\ [1ex]
  \envkv{\Dk}{\rho} &\defeq& \forall X,\,\kind{K},\,X\hksymbol\kind{K} \in \Dk \to (\rho(X) \in \valkind[K] \land \likeFunctionType{\rho(X)})  \\[2ex]
\end{array}\]

\judgement{Environment and Expression Interpretation}{\BOX{\envtype{\Gamma}} \BOX{\exptype{T}}}
\[\begin{array}{l@{\quad}l@{\quad}l@{}l}
  \envtype{\ctx{\Gamma}{\Dk}{\Dt}} &\defeq& \envtypeset[]{\DOM \Gamma = \DOM H \land \DOM \rho = \DOM \Dk = \DOM \Dt\, \land \\
    & & \qquad\qquad \envkv{\Dk}{\rho}\, \land \wfdelta{\Dt}{\Dk}\,\land  \\
    & & \qquad\qquad (\forall x,\, T,\, x : T \in \Gamma \to \exists v,\, H(x) = v \land v \in \hkvaltype[\slroot]{\haskind{\Dk}{T}{\KTpe}}{\rho})\, \land \\
    & & \qquad\qquad (\forall X,\, U,\, X <: U \in \Dt \to \exists \VVar[]{V}{K},\,\kind{K},\, X\hksymbol \kind{K}\in\Dk \land \rho(X) = \VVar[]{V}{K}\, \land \\
    & & \qquad\qquad\quad \hkvstp{\hkvaltype[]{\haskind{\Dk}{X}{K}}{\rho}}{\hkvaltype[]{\haskind{\Dk}{U}{K}}{\rho}}{K} \land \likeFunctionType{\VVar[]{V}{K}} )
  } \\[1ex] 
  \exptype{T} &\defeq& \{ \langle H,\Dk,\rho, t \rangle \mid \envkv{\Dk}{\rho} \land \exists v,\, \tevaln{H}{t}{v} \land v \in \hkvaltype[\slroot]{\haskind{\Dk}{T}{\KTpe}}{\rho}   \} \\[2ex] 
\end{array}\]

\vspace{1pt}

\caption{The environment and expression interpretation of the \hkdomranlang-calculus. }\label{image:fig:hkdr-LR-env}
\end{mdframed}
\vspace{-2ex}
\end{figure} 
With the semantic interpretation defined, we define the semantic interpretation of environments and invariants in \Cref{image:fig:hkdr-LR-env}.

\paragraph{Well-Formed Functionality.}
The well-formed functionality, $\likeFunctionType{\VVar[]{V}{\KTpe}}$, describes the property that the interpretation of every well-kinded type can potentially serve like that of an arrow type.
Namely, if we can find a value $v_x$ inhabiting the domain of $\VVar[s]{V}{K}$, then given $v$ inhabiting $\VVar[s]{V}{K}$,
$v$ must be a $\lambda$-abstraction and the return value from applying it to $v_x$ inhabits the range of $\VVar[s]{V}{K}$.
Notably, the domain and range selectors are prepended instead of appended, meaning that regardless of previous projections,
the value type works like that of a function as long as we can find a domain witness after all existing projections are resolved.
The higher-kind cases are defined similarly, with the extra premise that the argument value type also respects the property.

\paragraph{Environment Invariants.}
Recall that for the simplicity of kinding judgements and semantic interpretations,
we split the type variable $\Delta$ context into two separate contexts: $X <: U \in \Dt$ for type bounds, and $X \hksymbol\kind{K} \in \Dk$ for type kinds.
The invariant $\wfdelta{\Dt}{\Dk}$ is the syntactic invariant ensuring that $\Dk$ actually describes the kinds of type bounds in $\Dt$.
Similarly, the invariant $\envkv{\Dk}{\rho}$ ensures that the value types in the semantic environment $\rho$ have the same kinds as described in $\Dk$.

\paragraph{Environment Interpretations.}
The syntactic typing contexts for term variables $\Gamma$ and type variables $\ctxdelta{\Dt}{\Dk}$ are interpreted as semantic contexts $H$ and $\rho$, respectively.
$H$ is a context of values for term variable bindings with the same domain as $\Gamma$, derived from term reduction, and guarantees that each value stored in $H$ inhabits the interpretation of the corresponding $T$ in $\Gamma$.
$\rho$ is a context of value types for type variable bindings, so its domain equals both the $\Dk$ and $\Dt$ contexts.
In addition to the environment invariants, the value types in $\rho$ should be proper bounds, \ie, there should exist some well-formed value types bounded by them.

\begin{figure}[ht]\small
\begin{mdframed}
\judgement{Semantic Interpretation of Judgements}{\BOX{\semeq{\Dk}{T}{T}} \BOX{\semstp{\ctxdelta{\Dt}{\Dk}}{T}{T}} \BOX{\semtype{\ctx{\Gamma}{\Dt}{\Dk}}{t}{T}}}%
  
\[\begin{array}{l@{\quad}l@{\quad}l@{}l}
  \semeq{\Dk}{T_1}{T_2} &\defeq& \forall \kind{K},\, \envs{.}{\rho} \in \envtype{\ctx{.}{.}{\Dk}} ,\, \hkveq{\hkvaltype[]{\haskind{\Dk}{T_1}{K}}{\rho}}{\hkvaltype[]{\haskind{\Dk}{T_2}{K}}{\rho}}{K} \\
  \semstp{\ctxdelta{\Dt}{\Dk}}{T_1}{T_2} &\defeq& \forall \kind{K},\, \envs{.}{\rho} \in \envtype{\ctx{.}{\Dt}{\Dk}},\, \hkvstp{\hkvaltype[]{\haskind{\Dk}{T_1}{K}}{\rho}}{\hkvaltype[]{\haskind{\Dk}{T_2}{K}}{\rho}}{K} \\[1ex]
  \semtype{\ctx{\Gamma}{\Dt}{\Dk}}{t}{T} &\defeq& \forall \envs{H}{\rho} \in \envtype{\ctx{\Gamma}{\Dt}{\Dk}},\, \exps{H,\Dk}{\rho}{t} \in \exptype{T}
\end{array}\]

\vspace{2ex}
\judgement{Semantic Interpretation of Typing (Selected)}{\BOX{\semtype{\ctx{\Gamma}{\Dt}{\Dk}}{t}{T}}} \\
\typicallabel{stsub}
\begin{tabular}{b{.48\linewidth}@{}b{.5\linewidth}}
  \infrule[\ruledef{image:st1sub}{st-sub}]{
    \semtype{\ctx{\Gamma}{\Dt}{\Dk}}{t}{T_1} \\
    \semstp{\ctxdelta{\Dt}{\Dk}}{T_1}{T_2}
  }{
    \semtype{\ctx{\Gamma}{\Dt}{\Dk}}{t}{T_2}
  }
  &
  \infrule[\ruledef{image:st1appdr}{st-appdr}]{
    \semtype{\ctx{\Gamma}{\Dt}{\Dk}}{t_1}{F} \\
    \semtype{\ctx{\Gamma}{\Dt}{\Dk}}{t_2}{\TDom{F}}
  }{
    \semtype{\ctx{\Gamma}{\Dt}{\Dk}}{t_1~t_2}{\TRange{F}}
  }
\end{tabular} \\[1.5em]

\vspace{1pt}

\caption{The semantic interpretation of equivalence, subtyping, and typing (selected) of the \hkdomranlang-calculus. }\label{image:fig:hkdr-LR-semtype}
\end{mdframed}
\vspace{-2ex}
\end{figure} 
\paragraph{Expression Interpretations and Semantic Interpretation of Typing.}
The interpretation is expressed by the existence of some $v$ that inhabits the semantic interpretation of the type $T$ under some environments.
Notably, the existence of such $v$ implicitly entails the well-kindedness $\haskind{\Dk}{T}{\KTpe}$ of the type $T$.
The semantic interpretations of equivalence, subtyping, and typing in \Cref{image:fig:hkdr-LR-semtype} are defined by assembling the environment and expression interpretations, with all definitions in place.
They model the interpretation of expressions under interpreted environments, corresponding to the syntactic typing rules, and thus the full presentation is omitted here.

\subsubsection{Metatheory} \label{image:sec:hkdr-LR-metatheory}

The soundness of \hkdomranlang is established via logical relations.
We present the fundamental theorem and key lemmas here; some details are reasonably omitted.

First, we present the invariant that every interpretation of a well-kinded type can serve like that of an arrow type.
\begin{lemma}[Functionality]
  If\phantom{n}$\haskind{\Dk}{T}{K}$, and $\envkv{\Dk}{\rho}$, then $\likeFunctionType{\hkvaltype[]{\haskind{\Dk}{T}{K}}{\rho}}$.
\end{lemma}
\begin{proof}
  By induction on the type $T$. We need a context invariant that every value type in $\rho$ respects functionality. The mechanism has a separate invariant for this without loss of generality.
\end{proof}

Then, we show the standard weakening and substitution.
\begin{lemma}[Semantic Weakening]
  If\phantom{n}$\envkv{\Dk}{\rho}$, and for some $\VVar[]{V}{K'} = \hkvaltype[]{\haskind{\Dk}{T'}{K'}}{\rho}$, then $\hkvaltype[]{\haskind{\Dk}{T}{K}}{\rho} = \hkvaltype[]{\haskind{\ctxext{\Dk}{X\hksymbol\kind{K'}}}{T}{K}}{\rho,\,(X,\VVar[]{V}{K'})}$.
\end{lemma}
\begin{lemma}[Semantic Substitution]
  If\phantom{n}$\envkv{\Dk}{\rho}$, and for some $\VVar[]{V}{K'} = \hkvaltype[]{\haskind{\Dk}{T'}{K'}}{\rho}$, then $\hkvaltype[]{\haskind{\ctxext{\Dk}{X\hksymbol \kind{K'}}}{T}{K}}{\rho,\,(X,\VVar[]{V}{K'})} = \hkvaltype[]{\haskind{\Dk}{\substty[X]{T'}{T}}{K}}{\rho}$.
\end{lemma}
\begin{proof}
  By induction on the type $T$ and applying the syntactic weakening and substitution lemmas.
\end{proof}

The heterogeneous equality in the semantic interpretation of subtyping is resolved by applying \Cref{image:hkdr-thm-hkunique}.
\begin{lemma}[Proof Irrelevance]
  If\phantom{n}$h_1$ and $h_2$ are proofs of $\haskind{\Dk}{T}{K}$, then $\hkvaltype[]{h_1}{\rho} = \hkvaltype[]{h_2}{\rho}$.
\end{lemma}

Finally, the fundamental theorems of equivalence, subtyping, and typing.

\begin{theorem}[Fundamental of Equivalence] \label{image:thm:hkdr-veq}
  If\phantom{n}$\eqtype{\Dk}{T_1}{T_2}$, and for some $\kind{K}$, $\haskind{\Dk}{T_1}{K}$ and $\envkv{\Dk}{\rho}$, then $\hkveq{\hkvaltype[]{\haskind{\Dk}{T_1}{K}}{\rho}}{\hkvaltype[]{\haskind{\Dk}{T_2}{K}}{\rho}}{K}$.
\end{theorem}
\begin{proof}
  By induction on $\eqtype{\Dk}{T_1}{T_2}$, and for each case, dependent destruction on the selection path. For the root path, prove by the semantic interpretation; for non-root paths, match on the positivity of the path.
  Notably, this theorem requires function extensionality; otherwise we would need to state the theorem using $\leftrightarrow$ instead of equality, which requires proving four subcases.
\end{proof}

\begin{theorem}[Fundamental of Subtyping] \label{image:thm:hkdr-vstp}
  If\phantom{n}$\hkstp{\Dt}{\Dk}{T_1}{T_2}$, $\wfdelta{\Dt}{\Dk}$, $\envkv{\Dk}{\rho}$, and $\rho$ respects the bounds of $\Dt$, \ie, $\hkvstp{\hkvaltype[]{\haskind{\Dk}{X}{K'}}{\rho}}{\hkvaltype[]{\haskind{\Dk}{U}{K'}}{\rho}}{K'}$ for every $X <: U \in \Dt$ with $X \hksymbol \kind{K'} \in \Dk$, then for every $\kind{K}$ with $\haskind{\Dk}{T_1}{K}$, $\hkvstp{\hkvaltype[]{\haskind{\Dk}{T_1}{K}}{\rho}}{\hkvaltype[]{\haskind{\Dk}{T_2}{K}}{\rho}}{K}$.
\end{theorem}
\begin{proof}
  By induction on $\hkstp{\Dt}{\Dk}{T_1}{T_2}$ following a similar strategy. The variable case is derived from invariants, and the equivalence case from \Cref{image:thm:hkdr-veq}. Notably, $\haskind{\Dk}{T_2}{K}$ holds directly by Lemma~\ref{image:hkdr-thm-stp-preserve-kind}.
\end{proof}

\begin{theorem}[Fundamental] \label{image:thm:hkdr-fundamental}
  If\phantom{n}$t$ is syntactically well-typed, $\hkhastype{\Gamma}{\Dt}{\Dk}{t}{T}$, then $t$ is also in the semantic interpretation of its type, $\semtype{\ctx{\Gamma}{\Dt}{\Dk}}{t}{T}$.
\end{theorem}
\begin{proof}
  By induction on $\hkhastype{\Gamma}{\Dt}{\Dk}{t}{T}$. In each case, derive the semantic interpretation with kind $\kind{\KTpe}$ and path $\slroot$. The subsumption case follows from the fundamental theorem of subtyping (\Cref{image:thm:hkdr-vstp}).
\end{proof}

The fundamental theorem establishes weak normalization and type soundness.
\begin{corollary}[Safety] \label{image:thm:hkdr-safety}
  If\phantom{n}$t$ is syntactically well-typed in an empty context, $\hkhastype{\varnothing}{\varnothing}{\varnothing}{t}{T}$, then $\tevaln{\varnothing}{t}{v}$ in finitely many steps, and $v \in \hkvaltype[\slroot]{\haskind{\varnothing}{T}{\KTpe}}{\varnothing}$.
\end{corollary}
\begin{proof}
  By applying the fundamental theorem (\Cref{image:thm:hkdr-fundamental}).
\end{proof} 

\section{Usage Analysis of Utility Types} \label{domrange:appendix:usage}

This appendix substantiates the claim of \Cref{domrange:sec:overview-practice} that the utility types @Parameters<T>@ and @ReturnType<T>@ are used in real TypeScript code, and that programmers reach for the unsafe bound when the projected type is abstract.

\paragraph{Use of AI\@.}
Both studies in this section, the survey of the 100 repositories below and the case study of Rocket.Chat in \Cref{domrange:appendix:rocketchat}, were designed and carried out by an AI agent, as stated in the disclosure of the use of AI in the paper.
The counts, the classifications, and the selection of code snippets are collected by AI.
We disclose the methodology of the research, however, the workflow has not been end-to-end checked by authors.
For these reasons, we do not suggest taking these data as a trusted source without further validation. 

\paragraph{Methodology.}
In March 2026, the AI agent retrieved the 100 most-starred TypeScript repositories through GitHub's repository search, and queried GitHub's code search within each repository for the tokens @Parameters<@ and @ReturnType<@, keeping the returned code fragments that contain the exact token.
The code-search API returns at most 30 matching files per repository and query, each as a set of short fragments, so every count below is a lower bound.
The complete per-repository counts are listed in \Cref{domrange:tab:usage-repos}, and \Cref{domrange:tab:usage} summarizes the distribution.

\paragraph{Rate of Occurrence.}
@ReturnType@ appears in 88 of the 100 repositories, with 1{,}782 matched snippets across 1{,}470 files, and @Parameters@ appears in 82 repositories, with 1{,}155 snippets across 924 files.
Overall, 90 repositories use at least one of the two and 80 use both.
The adopters span every project category in the sample, from IDEs (VSCode \cite{MicrosoftVscode2026}, the TypeScript compiler itself) through frameworks (Nuxt, Astro, Vite) to applications (cal.com, Rocket.Chat \cite{RocketChatRocketChat2026}), while the repositories with no matches are mostly non-TypeScript or content-only projects.

\begin{table}
  \caption{Snippets containing a projection, per repository, over the 100 most-starred TypeScript repositories on GitHub. Counts are lower bounds due to search-API caps.}
  \label{domrange:tab:usage}
  \begin{tabular}{lcccc}
    \toprule
    Snippets per repository & 0 & 1--10 & 11--30 & over 30 \\
    \midrule
    @ReturnType@ & 12 & 30 & 26 & 32 \\
    @Parameters@ & 18 & 45 & 24 & 13 \\
    \bottomrule
  \end{tabular}
\end{table}

\begin{table}
  \caption{Matched snippets per repository (R for \texttt{ReturnType}, P for \texttt{Parameters}), in decreasing order of stars, left column first.}
  \label{domrange:tab:usage-repos}
  \footnotesize
  \setlength{\tabcolsep}{3.5pt}
  \begin{tabular}{lrr@{\hspace{2em}}lrr}
    \toprule
    Repository & R & P & Repository & R & P \\
    \midrule
    \texttt{freeCodeCamp/freeCodeCamp} & 16 & 0 & \texttt{remix-run/react-router} & 25 & 10 \\
    \texttt{kamranahmedse/developer-roadmap} & 2 & 0 & \texttt{appwrite/appwrite} & 0 & 0 \\
    \texttt{openclaw/openclaw} & 42 & 41 & \texttt{laurent22/joplin} & 14 & 1 \\
    \texttt{vuejs/vue} & 0 & 0 & \texttt{vuejs/core} & 11 & 3 \\
    \texttt{microsoft/vscode} & 35 & 15 & \texttt{ionic-team/ionic-framework} & 6 & 0 \\
    \texttt{n8n-io/n8n} & 36 & 29 & \texttt{DefinitelyTyped/DefinitelyTyped} & 40 & 36 \\
    \texttt{yangshun/tech-interview-handbook} & 0 & 4 & \texttt{FlowiseAI/Flowise} & 14 & 1 \\
    \texttt{langgenius/dify} & 35 & 19 & \texttt{cypress-io/cypress} & 37 & 41 \\
    \texttt{anomalyco/opencode} & 35 & 22 & \texttt{lyswhut/lx-music-desktop} & 2 & 1 \\
    \texttt{excalidraw/excalidraw} & 27 & 8 & \texttt{upstash/context7} & 1 & 1 \\
    \texttt{iptv-org/iptv} & 0 & 0 & \texttt{TanStack/query} & 12 & 16 \\
    \texttt{shadcn-ui/ui} & 17 & 7 & \texttt{expo/expo} & 34 & 33 \\
    \texttt{microsoft/TypeScript} & 36 & 40 & \texttt{type-challenges/type-challenges} & 1 & 1 \\
    \texttt{clash-verge-rev/clash-verge-rev} & 9 & 2 & \texttt{slab/quill} & 2 & 2 \\
    \texttt{angular/angular} & 22 & 15 & \texttt{pixijs/pixijs} & 1 & 6 \\
    \texttt{supabase/supabase} & 32 & 19 & \texttt{makeplane/plane} & 5 & 5 \\
    \texttt{google-gemini/gemini-cli} & 40 & 7 & \texttt{OpenCut-app/OpenCut} & 3 & 6 \\
    \texttt{ant-design/ant-design} & 8 & 8 & \texttt{tldraw/tldraw} & 38 & 9 \\
    \texttt{immich-app/immich} & 12 & 6 & \texttt{prisma/prisma} & 23 & 4 \\
    \texttt{tailwindlabs/tailwindcss} & 13 & 3 & \texttt{jestjs/jest} & 20 & 12 \\
    \texttt{puppeteer/puppeteer} & 27 & 0 & \texttt{RocketChat/Rocket.Chat} & 39 & 48 \\
    \texttt{firecrawl/firecrawl} & 25 & 9 & \texttt{slidevjs/slidev} & 12 & 1 \\
    \texttt{storybookjs/storybook} & 37 & 32 & \texttt{vercel/hyper} & 7 & 3 \\
    \texttt{ChatGPTNextWeb/NextChat} & 5 & 0 & \texttt{react-hook-form/react-hook-form} & 1 & 2 \\
    \texttt{mermaid-js/mermaid} & 16 & 2 & \texttt{upscayl/upscayl} & 3 & 0 \\
    \texttt{microsoft/playwright} & 18 & 18 & \texttt{streamich/react-use} & 7 & 6 \\
    \texttt{realworld-apps/realworld} & 0 & 0 & \texttt{babel/babel} & 10 & 24 \\
    \texttt{modelcontextprotocol/servers} & 0 & 1 & \texttt{AykutSarac/jsoncrack.com} & 0 & 0 \\
    \texttt{vitejs/vite} & 12 & 7 & \texttt{FuelLabs/fuels-ts} & 11 & 5 \\
    \texttt{hoppscotch/hoppscotch} & 39 & 4 & \texttt{DIYgod/RSSHub} & 8 & 1 \\
    \texttt{coder/code-server} & 0 & 0 & \texttt{colinhacks/zod} & 4 & 3 \\
    \texttt{Stirling-Tools/Stirling-PDF} & 19 & 2 & \texttt{fastapi/full-stack-fastapi-template} & 0 & 0 \\
    \texttt{nestjs/nest} & 6 & 7 & \texttt{siyuan-note/siyuan} & 0 & 0 \\
    \texttt{lobehub/lobehub} & 34 & 28 & \texttt{CherryHQ/cherry-studio} & 38 & 17 \\
    \texttt{grafana/grafana} & 33 & 32 & \texttt{hexojs/hexo} & 31 & 29 \\
    \texttt{abi/screenshot-to-code} & 0 & 0 & \texttt{payloadcms/payload} & 38 & 35 \\
    \texttt{strapi/strapi} & 33 & 29 & \texttt{janhq/jan} & 4 & 1 \\
    \texttt{apache/superset} & 22 & 25 & \texttt{styled-components/styled-components} & 5 & 0 \\
    \texttt{Eugeny/tabby} & 0 & 0 & \texttt{vuetifyjs/vuetify} & 8 & 2 \\
    \texttt{toeverything/AFFiNE} & 35 & 38 & \texttt{calcom/cal.com} & 40 & 38 \\
    \texttt{apache/echarts} & 40 & 20 & \texttt{twentyhq/twenty} & 40 & 32 \\
    \texttt{daytonaio/daytona} & 5 & 2 & \texttt{chakra-ui/chakra-ui} & 6 & 4 \\
    \texttt{facebook/docusaurus} & 8 & 23 & \texttt{code-yeongyu/oh-my-openagent} & 34 & 41 \\
    \texttt{socketio/socket.io} & 2 & 8 & \texttt{trpc/trpc} & 31 & 14 \\
    \texttt{nocodb/nocodb} & 38 & 18 & \texttt{remotion-dev/remotion} & 34 & 19 \\
    \texttt{reduxjs/redux} & 1 & 0 & \texttt{appsmithorg/appsmith} & 41 & 24 \\
    \texttt{nuxt/nuxt} & 34 & 11 & \texttt{mifi/lossless-cut} & 8 & 3 \\
    \texttt{cline/cline} & 21 & 3 & \texttt{chatboxai/chatbox} & 11 & 4 \\
    \texttt{withastro/astro} & 30 & 20 & \texttt{koala73/worldmonitor} & 36 & 12 \\
    \texttt{pmndrs/zustand} & 6 & 5 & \texttt{novuhq/novu} & 23 & 10 \\
    \bottomrule
  \end{tabular}
\end{table}

\paragraph{Projections over concrete types.}
Most snippets project a concrete type: 1{,}363 of the 1{,}782 @ReturnType@ snippets (76\%) and 624 of the 1{,}155 @Parameters@ snippets (54\%) apply the projection to a @typeof@, the type-level computation that Section 2.1 shows working smoothly.
The single most common target is the timer idiom @ReturnType<typeof setTimeout>@ and its @setInterval@ sibling from \Cref{domrange:sec:overview-practice}, with 224 snippets.
Other recurring concrete patterns, which may overlap, include type aliases read off a function as in @type X = ReturnType<typeof f>@ (304 snippets), asynchronous unwrapping via @Awaited<ReturnType<...>>@ (269), indexed access into a returned collection such as @ReturnType<...>[number]@ (117), and typing mocks as @ReturnType<typeof vi.fn>@ in test suites (47).

\paragraph{Projections over abstract function types.}
The remaining snippets project named function types and generic type variables, and here the data meets the abstract case of \Cref{domrange:sec:overview-practice}.
A bare identifier is the projection's argument, as in @ReturnType<T>@ or @Parameters<Fn>@, in 136 @ReturnType@ and 176 @Parameters@ snippets.
The recurring form is the reusable wrapper: debounce and throttle utilities appear independently in at least seven repositories, and the sample also contains fan-out event-handler merging (Chakra UI), safe invocation and signature-rewriting utilities (VSCode), observability wrappers (cal.com), and a middleware system (Rocket.Chat) that names the typed $\eta$-expansion itself:
\begin{lstlisting}
  type Next<T extends (...args: any[]) => any> =
    (...args: Parameters<T>) => ReturnType<T>;
  type Middleware<T extends (...args: any[]) => any> =
    (context: Parameters<T>, next: Next<T>) => ReturnType<T>;
\end{lstlisting}
Almost every definition of this family in the sample constrains its type variable by @extends (...args: any[]) => any@ or a variant, the unsafe arrow bound of \Cref{domrange:sec:overview-practice}.
Sound arrow bounds are rare: in the case study below, 3 of the 40 declarations of Rocket.Chat use one.
The data thus supports both halves of \Cref{domrange:sec:overview-practice}.
The projections are pervasive, and where they range over an abstract function type, production code mostly uses the unsafe bound.
These are the signatures that \domranlang aims to type soundly, with the $\eta$-expansion of \Cref{domrange:sec:introduction} as their minimal form.

\subsection{Case Study: Rocket.Chat} \label{domrange:appendix:rocketchat}

The counts in 100 repos say how often the projections range over an abstract function type, not what the code does with them.
To see whether such uses would type-check in \domranlang, and whether some depend on @any@ in a critical way, every one of them in a single repository, Rocket.Chat, was examined.
It has by far the most abstract uses among the four repositories sized for the study (Rocket.Chat, cal.com, Chakra UI, and VSCode), and its middleware is the one quoted above.

\paragraph{Use of AI in the case study.}
The case study was carried out in October 2026 by Claude, running as the coding agent Claude Code in the authors' working environment.
The agent designed and executed the study.
It cloned Rocket.Chat at commit \texttt{541804a}, wrote scripts against the TypeScript compiler API and ran them to obtain the counts below, read each of the 40 declarations and labelled it with a written rationale, and wrote small TypeScript programs, such as the probes below, whose acceptance or rejection it checked with the TypeScript 5.8.3 compiler.
Claude also drafted this section, specially including the selection and simplification of the code snippets.
The authors suggest the agent's first rubric with the classification used below (known and unknown function types; in and out of scope; sound with and without @any@), decided what to report, and revised the text.
The scripts and labels are not part of the artifact, as the research does not directly support or related to the contribution of this paper.

\paragraph{Method.}
The agent's scripts used the TypeScript 5.8.3 compiler API, syntax only, on all 9{,}932 TypeScript files of Rocket.Chat at commit \texttt{541804a} (October 2026), and sorted each of the 579 uses of @Parameters<X>@ or @ReturnType<X>@ by what @X@ is (\Cref{domrange:tab:rocketchat-uses}).

\begin{table}
  \caption{All uses of the projections in Rocket.Chat, by the projected type.}
  \label{domrange:tab:rocketchat-uses}
  \begin{tabular}{lcccc}
    \toprule
    & known function type & type parameter & keyed lookup & total \\
    \midrule
    \texttt{ReturnType} & 369 & 41 & 7 & 417 \\
    \texttt{Parameters} & 81 & 42 & 39 & 162 \\
    \bottomrule
  \end{tabular}
\end{table}

Two patterns must be told apart.
When @X@ is \emph{known} to be a function type, such as @typeof f@ for a concretely typed @f@, a named function type, a member of a concrete interface, or an entry of a concrete table that a mapped type iterates over, the projection is type-level computation.
It reduces by elimination, as in $\TDom{\TFun{A}{B}} \equiv A$, and needs nothing beyond first-class projections (\Cref{domrange:sec:overview-projection}).
These are the majority, 450 of the 579 uses.
TypeScript defines @Parameters@ and @ReturnType@ with the bound @(...args: any) => any@ in their signatures, so even these uses go through @any@.
For a known function type, however, an implementation without @any@ is possible: bounding by @(...args: never) => unknown@, of which every function type is a subtype since parameters are contravariant, and computing the projection by a conditional type agrees with the built-in definitions on known function types, which the agent checked with the TypeScript compiler.
The same bound fails only once the projected type is unknown and its function is applied, which is the rejected @logDecorator@ of \Cref{domrange:sec:introduction}.

When @X@ is \emph{unknown}, the function type is abstract, and the projection must be well-formed, and the function applicable, without exposing it as an arrow type.
This is what domain and range types and their application rule address.
In 83 uses @X@ is itself a type parameter, and these occur in 40 generic declarations, which the rest of this section examines.
In the other 46, @X@ selects a function type from a known table by an abstract key, as in @Parameters<ServerMethods[M]>@, mostly to type event emitters and remote calls.
They belong to the second pattern as well, but need indexed access by a type-level key, which \domranlang lacks, so we leave them aside.
Of the 40 declarations, 34 bound the type parameter by @(...args: any[]) => any@ or a variant mentioning @any@, 3 by a sound arrow type, and 3 by nothing at all.
All 40 thus belong to the second pattern: none projects a known function type, so none is sound merely by elimination on a known arrow type.

The agent read and classified each declaration, and two questions must be kept apart.
The first is what the declaration \emph{declares}: of the 40, 34 bound the type parameter with @any@, as counted above.
The second is what its code \emph{does} with the projections.
A declaration is \emph{sound} when it uses them correctly, so that \domranlang types it; \emph{unsound} when a cast or a typing rule of TypeScript hides a type error, so that its signature admits a call that TypeScript accepts and that fails at runtime; and \emph{relies on \texttt{any} only} when TypeScript needs @any@ to accept it although nothing goes wrong.
A declaration is \emph{in scope} when \domranlang can express its signature and body as written, assuming data that is Church-encodable (tuples, records, optional values, arrays, and maps) and ordinary library calls.

An @any@ in the bound does not by itself make a declaration unsound, and is often not even needed.
Programmers write it because TypeScript's own definitions of @Parameters@ and @ReturnType@ carry it, but it only matters where the body applies the abstract function.
To see whether it matters, the bound can be replaced by the sound @(...args: never) => unknown@, and the agent checked with the TypeScript compiler, on representative declarations, whether TypeScript still accepts them.
The middleware type above passes this test:
\begin{lstlisting}
  type Next<T extends (...args: never) => unknown> =
    (...args: Parameters<T>) => ReturnType<T>;     // still accepted
\end{lstlisting}
since @T@ appears only inside the projections and no value of type @T@ is called.
The decorator @pipe@, which occurs twice in the code base, fails it:
\begin{lstlisting}
  const pipe = <T extends (...args: any[]) => any>(fn: T) =>
    (...args: Parameters<T>): ReturnType<T> => fn(...args);
\end{lstlisting}
With the sound bound, TypeScript rejects @fn(...args)@ with the two errors of the rejected @logDecorator@ in \Cref{domrange:sec:introduction}.
We call the first kind sound \emph{without \texttt{any}} and the second sound \emph{with \texttt{any}}.
The second kind is what \domranlang adds, as it types @fn(...args)@ by \rulename{domrange:tappdomrange} without @any@.

The declarations outside the calculus need @this@ (8), asynchronous code (7), conditional types (3), callable objects (3), timers (1), or a self-referential bound (1).
For them, the classification assumes the missing feature with its standard subtyping behavior and asks whether the projections would be used soundly; this part is a best-effort approximation without any concrete reasoning.
\Cref{domrange:tab:rocketchat} gives the outcome on both questions.
The two axes differ: 11 of the 16 declarations that are sound without @any@ still declare an @any@ bound, which is redundant in them, and all 10 unsound declarations are among the 34 with an @any@ bound.

\begin{table}
  \caption{The 40 generic declarations of Rocket.Chat that project an abstract function type. Sound declarations are split by whether TypeScript accepts them under the sound bound \texttt{(...args: never) => unknown} (without \texttt{any}) or only through \texttt{any} (with \texttt{any}), unsound ones by the kind of error. The last column counts the declarations whose declared bound mentions \texttt{any}. The out-of-scope column assumes each missing feature with its standard subtyping behavior and is an approximation.}
  \label{domrange:tab:rocketchat}
  \begin{tabular}{llcccc}
    \toprule
    & & in scope & out of scope & total & \makecell{declared\\ \texttt{any} bound} \\
    \midrule
    sound & without \texttt{any} & 9 & 7 & 16 & 11 \\
          & with \texttt{any}    & 4 & 9 & 13 & 12 \\
    \multicolumn{2}{l}{relies on \texttt{any} only} & 0 & 1 & 1 & 1 \\
    unsound & cast back to $F$ (\texttt{as F})    & 1 & 3 & 4 & 4 \\
            & cast hiding a changed value          & 0 & 3 & 3 & 3 \\
            & argument read at the bound's type    & 3 & 0 & 3 & 3 \\
    \midrule
    \multicolumn{2}{l}{total} & 17 & 23 & 40 & 34 \\
    \bottomrule
  \end{tabular}
\end{table}

\paragraph{Decorators: the uses \domranlang solves.}
The 13 sound declarations in scope split into the two kinds above.
Nine are sound without @any@, like @Next<T>@ and @Middleware<T>@: they use the projections in types only, and in \domranlang the projections are well-formed over every type.
Four are sound with @any@: they apply the abstract function, three as decorators that forward their arguments, following the $\eta$-expansion, and one as a call of a bounded function variable.
The decorator @pipe@ above is the $\eta$-expansion of \Cref{domrange:sec:introduction}, and @logSubmit@, which logs its arguments before forwarding them, is @logDecorator@.
The body applies @fn : T@ to @args : Parameters<T>@, and TypeScript has no way to justify that call from a sound bound (\Cref{domrange:sec:overview-practice}).
Under the @any@ bound the call has type @any@, which is the only reason it is accepted at @ReturnType<T>@.
TypeScript's own definitions of @Parameters@ and @ReturnType@ require that bound, so the @any@ is forced on the programmer, not chosen.
In \domranlang the same body types with @T <: Top@: \rulename{domrange:tappdomrange} gives @fn(...args)@ the type $\TRange{T}$ from @args@ $: \TDom{T}$, and the call is safe for every instantiation by the witness property of \Cref{domrange:sec:boundless-whysound}.

The remaining sound declaration in scope uses a bounded function variable:
\begin{lstlisting}
  <F extends (key: ISetting['_id']) => any>(
      fn: F, key: ISetting['_id']): ReturnType<F> => {
    ...                                    // record an audit event for `key`
    return fn(key);
  }
\end{lstlisting}
Here @key@ is a string that the caller supplies separately, and the call is justified by the bound: from $F <: \TFun{\mathsf{String}}{\TTop}$, rules \rulename{domrange:sdomintro} and \rulename{domrange:sdom} give $\mathsf{String} <: \TDom{F}$, so \rulename{domrange:tappdomrange} types @fn(key)@ at $\TRange{F}$.
Expanding $F$ into $\TFun{A}{B}$ in \Fsub would require the lower bound $\mathsf{String} <: A$, which \Fsub cannot state (\Cref{domrange:sec:boundless-lowerbound}).
Both forms are covered by the mechanization: \textit{pipe} is the boundless $\eta$-expansion, typed at $\TAll[F]{\TTop}{\TFun{F}{\TFun{\TDom{F}}{\TRange{F}}}}$ (lemma \texttt{etaHasTypeAny} in \texttt{fsub\_domrange.v}), and the call @fn(key)@ takes the same step as the inner call $f\,\mathsf{true}$ of $g$ in \Cref{domrange:sec:boundless-lowerbound} (\texttt{lowerbound\_example.v}).

\paragraph{Returning the wrapper at the original type.}
Four declarations build a wrapper and return it \emph{as} the original function type: @cachedFunction@ in scope, and @withAsyncTOTP@, @use@, and @withMiddleware@ out of scope, for async functions, @this@, and callable objects, respectively.
The one in scope reads:
\begin{lstlisting}
  const cachedFunction = <F extends (...args: any[]) => any>(fn: F) => {
    const cache = new Map<string, unknown>();
    return ((...args) => {
      const cacheKey = JSON.stringify(args);
      if (cache.has(cacheKey)) { return cache.get(cacheKey) as ReturnType<F>; }
      const result = fn(...args);
      cache.set(cacheKey, result);
      return result;
    }) as F;
  };
\end{lstlisting}
The inner call is plain forwarding, and the inner cast disappears in \domranlang, where the cache can be typed by the first-class projection @Map<string, Range<F>>@ (\Cref{domrange:sec:overview-projection}).
The unsound step is the final @as F@.
The wrapper reproduces only the \emph{call} behavior of @fn@, while a value of type @F@ may carry more, such as properties.
For a function @inc@ with a property @label@, TypeScript accepts @cachedFunction(inc).label.toUpperCase()@, which throws at runtime.
The three declarations out of scope follow the same pattern, reduced here to the wrapper and the final cast:
\begin{lstlisting}
  const withAsyncTOTP =
    <T extends (name: string, ...args: any[]) => Promise<any>>(callAsync: T): T =>
      async function (methodName: string, ...args: unknown[]): Promise<ReturnType<T>> {
        try { return await callAsync(methodName, ...args); }
        catch (error: unknown) { /* retry with a two-factor code */ }
      } as T;

  const use = <T extends (...args: any[]) => any>(fn: T, middleware: Middleware<T>): T =>
    function (this: unknown, ...context: Parameters<T>): ReturnType<T> {
      return middleware(context, pipe(fn.bind(this)));
    } as T;

  function withMiddleware<F extends (...args: any[]) => any>(fn: F) {
    /* a list of middlewares and a runner built from it */
    const run = (...args: Parameters<F>): ReturnType<F> => runner(args);
    return Object.assign(run as F, { use, fn });
  }
\end{lstlisting}
Each cast claims that the wrapper has the full type of the wrapped function, while it reproduces only the call.

Each cast asserts $\TFun{\TDom{F}}{\TRange{F}} <: F$, the half of the $\eta$-equivalence that \domranlang deliberately lacks (\Cref{domrange:sec:casestudy-destructors}); \domranlang types the wrapper at $\TFun{\TDom{F}}{\TRange{F}}$, the type of \textit{eta}.
Programmers cast nonetheless because the honest type is lossy in TypeScript: for an overloaded function, @Parameters@ and @ReturnType@ keep only the last overload, and for a generic function they replace its type variables by @unknown@.
Both lie beyond \domranlang as well.
Overloaded functions are intersections of arrows, whose projections we discuss in \Cref{domrange:sec:limitations-intersection}, and the domain of a quantified type is uninhabited in \domranlang.
All four casts are latent in Rocket.Chat, whose callers pass plain functions.
For \Cref{domrange:sec:overview-practice}, we reduced @cachedFunction@ to the minimal @wrap@ by removing the cache, and added the caller @taggedInc@ to trigger the latent cast.

\paragraph{Reading an argument at the bound's type.}
Three declarations of the settings audit log read an argument out of @Parameters<F>@ (in scope):
\begin{lstlisting}
  <T extends ISetting['value'],
   F extends (_id: ISetting['_id'], value: T, ...args: any[]) => any>(
      fn: F, ...args: Parameters<F>): ReturnType<F> => {
    const [key, value, ...rest] = args;
    const setting = settings.getSetting(key);
    ...                                    // record an audit event for `key`
    return fn(key, value, ...rest);
  }
\end{lstlisting}
TypeScript accepts @settings.getSetting(key)@ because it bounds @Parameters<F>@ from above by the parameters of the bound of @F@, so @key@ is taken to be a string.
That is unsound, because the parameter position is contravariant.
A function that accepts more, such as @(_id: unknown, value: number) => value@, also satisfies the bound, and its caller may then pass a number as @key@.
The following probe needs no @any@ at all; TypeScript accepts it under @--strict@, and it throws at runtime:
\begin{lstlisting}
  function getSetting(id: string): string { return id.toUpperCase(); }
  function audited<F extends (_id: string, value: number) => unknown>(
      fn: F, ...args: Parameters<F>) {
    const [key] = args;
    getSetting(key);                       // accepted: key is read at the bound's type
  }
  const lenient = (_id: unknown, value: number): number => value;
  audited(lenient, 42, 1);                 // accepted; getSetting(42) throws
\end{lstlisting}
The error is one of variance.
In \domranlang, the bound $F <: \TFun{A}{B}$ yields only the lower bound $A <: \TDom{F}$, by \rulename{domrange:sdomintro} and the contravariant congruence \rulename{domrange:sdom}, and no rule bounds $\TDom{F}$ from above, so the read is rejected.
A sound version takes the key as a separate parameter, as the bounded declaration above does.
This answers the question whether \domranlang would catch soundness issues in real code: the hole is latent, since all current callers pass strings, but the signature admits ill-typed calls that TypeScript accepts and \domranlang rejects.

\paragraph{Casts that hide a changed value.}
Three declarations change the arguments or the result and cast the change away, for example (out of scope, asynchronous):
\begin{lstlisting}
  return (...args: [...loginArgs: Parameters<TLoginFunction>, callback?: LoginCallback]) => {
    const loginArgs = args.slice(0, -1) as Parameters<TLoginFunction>;
    const callback = args.slice(-1)[0] as LoginCallback | undefined;
    return login(...loginArgs) ...         // continue with the result and the callback
  };
\end{lstlisting}
The wrapper accepts the login arguments followed by an optional callback, and strips the callback with @slice(0, -1)@.
When the caller omits the callback, @slice@ drops the last \emph{login} argument instead, and the cast to @Parameters<TLoginFunction>@ hides it.
These casts hide changes to argument lists: appending, dropping, or popping an argument needs type-level operations on tuples, which \domranlang does not have and which TypeScript cannot apply to @slice@.
The other two re-wrap a promise and pop a trailing argument, respectively.

\paragraph{Relying on \texttt{any} only.}
One declaration needs @any@ although nothing goes wrong (out of scope, asynchronous):
\begin{lstlisting}
  <T extends (...args: any[]) => any>(endpoint: string, handler: T) =>
    async (...args: Parameters<T>): Promise<ReturnType<T>> => {
      action('fetch')(...args);            // logging, simplified
      await simulateNetworkDelay();
      return handler(...args);
    }
\end{lstlisting}
It forwards its arguments, but inside an asynchronous function.
When @handler@ itself returns a promise, JavaScript flattens the nested promise, so the wrapper returns @Promise<X>@ where the type promises @Promise<Promise<X>>@.
TypeScript accepts the body only because @handler(...args)@ has type @any@.
Flattening is not a subtyping law, so a sound treatment needs type-level computation on the range, such as TypeScript's @Awaited@.
Assuming it, the projections are used soundly.

\paragraph{Other uses outside the calculus.}
With the receiver @this@ read as an additional argument, promises as covariant, and conditional types as type-level case analysis, the remaining out-of-scope declarations use the projections soundly, apart from the casts discussed above.
Two kinds are workarounds for the very projections \domranlang provides.
A conditional type such as @GetResult<T> = T extends (...args: any) => any ? ReturnType<T> : never@ exists only to avoid the bound, a hand-made boundless $\TRange{T}$.
A self-referential bound such as @F extends (span?: Span) => ReturnType<F>@ exists only to give the call @fn()@ a precise type; it is how TypeScript can apply an abstract function without @any@, with a bound that mentions the variable itself, and \rulename{domrange:tappdomrange} gives the same type directly under a bound with range $\TTop$.
Callable objects, functions carrying properties, are intersections of an arrow with a record, which have a single arrow component, so the problem that intersections of arrows pose for the projections does not arise.

\paragraph{Summary.}
In scope, \domranlang types 13 of the 17 declarations as written: the 9 that TypeScript would also accept under a sound bound, and the 4 that apply the abstract function, which TypeScript accepts only through @any@.
No declaration in scope falls in between: each is either typed by \domranlang or in fact ill-typed, and \domranlang cannot type the 4 unsound ones at their declared signatures.
Outside the calculus the picture is similar, and the unsound uses trace back to $\eta$ and to changed arguments or results, not to the projections themselves.
The case study covers one repository, and its classification is the AI agent's judgement, which the authors did not repeat independently. 
Again, we view the study as an introduction to our motivation, instead of a contribution for a solid report.  
\newcommand{\systep}[2]{\ensuremath{#1 \longrightarrow #2}}
\newcommand{\syshape}[1]{\ensuremath{\lvert #1 \rvert}}
\newcommand{\syarr}[2]{\ensuremath{#1 \Rightarrow #2}}
\newcommand{\syall}[2]{\ensuremath{\forall(#1,\,#2)}}
\newcommand{\systuck}{\ensuremath{\ast}}
\newcommand{\syle}[3]{\ensuremath{#1 \vdash #2 \sqsubseteq #3}}
\newcommand{\sydom}{\mathsf{dom}}
\newcommand{\syrng}{\mathsf{rng}}
\newcommand{\sycanon}{\mathsf{canon}}
\newcommand{\synf}{\mathsf{nf}}
\newcommand{\syndom}{\mathsf{ndom}}
\newcommand{\synrng}{\mathsf{nrng}}

\section{Syntactic Type Soundness} \label{domrange:appendix:syntactic}

Besides the semantic proof of \Cref{domrange:sec:LR}, we have mechanized a syntactic proof of type safety for \domranlang by progress and preservation, over a small-step call-by-value semantics.
It covers the calculus of \Cref{domrange:fig:formal-syntax,domrange:fig:formal-subtyping,domrange:fig:formal-typing} with the full variant of \Fsub.
The proof follows the textbook proof for \Fsub \cite{tapl_book} except for the inversion lemmas.
The textbook proof inverts subtyping derivations, which requires transitivity elimination, and we have not proved transitivity elimination for the extended subtyping relation (\Cref{domrange:sec:boundless-whysound}).
We obtain the inversions from the projection types themselves and from an approximation of subtyping by shapes instead.
The proof also establishes the witness property of \Cref{domrange:sec:boundless-whysound} as a theorem about types.

\subsection{Formulation} \label{domrange:appendix:syntactic-formulation}

\begin{figure}[t]\small
\begin{mdframed}
\judgement{Values}{\BOX{\domranlang}}\small
  \[\begin{array}{l@{\qquad}l@{\qquad}l@{\qquad}l}
    v & ::= & c \mid \lambda x{:}T.\,t \mid \TLam{X <: T}{t} & \text{Values}
  \end{array}\]

\judgement{Small-Step Reduction}{\BOX{\systep{t}{t}}}
\begin{tabular}{b{.48\linewidth}@{}b{.5\linewidth}}
  \infrule[\textsc{ss-beta}]{
  }{
    \systep{(\lambda x{:}T.\,t)~v}{\substty[x]{v}{t}}
  }
  &
  \infrule[\textsc{ss-tbeta}]{
  }{
    \systep{\TApp{(\TLam{X <: U}{t})}{T}}{\substty[X]{T}{t}}
  }
  \\[0.5em]
  \infrule[\textsc{ss-app1}]{
    \systep{t_1}{t_1'}
  }{
    \systep{t_1~t_2}{t_1'~t_2}
  }
  &
  \infrule[\textsc{ss-app2}]{
    \systep{t}{t'}
  }{
    \systep{v~t}{v~t'}
  }
  \\[0.5em]
  \infrule[\textsc{ss-tapp}]{
    \systep{t}{t'}
  }{
    \systep{\TApp{t}{T}}{\TApp{t'}{T}}
  }
  &
  \\[0.5em]
\end{tabular}
\caption{The values and the small-step call-by-value reduction of the syntactic proof.} \label{domrange:fig:syn-step}
\end{mdframed}
\vspace{-2ex}
\end{figure}
 
The subtyping and typing rules are those of \Cref{domrange:fig:formal-subtyping,domrange:fig:formal-typing}, with binders represented in the locally nameless style.
The only additional premises are closedness conditions, and a rule carries one only where it introduces a type that no other premise constrains.
These are \rulename{domrange:srefl}, \rulename{domrange:stop}, \rulename{domrange:sbot}, and the four introduction and elimination rules of the projections for the types they mention, \rulename{domrange:sall} for the two quantifier bodies, \rulename{domrange:tabs} for the parameter type and the body, and \rulename{domrange:ttabs} for the bound, the body type, and the body.
The standard rule \rulename{domrange:tapp} and type instantiation below the bound are derived as stated in \Cref{domrange:sec:formal}.

\Cref{domrange:fig:syn-step} defines the values and the small-step reduction.
Unlike the big-step semantics of \Cref{domrange:sec:LR}, values are terms and carry no environment, and reduction substitutes the argument value into the body.
Reduction is call-by-value and left to right, so a function is applied only to a value, as the witness property requires (\Cref{domrange:sec:discussion-cbv}).

The standard weakening, narrowing, and substitution lemmas hold for subtyping and typing.
Narrowing is used only once, for the transitivity of the shape order below.

\subsection{Inversion without Transitivity Elimination} \label{domrange:appendix:syntactic-inversion}

With \rulename{domrange:strans} as a rule, a derivation of $\stp{\G}{\TFun{S_1}{S_2}}{T}$ may end with transitivity through an arbitrary middle type, so it cannot be inverted by induction.
Type safety needs two such inversions, one for arrows and one for quantifiers, together with canonical forms.

\paragraph{Arrow inversion.}
For arrows, the projections give the inversion directly.
Applying a congruence rule to the whole subtyping and then an elimination or introduction rule gives the two facts that a $\beta$-redex of \rulename{domrange:tappdomrange} needs, without looking into the derivation.

\begin{lemma}[Arrow inversion by projection] \label{domrange:thm:syn-arrow}
  If $\,\stp{\G}{\TFun{A}{B}}{T}$, then $\stp{\G}{\TDom{T}}{A}$ and $\stp{\G}{B}{\TRange{T}}$.
\end{lemma}
\begin{proof}
  For the domain, \rulename{domrange:sdom} gives $\stp{\G}{\TDom{T}}{\TDom{\TFun{A}{B}}}$, and \rulename{domrange:sdomelim} gives $\stp{\G}{\TDom{\TFun{A}{B}}}{A}$.
  For the range, \rulename{domrange:srangeintro} gives $\stp{\G}{B}{\TRange{\TFun{A}{B}}}$, and \rulename{domrange:srange} gives $\stp{\G}{\TRange{\TFun{A}{B}}}{\TRange{T}}$.
  Both conclusions follow by \rulename{domrange:strans}.
\end{proof}

The textbook inversion $\stp{\G}{T_1}{S_1}$ and $\stp{\G}{S_2}{T_2}$ from $\stp{\G}{\TFun{S_1}{S_2}}{\TFun{T_1}{T_2}}$ follows as a corollary, with $T = \TFun{T_1}{T_2}$ and the projections eliminated once more.

\paragraph{Shape approximation.}
\begin{figure}[t]\small
\begin{mdframed}
\judgement{Shapes}{}\small
  \[\begin{array}{l@{\qquad}l@{\qquad}l@{\qquad}l}
    \sigma & ::= & \bot \mid \top \mid \mathsf{B} \mid \syarr{\sigma}{\sigma} \mid \syall{U}{T} \mid \systuck & \text{Shapes}
  \end{array}\]

\judgement{Shape of a Type}{\BOX{\syshape{T}}}\small
  \[\begin{array}{l@{\;=\;}l@{\qquad\qquad}l@{\;=\;}l}
    \syshape{B} & \mathsf{B} & \syshape{\TFun{T_1}{T_2}} & \syarr{\syshape{T_1}}{\syshape{T_2}} \\
    \syshape{\TTop} & \top & \syshape{\TAll{U}{T}} & \syall{U}{T} \\
    \syshape{\TBot} & \bot & \syshape{\TDom{T}} & \sydom(\syshape{T}) \\
    \syshape{X} & \bot & \syshape{\TRange{T}} & \syrng(\syshape{T})
  \end{array}\]

\judgement{Projections of Shapes}{}\small
  \[\begin{array}{l@{\;=\;}l@{\qquad}l@{\;=\;}l@{\qquad}l@{\;=\;}l@{\qquad}l@{\;=\;}l}
    \sydom(\syarr{\sigma_1}{\sigma_2}) & \sigma_1 & \sydom(\top) & \bot & \sydom(\bot) & \top & \sydom(\sigma) & \systuck \text{ otherwise} \\
    \syrng(\syarr{\sigma_1}{\sigma_2}) & \sigma_2 & \syrng(\top) & \top & \syrng(\bot) & \bot & \syrng(\sigma) & \systuck \text{ otherwise}
  \end{array}\]

\judgement{Order on Shapes}{\BOX{\syle{\G}{\sigma}{\sigma}}}
\begin{tabular}{b{.24\linewidth}@{}b{.24\linewidth}@{}b{.24\linewidth}@{}b{.26\linewidth}}
  \infrule[\textsc{q-top}]{
  }{
    \syle{\G}{\sigma}{\top}
  }
  &
  \infrule[\textsc{q-bot}]{
  }{
    \syle{\G}{\bot}{\sigma}
  }
  &
  \infrule[\textsc{q-base}]{
  }{
    \syle{\G}{\mathsf{B}}{\mathsf{B}}
  }
  &
  \infrule[\textsc{q-stuck}]{
  }{
    \syle{\G}{\systuck}{\systuck}
  }
\end{tabular}
\begin{tabular}{b{.45\linewidth}@{}b{.53\linewidth}}
  \infrule[\textsc{q-arr}]{
    \syle{\G}{\sigma_2}{\sigma_1} \qquad \syle{\G}{\sigma_1'}{\sigma_2'}
  }{
    \syle{\G}{\syarr{\sigma_1}{\sigma_1'}}{\syarr{\sigma_2}{\sigma_2'}}
  }
  &
  \infrule[\textsc{q-all}]{
    \stp{\G}{T_1}{S_1} \qquad \stp{\ctxext{\G}{X <: T_1}}{S_2}{T_2}
  }{
    \syle{\G}{\syall{S_1}{S_2}}{\syall{T_1}{T_2}}
  }
\end{tabular}
\vspace{1ex}

\judgement{Canonical Forms}{\BOX{\sycanon(v, \sigma)}}\small
  \[\begin{array}{l@{\;\;\text{iff}\;\;}l}
    \sycanon(v, \syarr{\sigma_1}{\sigma_2}) & v = \lambda x{:}T.\,t \text{ for some } T, t \\
    \sycanon(v, \syall{U}{T}) & v = \TLam{X <: U}{t} \text{ for some } U, t \\
    \sycanon(v, \bot),\ \sycanon(v, \systuck) & \text{never} \\
    \sycanon(v, \sigma) & \text{always, otherwise}
  \end{array}\]
\caption{The approximation of subtyping by shapes, used for quantifier inversion and canonical forms.} \label{domrange:fig:syn-shape}
\end{mdframed}
\vspace{-2ex}
\end{figure}
 For quantifiers and canonical forms, we map types to the shapes of \Cref{domrange:fig:syn-shape} and prove that subtyping is sound for the order on shapes.
The order has no transitivity rule, so it can be inverted by cases, and inverting the order replaces inverting subtyping.
Only soundness is needed, so the shapes may identify more types than subtyping does.

The shape of an arrow type keeps the shapes of both components, so that $\sydom$ and $\syrng$ read them back, and the four introduction and elimination rules of the projections hold by computation.
On $\top$ and $\bot$, the projections of shapes follow the variance, \eg, $\sydom(\top) = \bot$ since the domain is contravariant.
On any other shape, a projection is the stuck shape $\systuck$, which is below no shape but itself and has no values.
This keeps $\TDom{\TBool}$ uninhabited in the shape order.
The shape of a quantifier keeps its bound and body as types, ordered as \rulename{domrange:sall} compares them.
A type variable has shape $\bot$, which identifies it with the empty type.
This is sound because $\bot$ is below every bound.

\begin{lemma}[Soundness of shapes] \label{domrange:thm:syn-shape}
  If $\,\stp{\G}{S}{T}$, then $\syle{\G}{\syshape{S}}{\syshape{T}}$.
\end{lemma}
\begin{proof}
  By induction on the subtyping derivation.
  The case \rulename{domrange:strans} uses the transitivity of the shape order, whose quantifier case needs narrowing followed by \rulename{domrange:strans} on the bodies.
  The congruence rules of the projections hold because $\sydom$ is antitone and $\syrng$ is monotone, and the introduction and elimination rules hold by reflexivity, since $\syshape{\TDom{\TFun{A}{B}}} = \syshape{A}$ and $\syshape{\TRange{\TFun{A}{B}}} = \syshape{B}$.
\end{proof}

\begin{corollary}[Quantifier inversion] \label{domrange:thm:syn-all}
  If $\,\stp{\G}{\TAll{S_1}{S_2}}{\TAll{T_1}{T_2}}$, then $\stp{\G}{T_1}{S_1}$ and $\stp{\ctxext{\G}{X <: T_1}}{S_2}{T_2}$.
\end{corollary}
\begin{proof}
  By Lemma~\ref{domrange:thm:syn-shape} and inverting \rulefmt{q-all}.
\end{proof}

\begin{lemma}[Canonical forms] \label{domrange:thm:syn-canon}
  If $\,\hastype{\varnothing}{v}{T}$ for a value $v$, then $\sycanon(v, \syshape{T})$.
\end{lemma}
\begin{proof}
  By induction on the typing derivation, using Lemma~\ref{domrange:thm:syn-shape} at \rulename{domrange:tsub}, since $\sycanon$ is preserved along the shape order.
\end{proof}

The witness property holds for shapes.
A value of type $\TDom{F}$ rules out every shape of $F$ whose domain is $\bot$ or $\systuck$, which leaves an arrow shape or $\bot$.
A value of type $F$ then rules out $\bot$.

\begin{lemma}[Witness property on shapes] \label{domrange:thm:syn-witness-shape}
  If $\,\hastype{\varnothing}{v}{\TDom{F}}$ for a value $v$, then $\syshape{F} = \bot$ or $\syshape{F} = \syarr{\sigma_1}{\sigma_2}$ for some $\sigma_1$ and $\sigma_2$.
  If moreover $\hastype{\varnothing}{f}{F}$ for a value $f$, then $f = \lambda x{:}T.\,t$ for some $T$ and $t$.
\end{lemma}
\begin{proof}
  By Lemma~\ref{domrange:thm:syn-canon} on $\sydom(\syshape{F})$ and then on $\syshape{F}$.
\end{proof}

\subsection{Type Safety} \label{domrange:appendix:syntactic-safety}

\begin{theorem}[Preservation] \label{domrange:thm:syn-preservation}
  If $\,\hastype{\varnothing}{t}{T}$ and $\systep{t}{t'}$, then $\hastype{\varnothing}{t'}{T}$.
\end{theorem}
\begin{proof}
  By induction on the typing derivation.
  In the $\beta$-case of \rulename{domrange:tappdomrange}, the function is $\lambda x{:}T.\,t_0$ of type $F$, and the generation lemma gives $\stp{\varnothing}{\TFun{A}{B}}{F}$ with the body typed at $B$ under $x : A$.
  Lemma~\ref{domrange:thm:syn-arrow} gives $\stp{\varnothing}{\TDom{F}}{A}$ for the argument and $\stp{\varnothing}{B}{\TRange{F}}$ for the result, and the substitution lemma concludes.
  In the $\beta$-case of \rulename{domrange:ttapp}, the generation lemma for type abstractions, Corollary~\ref{domrange:thm:syn-all}, and type substitution conclude similarly.
\end{proof}

\begin{theorem}[Progress] \label{domrange:thm:syn-progress}
  If $\,\hastype{\varnothing}{t}{T}$, then $t$ is a value or $\systep{t}{t'}$ for some $t'$.
\end{theorem}
\begin{proof}
  By induction on the typing derivation.
  In the case of \rulename{domrange:tappdomrange} with $f : F$ and an argument of type $\TDom{F}$ that are both values, Lemma~\ref{domrange:thm:syn-witness-shape} makes $f$ a $\lambda$-abstraction, so \rulefmt{ss-beta} applies.
  The case of \rulename{domrange:ttapp} uses Lemma~\ref{domrange:thm:syn-canon}.
\end{proof}

\begin{corollary}[Type safety] \label{domrange:thm:syn-safety}
  If $\,\hastype{\varnothing}{t}{T}$ and $t \longrightarrow^{*} t'$, then $t'$ is a value or $\systep{t'}{t''}$ for some $t''$.
\end{corollary}

In this proof, the arrow obligation of \rulename{domrange:tappdomrange} is discharged in the progress case, where both the function and the argument are values.

\subsection{The Witness Property on Types} \label{domrange:appendix:syntactic-witness}

\begin{figure}[t]\small
\begin{mdframed}
\judgement{Normal Forms}{\BOX{\synf(T)}}\small
  \[\begin{array}{l@{\;=\;}l}
    \synf(\TFun{T_1}{T_2}) & \TFun{\synf(T_1)}{\synf(T_2)} \\
    \synf(\TDom{T}) & \syndom(\synf(T)) \\
    \synf(\TRange{T}) & \synrng(\synf(T)) \\
    \synf(T) & T \text{ otherwise}
  \end{array}\]
  \vspace{1ex}
  \[\begin{array}{l@{\;=\;}l@{\qquad\qquad}l@{\;=\;}l}
    \syndom(\TFun{A}{B}) & A & \synrng(\TFun{A}{B}) & B \\
    \syndom(\TTop) & \TBot & \synrng(\TTop) & \TTop \\
    \syndom(\TBot) & \TTop & \synrng(\TBot) & \TBot \\
    \syndom(N) & \TDom{N} \text{ otherwise} & \synrng(N) & \TRange{N} \text{ otherwise}
  \end{array}\]

\judgement{Stuck and Normal Types}{}\small
  \[\begin{array}{l@{\;\;\text{iff}\;\;}l}
    N \text{ is stuck} & N \text{ is not an arrow type, } \TTop, \text{ or } \TBot \\
    \TFun{T_1}{T_2} \text{ is normal} & T_1 \text{ and } T_2 \text{ are normal} \\
    \TDom{N} \text{ and } \TRange{N} \text{ are normal} & N \text{ is normal and stuck} \\
    T \text{ is normal} & \text{always, otherwise}
  \end{array}\]
\caption{Normal forms of types, used to state the witness property about types.} \label{domrange:fig:syn-nf}
\end{mdframed}
\vspace{-2ex}
\end{figure}
 
\Cref{domrange:sec:boundless-whysound} states the witness property about types: a value of type $\TDom{F}$ exists only when $F$ is a subtype of an arrow.
Lemma~\ref{domrange:thm:syn-witness-shape} states it about shapes, which is all that type safety needs.
To return to types, we use the normal forms of \Cref{domrange:fig:syn-nf}, which compute the projections away where the subtyping rules allow it.
Besides the four introduction and elimination rules, normalization uses the equivalences $\TDom{\TTop} \equiv \TBot$, $\TDom{\TBot} \equiv \TTop$, $\TRange{\TTop} \equiv \TTop$, and $\TRange{\TBot} \equiv \TBot$, each derivable through an arrow type, \eg, $\TDom{\TTop} <: \TDom{(\TFun{\TBot}{\TTop})} <: \TBot$.
We write $T_1 \equiv T_2$ for subtyping in both directions.

\begin{lemma}[Normal forms] \label{domrange:thm:syn-nf}
  For every type $T$ closed under $\G$, $\G \vdash T \equiv \synf(T)$, and $\synf(T)$ is normal.
  Moreover, if $N$ is closed, normal, and stuck, then $\sydom(\syshape{N}) = \syrng(\syshape{N}) = \systuck$.
\end{lemma}
\begin{proof}
  The first two claims by induction on $T$, using the congruence rules and the equivalences above.
  The last by induction on $N$: a closed stuck type is a base type, a quantifier, or a projection of a stuck type.
\end{proof}

\begin{theorem}[Witness property] \label{domrange:thm:syn-witness}
  If $\,\hastype{\varnothing}{v}{\TDom{F}}$ for a value $v$, then $F \equiv \TBot$ or $F \equiv \TFun{A}{B}$ for some $A$ and $B$, and hence $F$ is a subtype of an arrow type.
  If moreover $\hastype{\varnothing}{f}{F}$ for a value $f$, then $F \equiv \TFun{A}{B}$.
\end{theorem}
\begin{proof}
  By Lemma~\ref{domrange:thm:syn-nf}, $v$ also has type $\TDom{\synf(F)}$, so Lemma~\ref{domrange:thm:syn-witness-shape} gives $\syshape{\synf(F)} = \bot$ or an arrow shape.
  Among closed normal types, only $\TBot$ and arrow types have such shapes: base types, $\TTop$, and quantifiers have other shapes, and a normal projection has the stuck shape by Lemma~\ref{domrange:thm:syn-nf}.
  For the second claim, a value of type $F \equiv \TBot$ would have shape $\bot$, which Lemma~\ref{domrange:thm:syn-canon} rules out.
\end{proof}

The theorem gives the conclusion of the informal argument in \Cref{domrange:sec:boundless-whysound} without inverting any subtyping derivation.
The second claim is the case where the call can happen, in which the application is a standard one with $\TDom{F} \equiv A$ and $\TRange{F} \equiv B$.

\end{document}